\documentclass[12pt,reqno]{amsart}
\usepackage[utf8]{inputenc}
\usepackage[T1]{fontenc}
\usepackage[margin=1in,dvips]{geometry}
\usepackage{amsmath}
\usepackage{amsfonts}
\usepackage{amssymb}
\usepackage{amsthm}
\usepackage{graphicx}
\usepackage{hyperref}
\usepackage{slashed}
\usepackage{tikz}
\usepackage{caption}
\usepackage{subcaption}
\usepackage{float}
\usepackage{xcolor}
\usepackage{lipsum}
\usepackage{url}
\usetikzlibrary{decorations.pathmorphing}
\usetikzlibrary{snakes}
\usepackage[final]{microtype}

\newtheorem{thm}{Theorem}
\newtheorem{cor}{Corollary}[section]

\theoremstyle{definition}
\newtheorem{prop}{Proposition}[section]

\theoremstyle{remark}
\newtheorem{rem}{Remark}[section]

 \numberwithin{equation}{section}
 \def\ub {\underline{u}}

\def\th {\theta}

\def\Hb {\underline{H}}
\def\chib {\underline{\chi}}

\def\etab {\underline{\eta}}
\def\bb {\underline{\beta}}
\def\ab {\underline{\alpha}}

\def\a {\alpha}
\def\b {\beta}
\def\nab {\nabla}
\def\ep {\epsilon}
\def\om {\omega}
\def\omb {\underline{\omega}}
\def\Om {\Omega}
\def\p {\partial}

\def\f {\frac}

\title[Weak Null Singularity for the Einstein--Euler System]{Weak Null Singularity\\ for the Einstein--Euler System}
   \author{Yuefeng Song}
   \address{Department of Mathematics\\ Stanford University\\ Stanford CA 94305-2125\\ USA}
   \email{songyf@stanford.edu}

\begin{document}
   \begin{abstract}
     We study the behavior of a self-gravitating perfect relativistic fluid satisfying the Einstein--Euler system in the presence of a weak null terminal spacetime singularity. This type of singularities is expected in the interior of generic dynamical black holes. In the vacuum case, weak null singularities have been constructed locally by Luk, where the metrics extend continuously to the singularities while the Christoffel symbols fail to be square integrable in any neighborhood of any point on the singular boundaries. We prove that this type of singularities persists in the presence of a self-gravitating fluid. Moreover, using the fact that the speed of sound is strictly less than the speed of light, we prove that the fluid variables also extend continuously to the singularity.
   \end{abstract}
   \maketitle
   \section{Introduction}
   \subsection{Motivation}
   Recent works suggest that generic dynamical \emph{vacuum} black holes terminate with an essential weak null singularity. This type of singularities were first rigorously constructed in \cite{luk2017weak} (see also \cite{Ori1995HowGA} for analytic examples), and have the property that in a suitable coordinate system, the metric coefficients extend continuously up to the null boundary while the Christoffel symbols blow up. Indeed, in \cite{dafermos2017interiordynamicalvacuumblack}, Dafermos and Luk showed that the metric in the interior of any black hole settling down to a subextremal rotating Kerr black hole suitably fast can be extended continuously across a small piece of the Cauchy horizon $\mathcal{CH}^+$ (see Figure \ref{Fig1}). When combined with the intuition from linear and spherical symmetric problems, \cite{LocalizedBigBang,a965bffe-de90-3a7a-a652-4c0d2147a13b,
   https://doi.org/10.1002/cpa.20071,Dafermos_2016,
   gautam2024latetimetailsmassinflation,
   Gurriaran_2025,gurriaran2026nonlinearinstabilitykerrcauchy,HISCOCK1981110,
   Luk:2015qja,luk2019strong,Luk2019,LOSQ,LUK20161948,luk2026formationweaknullsingularity,
   MR4657222,PhysRevLett.63.1663,PhysRevD.41.1796,
   SBierSKIJANPROOF,Moortel2017StabilityAI,
   93760b6492084bb69a1843fdee21adfc,VandeMoortel:2019ike,vandemoortel2025strongcosmiccensorshipconjecture} this suggests that the Cauchy horizon is a weak null singularity for a generic subclass of data, at least in a small neighborhood of Kerr.
   

   

   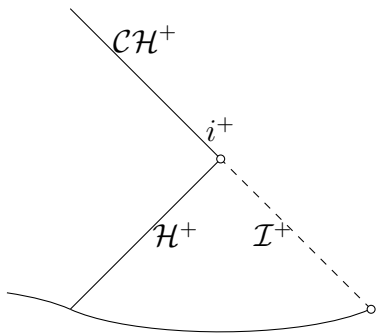
\begin{figure}[h]
   \centering
        \begin{tikzpicture}
          \draw[dashed](4,0)--(2,2);
          \draw(0,4)--(2,2);
          \draw(2,2)--(0,0);
          \filldraw[fill=white,draw=white] (0,0) arc (210:330:2.3cm and 0.6cm)--(2,-0.3)--(0,-0.3)--(0,0);
          \draw(0,0) arc (210:330:2.3cm and 0.6cm);
          \draw(0,0) arc (30:60:2.3cm and 0.6cm);
          \draw[fill=white](2,2) circle (1.5pt);
          \draw[fill=white](4,0) circle (1.5pt);
          \draw(1,3.6)node{$\mathcal{CH}^+$};
          \draw(1.4,1)node{$\mathcal{H}^+$};
          \draw(2.7,1)node{$\mathcal{I}^+$};
          \draw(2,2.4)node{$i^+$};
        \end{tikzpicture}
        \caption{A generic dynamical vacuum black hole}\label{Fig1}
       \end{figure}
       
       In view of this picture, it is natural to investigate the stability of the vacuum essential weak null singularity in the presence of matter fields. In cases where the matter fields propagate at the speed of light (e.g., Maxwell, (charged)-scalar field, null dust, etc.), various works in spherical symmetry showed that the singularity structure persists and is very similar to the vacuum case. In fact, the ideas of Luk can be extended to the type of matter fields described above, further implying that the singularity structure persists even outside symmetry (see \cite[Footnote 8]{luk2017weak}). In particular, in this case, the matter field blows up in a manner such that the Ricci curvature blows up along a parallelly propagated frame towards the weak null singularity. (See \cite[Lemma 11.7]{luk2019strong}.)
       

   
   The purpose of this paper is to understand the behavior of a perfect relativistic self-gravitating fluid at a weak null singularity, i.e., whether a Penrose diagram as in Figure \ref{Fig3} occurs. In the physical case, the speed of sound is strictly slower than the speed of light, in contrast to all the matter models that have been previously studied.
   In this paper, we focus on the \emph{local} problem near the singularity (see Figure \ref{Fig4}). One expects additional difficulties near the boundary of the fluid (see \cite{disconzi2023recent,hao2024wellposednessfreeboundarybarotropic} for discussions about the free boundary problem).
   
   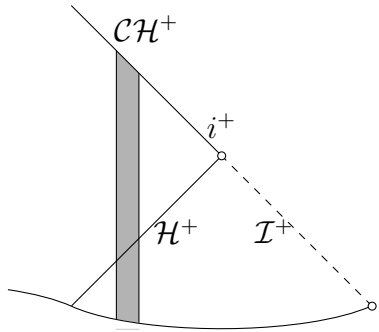
\begin{figure}[h]
   \centering
        \begin{tikzpicture}
          \draw[fill=black!30](0.6,3.4)--(0.6,-0.3)--(0.9,-0.3)--(0.9,3.1);
          \draw[dashed](4,0)--(2,2);
          \draw(0,4)--(2,2);
          \draw(2,2)--(0,0);
          \filldraw[fill=white,draw=white] (0,0) arc (210:330:2.3cm and 0.6cm)--(2,-0.3)--(0,-0.3)--(0,0);
          \draw(0,0) arc (210:330:2.3cm and 0.6cm);
          \draw(0,0) arc (30:60:2.3cm and 0.6cm);
          \draw[fill=white](2,2) circle (1.5pt);
          \draw[fill=white](4,0) circle (1.5pt);
          \draw(1,3.7)node{$\mathcal{CH}^+$};
          \draw(1.4,1)node{$\mathcal{H}^+$};
          \draw(2.7,1)node{$\mathcal{I}^+$};
          \draw(2,2.4)node{$i^+$};
        \end{tikzpicture}
        \caption{Fluids evolving until a weak null singularity}\label{Fig3}
       \end{figure}
   
   We prove that the weak null singularity is preserved under introduction of a fluid. In particular, no shock is formed in the fluid and the fluid variables are bounded up to the weak null singularity. Moreover, we establish that the fluid variables themselves, and thus the Ricci curvature of the spacetime metric, can be extended continuously up to the singularity. This is in stark contrast with the case mentioned above and is ultimately a consequence of the fact that the speed of sound is strictly slower than the speed of light. See the discussions in Section \ref{Section1.3.2}.
   
   \begin{figure}[h]
   \centering
        \begin{tikzpicture}
          \draw[dashed](4,0)--(2,2);
          \draw(0,4)--(2,2);
          \draw(2,2)--(0,0);
          \draw(0,0) arc (210:330:2.3cm and 0.6cm);
          \draw(0,0) arc (30:60:2.3cm and 0.6cm);
          \draw[fill=white](2,2) circle (1.5pt);
          \draw[fill=white](4,0) circle (1.5pt);
          \draw(1,3.7)node{$\mathcal{CH}^+$};
          \draw(1.4,1)node{$\mathcal{H}^+$};
          \draw(2.7,1)node{$\mathcal{I}^+$};
          \draw(2,2.4)node{$i^+$};
          \draw[fill=black!30](0.8,3.2)--(0.6,3)--(0.4,3.2)--(0.6,3.4);
        \end{tikzpicture}
        \caption{Domain of existence}\label{Fig4}
       \end{figure}
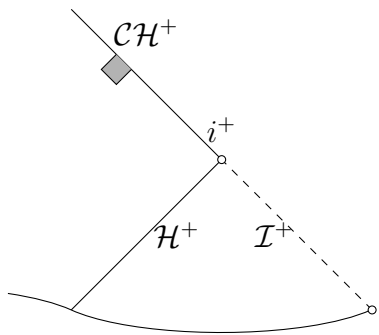
       
   \subsection{Description of the Main Results}\label{SectDescriptionofMainResults}
   The Einstein--Euler system can be formulated for a quadruple $(\mathcal{M},g,v,\tau)$ comprised of a $4$-dimensional time-oriented Lorentzian manifold $(\mathcal{M},g)$ of signature $(-,+,+,+)$, the fluid four-velocity $v^\nu$, and the proper energy density $\tau$ of the fluid,
   solving the Einstein--Euler equations
   \begin{align}
     \left\{
     \begin{aligned}
       &\mathrm{Ric}_{\mu\nu}-\frac{1}{2}Rg_{\mu\nu}=T_{\mu\nu}, \\
       &D_\alpha T^{\alpha\mu}=0,
     \end{aligned}
     \right.\label{Eq1}
   \end{align}
   where $T_{\mu\nu}$ is the energy-momentum tensor of a perfect fluid,
   \begin{align}
     &T_{\mu\nu}=(\tau+p)v_\mu v_\nu +pg_{\mu\nu},\label{Eq2}
   \end{align}
   $v^\nu$ is normalized by
   \begin{align}\label{Eq2.1.1}
     &g_{\mu\nu}v^\mu v^\nu=-1,
   \end{align}
   and $p$ relates to $\tau$ via the equation of state
   \begin{align}\label{Eq2.2.1}
     p=p(\tau),\ 0<p'<1.
   \end{align}
   The acoustical metric is given by
   \begin{align*}
     &G^{\alpha\beta}=p'g^{\alpha\beta}+(p'-1)v^\alpha v^\beta,
   \end{align*}
   see \cite{disconzi2023recent}.
   This defines a Lorentzian metric in spacetime whose inverse is given by $G^{\alpha\beta}$. 
   The characteristics for the Euler equations coincide with the sound cones $\{(x,\xi): G^{\alpha\beta} \xi_\alpha\xi_\beta=0\}$ of the acoustical metric $G$.
   Under the assumption that $p'<1$, the sound cones 
   lie strictly inside the light cones.

       The precise setup is as follows. In order to avoid considering issues regarding boundaries or topology, we consider a domain $D_{u_*,\underline{u}_*}$ diffeomorphic to $[0,u_*]\times[0,\underline{u}_*]\times\mathbb{T}^2$ (as opposed to $[0,u_*]\times[0,\underline{u}_*]\times\mathbb{S}^2$ in \cite{luk2017weak}) by local coordinate functions $(u,\underline{u},\theta^A)$ as a Lorentzian $4$-manifold with corners. One should think of $\{u=u_*\}$ as the singular boundary. Since the problem, due to the finite speed of propagation, is local in nature, we do not expect the assumptions on global topology to be essential.


   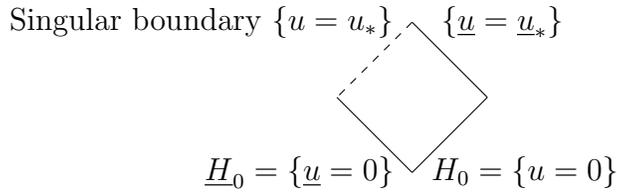
\begin{figure}[h]
   \centering
        \begin{tikzpicture}
          \draw(0,1)--(-1,0);
          \draw[dashed](0,1)--(1,0);
          \draw(0,-1)--(1,0);
          \draw(0,-1)--(-1,0);
          \draw(-1.8,-1)node{$H_0=\{u=0\}$};
          \draw(-1.5,1)node{$\{\underline{u}=\underline{u}_*\}$};
          \draw(1.8,-1)node{$\underline{H}_0=\{\underline{u}=0\}$};
          \draw(3,1)node{$\text{Singular boundary}\ \{u=u_*\}$};
        \end{tikzpicture}
        \caption{Region of existence in Theorem \ref{MainThm}}
        \label{Fig5}
       \end{figure}

   Our setup is the characteristic initial value problem in $D_{u_*,\underline{u}_*}$ with initial data given on two null hypersurfaces $H_0$ and $\underline{H}_0$ intersecting at $S_{0,0}$ diffeomorphic to $\mathbb{T}^2$. We will follow the general notations in \cite{Christodoulou:2008nj}, see also \cite{DafermosBlackHole}, \cite{417492a798bd4c2abf778c47bb3969dc},
   \cite{https://doi.org/10.1002/cpa.21531}, \cite{Yu2011DynamicalFO}.

   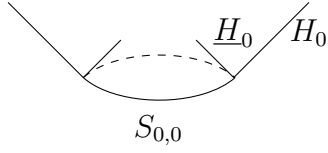
\begin{figure}[h]
     \centering
     \begin{tikzpicture}
       \draw(-1,-0.3) arc (210:330:1.165cm and 0.6cm);
       \draw[dashed](-1,-0.3) arc (150:30:1.165cm and 0.6cm);
       \draw(-1,-0.3)--(-2,0.7);
       \draw(1,-0.3)--(2,0.7);
       \draw(-1,-0.3)--(-0.5,0.2);
       \draw(1,-0.3)--(0.5,0.2);
       \draw(0,-1)node{$S_{0,0}=\mathbb{T}^2$};
       \draw(2,0.3)node{$H_0$};
       \draw(1,0.3)node{$\underline{H}_0$};
     \end{tikzpicture}
     \caption{The Basic Setup}
     \label{Fig6}
   \end{figure}
   
   We introduce a null frame $\{e_1,e_2,e_3,e_4\}$ adpated to a double null foliation $(u,\underline{u})$. Denote the constant $u$ hypersurfaces by $H_u$, the constant $\underline{u}$ hypersurfaces by $\underline{H}_{\underline{u}}$ and their intersections by $S_{u,\underline{u}}=H_u\cap\underline{H}_{\underline{u}}$.
   
   For $A\in\{1,2\}$, we let $\gamma_{AB}$ be a family of Riemannian metric on $\mathbb{T}^2$ parameterized by $(u,\underline{u})\in [0,u_*]\times[0,\underline{u}_*]$, with respect to local coordinates $\theta^A$. Correspondingly, we let $b^A$ be a vector field on $\mathbb{T}^2$, and let $\Omega$ be a positive function on $\mathbb{T}^2$. Consider a metric of the form
   \begin{align*}
     g=-2\Omega^2(du\otimes d\underline{u}+d\underline{u}\otimes du)
     +\gamma_{AB}(d\theta^A-b^Adu)\otimes(d\theta^B-b^Bdu).
   \end{align*}
   We further require that
   \begin{align*}
     \Omega|_{H_0\cup \underline{H}_0} &=1,\
     b^A|_{\underline{H}_0}=0.
   \end{align*}
   Define the vector fields
   \begin{align*}
     L'{}^\mu=-2
     g^{\mu\nu} \p_\nu  u,
     &\ \underline{L}'{}^\mu=-2
     g^{\mu\nu}\p_\nu \underline{u}.
   \end{align*}
   These are null and geodesic vector fields, where in coordinates
   \begin{align*}
     &L'=\Om^{-2}\p_{\ub},\
     \underline{L}'=\Om^{-2}(\p_u+b^A\p_A).
   \end{align*}
   Define
   \begin{align*}
     e_3=\Omega^2 \underline{L}',&\ e_4=L'
   \end{align*}
   to be the normalized pair 
   and
   \begin{align*}
     \underline{L}=\Omega^2\underline{L}', &\ L=\Omega^2 L'
   \end{align*}
   to be the so-called equivariant vector fields. Note that $u,\underline{u}$ are null variables, i.e.
   \begin{align*}
     g^{-1}(du,du)=0=g^{-1}(d\underline{u},d\underline{u}).
   \end{align*}
   Then $\theta^A$ is transported by $\underline{L}$ on $\underline{H}_0$ and $L$ everywhere, i.e.
   \begin{align*}
     L\theta^A=0,\
     \underline{L}\theta^A|_{\underline{H}_0}=0.
   \end{align*}
   Note that we have chosen $e_4$ to be geodesic. With this specific choice of frames, we will observe that $\omega=0$, as will be defined immediately below. This will be important later for the estimates, see Section \ref{SubsectionModifiedEstimates}. This is a different normalization as in \cite{luk2017weak}, though the arguments in \cite{luk2017weak} would have also worked with this different normalization.
   
   Denote the Weyl curvature tensor under the null frame $\{e_1,e_2,e_3,e_4\}$ by
   \begin{align*}
     & \alpha_{AB}=W(e_A,e_4,e_B,e_4),\
     \underline{\alpha}_{AB}=W(e_A,e_3,e_B,e_3),
     \\
     & \beta_A=\frac{1}{2}W(e_A,e_4,e_3,e_4),\
     \underline{\beta}_A=\frac{1}{2}W(e_A,e_3,e_3,e_4),
     \\
     & \rho=\frac{1}{4}W(e_4,e_3,e_4,e_3),\
     \sigma=\frac{1}{4}*W(e_4,e_3,e_4,e_3).
   \end{align*}
   We also define the Gauss curvature of the $2$-tori associated to the double null foliation to be $K$. Define also the following Ricci coefficients with respect to the null frame
   \begin{align*}
     & \chi_{AB}=g(D_Ae_4,e_B),\
     \underline{\chi}_{AB}=g(D_Ae_3,e_B),
     \\
     & \eta_A=-\frac{1}{2}g(D_3e_A,e_4),\
     \underline{\eta}_A=-\frac{1}{2}g(D_4e_A,e_3),
     \\
     & \omega=-\frac{1}{4}g(D_4e_3,e_4),\
     \underline{\omega}=-\frac{1}{4}g(D_3e_4,e_3),
     \\
     &\zeta_A=\frac{1}{2}g(D_Ae_4,e_3).
   \end{align*}
   Let $\hat{\chi}$ (resp. $\hat{\underline{\chi}}$) be the traceless part of $\chi$ (resp. $\underline{\chi}$).

   The data on $\underline{H}_0$ are given on $0\le u<u_*$ such that $\chib$ becomes singular as $u\rightarrow u_*$. More precisely, let $f:[0,u_*)\rightarrow\mathbb{R}$ be a smooth function such that $f(x)\ge0$ is decreasing and
   \begin{align*}
    &\int_{0}^{u_*}\frac{1}{f(x)^2}\mathrm{d}x<\infty.
   \end{align*}
   For simplicity, in the following, we fix $f$ to be $f(x)=(u_*-x)^{\frac{1}{2}}\log^p(\frac{1}{u_*-x})$ for $p>\frac{1}{2}$. Our main theorem shows local existence for a class of singular initial data with
   \begin{align*}
     &|\underline{\chi}(0,u)|\lesssim f(u)^{-2}.
   \end{align*}
   We construct a unique solution $(\mathcal{M},g)$ and $(v,\tau)$ to the Einstein--Euler equations in the region $0\le u < u_*,\ 0\le \underline{u}<\underline{u}_*$, where $u_*,\underline{u}_*\le\epsilon$.
   Here $(u,\underline{u})$ is a double null foliation for $(\mathcal{M},g)$
    and the metric $g$ takes the form as stated above.
   We further require $u_*$ to be sufficiently small such that
   \begin{align}\label{Eq1.5}
    &u_*\le \epsilon^2,\
    \int_{0}^{u_*}\frac{1}{f(x)^2}\mathrm{d}x<\epsilon^4.
   \end{align}
   Note that the second condition is more demanding in terms of smallness of $u_*$.
   In the following, we denote by $\ep$ the size of $\ub_*$, and by $\delta$ the size of $u_*$, where $\delta$ will be chosen significantly smaller than $\ep$.

   Define $\nabla$ to the induced Levi-Civita connection on the $2$-tori $S_{u,\ub}$ of constant $u$ and $\underline{u}$, i.e. $S_{u,\underline{u}}$ and $\nabla_3,\nabla_4,\slashed{\mathcal{L}}_\alpha$ to be the projections of the covariant derivatives $D_3,D_4$ and Lie derivatives $\mathcal{L}_\alpha$ to the tangent space of $S_{u,\underline{u}}$. Define the set of commutators $ Z \in\{\nabla_A,\epsilon\nabla_4\}$. 

   By prescribing $u_*$ to be sufficiently small as compared to the initial data, the spacetime reamins regular up to the weak null singularity. In particular, the fluid variables are free of shocks. We will see this in the estimate of the quadratic term in the equation of the fluid variables, see the third term in the proof of Proposition \ref{Prop4.4}. The following is a first version of the main result of this paper, the precise version can be found at Theorem \ref{SmoothExist} - \ref{ContinuousExist}:
   \begin{thm}\label{MainThm}
     For a class of singular characteristic initial data without any symmetry assumptions for the Einstein--Euler equations \eqref{Eq1} with a regular fluid profile and a singular geometric profile such that for some large $N$,
     \begin{align*}
       &\sum_{i\le N}|Z^i\hat{\underline{\chi}}|\sim (u_*-u)^{-1}\log^{-p}(\frac{1}{u_*-u}),\
       \text{for some }p>1,
     \end{align*}
     and for $\epsilon$ sufficiently small and $u_*,\underline{u}_*\le\epsilon$, there exists a unique smooth spacetime $(\mathcal{M},g)$ endowed with a double null foliation $(u,\underline{u})$ in $0\le u<u_*,\ 0\le \underline{u}<\underline{u}_*$, and unique smooth fluid represented by $(v,\tau)$, which satisfies the Einstein--Euler equations with the given data. Associated to $(\mathcal{M},g)$, there exists a coordinate system $(u,\underline{u},\theta^1,\theta^2)$ such that the metric and fluid variables $(v,\tau)$ extends continuously to the boundary, but the Christoffel symbols are not in $L^2$.
   \end{thm}
   
   \begin{rem}
     Theorem \ref{MainThm} shows that the Euler part does not form shocks in all of $\mathcal{M}$. This is despite, as we will show in the proof, the fact that the first derivatives of the fluid variables grow indefinitely as $u\to u_*$. In particular, despite the growth of the first derivatives of the fluid variables, there is insufficient time to form shocks before the weak null singularity.
   \end{rem}
   \begin{rem}
     While Theorem \ref{MainThm} is formulated so that both $u_*$ and $\ub_*$ are required to be small, one may expect that only smallness in $u_*$ is needed. In other words, for any $\ub_* >0$, there exists $\ep_{d,D,\ub_*}$ possibly smaller depending on $\ub_*$ such that the same result holds if the conditions \eqref{Eq1.5} are satisfied. To achieve this, one would use the structure of the Einstein equations as in \cite{10.1093/imrn/rnr201}. For the fluid part, one may expect to always gain a smallness constant due to its slower speed.
   \end{rem}
   
    Theorem 1 makes a statement about the metric being singular in terms of the Christoffel symbols not being in $L^2$. The statement is, however, not geometric, as the blowup is only guaranteed to occur in the $(u,\ub,\th^1,\th^2)$ coordinate system. The recent breakthrough of Sbierski \cite{Sbierski2024} shows that in the setting of a weak null singularity, the blowup of suitable curvature components can guarantee the Lipschitz inextendibility of the metric (in any coordinate systems). Using the results of \cite{Sbierski2024}, we prove that under an additional assumption of the initial data for the curvature, the spacetime metric is Lipschitz inextendible. Notice that the condition in the following corollary is an open condition.

   \begin{cor}\label{Rem1.1}
     Under the conditions of Theorem \ref{MainThm}, if the initial data also satisfy
     \begin{align}\label{CurvatureBlowup}
       |\underline{\alpha}_{11}|(u,0),\
       |\underline{\alpha}_{22}|(u,0) & \gtrsim \frac{1}{(u_*-u)^2\log^p(\frac{1}{u_*-u})},\
       |\underline{\alpha}_{12}|(u,0) \lesssim \frac{1}{(u_*-u)^2\log^{2p}(\frac{1}{u_*-u})},
     \end{align}
     then the metric does not admit a $C^{0,1}_{\mathrm{loc}}$-extension across the weak null singularity.
   \end{cor}
   For the proof of Corollary \ref{Rem1.1}, see Section \ref{Section6}.
   
   Our main theorem, Theorem \ref{MainThm}, can be stated precisely as a combination of Theorem \ref{SmoothExist}
   and Theorem \ref{ContinuousExist}. Theorem \ref{SmoothExist} shows that in double null coordinates, there exists a smooth spacetime and fluid variables solving the Einstein--Euler system with singular characteristic initial data before touching the singular boundary. Theorem \ref{ContinuousExist} shows that the metric components and the fluid variables can further extend continuously to the singular boundary.
   \begin{thm}
     \label{SmoothExist}
     Consider the characteristic initial value problem for Einstein--Euler system with data $g,v,\tau$ that are smooth and compatible at the corner on $H_0\cap\{0\le \underline{u}<\underline{u}_*\}$ and $\underline{H}_0\cap\{0\le u<u_*\}$ such that the following holds:
     \begin{itemize}
       \item The initial metric $\gamma_0$ on $S_{0,0}$ obeys
           \begin{align*}
             &0<d\le\det\gamma_0\le D,
           \end{align*}
           and
           \begin{align*}
             &|(\frac{\partial}{\partial\theta})^I\gamma_{BC}|\le D.
           \end{align*}
       \item The metric on $H_0,\underline{H}_0$ satisfies the gauge conditions
       \begin{align*}
         &\Omega=1\text{ on }H_0\text{ and }\underline{H}_0,
       \end{align*}
       and
       \begin{align*}
       &b^A=0\text{ on }\underline{H}_0.
       \end{align*}
       \item Ricci coefficients and fluid variables on the initial hypersurface $H_0$ verify
       \begin{align*}
         &\sum_{|i|\le 5}\sup_{\underline{u}}\| Z ^i\chi\|_{L^2(S_{0,\underline{u}})}
         +\sum_{|i|\le 4}\sup_{\underline{u}}\| Z ^i(\zeta,\mathrm{tr}\underline{\chi},v,\log\tau)\|_{L^2(S_{0,\underline{u}})}
         \le D.
       \end{align*}
       \item Ricci coefficients and fluid variables on the initial hypersurface $\underline{H}_0$ verify
       \begin{align*}
         &\sum_{|i|\le 5}\sup_{\underline{u}}\|f(u)^2 Z ^i\underline{\chi}\|_{L^2(S_{0,\underline{u}})}
         +\sum_{|i|\le 4}\sup_{\underline{u}}\| Z ^i(\zeta,\mathrm{tr}\chi,v,\log\tau)\|_{L^2(S_{0,\underline{u}})}
         \le D.
       \end{align*}
     \end{itemize}
     Then for $\epsilon=(1-(\sup p')^{\frac{1}{3}})^2\epsilon_{d,D}$, where $\epsilon_{d,D}$ (only depends on $d,D$) is
     sufficiently small, and $u_*,\ub_*$ sufficiently small such that $u_*,\underline{u}_*\le\epsilon^2$ and $\|f(u)^{-1}\|_{L^2_{u}}^2<\epsilon^4$, there exists a unique spacetime $(\mathcal{M},g)$ endowed with a double null foliation $(u,\underline{u})$ in $\{(u,\ub):0\le u<u_*,\ 0\le\underline{u}<\underline{u}_*\}$ and fluid represented by $(v,\tau)$, which is a solution to the Einstein--Euler equations \eqref{Eq1} with the given data. Moreover, the spacetime and fluid variables remains smooth in $0\le u<u_*$ and $0\le \underline{u}<\underline{u}_*$.
   \end{thm}
   
   \begin{rem}\label{Constructinginitialdata}
     There exist initial data $g,v,\tau$ prescribed on $H_0\cap\{0\le \underline{u}<\underline{u}_*\}$ and $\underline{H}_0\cap\{0\le u<u_*\}$ satisfying all the assumptions in the main theorem. See Proposition \ref{Prop3.1}.
   \end{rem}
   
   According to Theorem \ref{SmoothExist}, the metric components and fluid variables remains smooth at the boundary $\underline{H}_{\underline{u}_*}$. While the weight $f$ allows the spacetime to be singular, the spacetime metric and fluid variables can be extended beyond the singular hypersurface 
   $H_{u_*}$ continuously.

   \begin{thm}
     \label{ContinuousExist}
     Under the assumptions of Theorem \ref{SmoothExist}, the metric components of $g$ and fluid variables $(v,\tau)$ can be extended continuously up to and beyond the singular boundary $
     H_{u_*}$. Moreover, the induced metric and null second fundamental form on the interior of the limiting hypersurface $
     H_{u_*}$ is regular. More precisely, in double null frame, the metric components $\gamma, b, \Omega$ satisfy the following estimates
     \begin{align*}
       &\sum_{|i|\le4}\sup_{0\le u\le u_*}
       \|\slashed{\nabla}^i(\gamma,b,\Omega,v,\tau)\|_{L^2(S_{u,\underline{u}_*})}\le C.
     \end{align*}
     Moreover, for any fixed $u < u_*$, we have the following bounds for the Ricci coefficients
     \begin{align*}
       &\sum_{j\le1}\sum_{|i|\le 3-j}
       \|\nabla_3^j Z ^i(\hat{\underline{\chi}},\mathrm{tr}\underline{\chi},\underline{\omega},\eta,\underline{\eta})
       \|_{L^2(S_{u,\underline{u}_*})}\le C_u.
     \end{align*}
     Similar regularity statements hold on $H_{u_*}$.
   \end{thm}

   \subsection{Main Ideas of the Proof}The proof in this paper extends the arguments in \cite{luk2017weak} for the vacuum problem to the Einstein--Euler system. Introduction of the fluid requires modifying the proof in \cite{luk2017weak}. We show that the fluid variables are less singular and in turn exploit this to close the estimates for the coupled problem.

   \subsubsection{Vacuum Problem}
   We first recall the construction of weak null singularity in the vacuum case by Luk \cite{luk2017weak}. The main ingredients of the proof in \cite{luk2017weak} are energy estimates for renormalized curvature components and weighted $L^2$ estimates.
   Denoting by $\Gamma$ a general Ricci coefficient and $\Psi$ a general perpendicular curvature component with respect to a null frame adapted to the double null foliation. According to the estimates that they satisfy, we decompose $\Psi$ into $\Psi_H\in\{\rho,\sigma,\alpha,\b\},\Psi_{\Hb}\in\{\rho,\sigma,\ab,\bb\}$. Then the standard approach to obtain a priori bounds couples the $L^2$ energy estimates for the curvature components
   \begin{align}
   \label{Eq1.3}
     &\int_H\Psi_H^2+\int_{\underline{H}}\Psi_{\Hb}^2\le\text{Data}+\iint \Gamma\Psi\Psi,
   \end{align}
   and the estimates for the Ricci coefficients from transport equations
   \begin{align}
   \label{Eq1.4}
   \begin{split}
     &\nabla_3\Gamma=\Psi+\Gamma\Gamma,\\
     &\nabla_4\Gamma=\Psi+\Gamma\Gamma.
   \end{split}
   \end{align}


   However, in the setting of weak null singularities in the $u$ direction, some of the Ricci coefficients and curvature components are singular. For the Ricci coefficients, we group them into three groups, namely $\psi$, $\psi_H$, $\psi_{\underline{H}}$, according to the estimates that they obey. In particular, $\psi_{\underline{H}} \in \{ \hat{\underline{\chi}}, \mathrm{tr}\underline{\chi}, \underline{\omega}\}$ are singular components, which behaves like $f(u)^{-2}$ as $u \to u_*$, where $f$ is a weight function such that $f(u)^{-2}$ blows up in $L^2_u$ but still lies in $L^1_u$, while $\psi_H$ and $\psi$ are regular components which are bounded. In terms of the curvature components, $\rho,\sigma,\underline{\beta}$ also behaves like $f(u)^{-2}$, while $\underline{\alpha}$ is even worse. In particular, $(\rho,\sigma,\underline{\beta},\underline{\alpha})$ not being in $L^2$ does not allow the estimates \eqref{Eq1.3} to be proven.

   Instead of controlling the spacetime curvature components, Luk controlled renormalized curvature components, more precisely, the functions $K$ and $\check{\sigma}$ defined by
   \begin{align*}
     &K=-\rho+\frac{1}{2}\hat{\chi}\cdot\hat{\underline{\chi}}
     -\frac{1}{4}\mathrm{tr}\chi\mathrm{tr}\underline{\chi},
     \\
     &\check{\sigma}=\sigma+\frac{1}{2}\underline{\hat{\chi}}\wedge\hat{\chi}.
   \end{align*}
   Here, $K$ is the Gauss curvature of the $2$-spheres and $\check{\sigma}$ relates to the curvature of the normal bundle to the $2$-spheres ($2$-tori in our setting). In fact, $K,\check{\sigma}$ satisfy equations with terms less singular than the terms in the corresponding equations for $\rho,\sigma$. Up to lower order terms, the Bianchi equations
   \begin{align*}
     &\nabla_3\rho+\frac{3}{2}\mathrm{tr}\underline{\chi} \rho=-\mathrm{div}\underline{\beta}-\frac{1}{2}\hat{\chi}\cdot\underline{\alpha}+\cdots,
     \\
     &\nabla_3\sigma+\frac{3}{2}\mathrm{tr}\underline{\chi}\sigma
     =-\mathrm{div}*\underline{\beta}+\frac{1}{2}\hat{\chi}\wedge\underline{\alpha}+\cdots,
   \end{align*}
   contain the non-integrable curvature component $\underline{\alpha}$. On the other hand, the intrinsic curvatures $(K,\check{\sigma})$ obeys the equations
   \begin{align*}
     &\nabla_3K+\frac{3}{2}\mathrm{tr}\chi K=\mathrm{div}\underline{\beta}+\cdots,
     \\
     &\nabla_3\check{\sigma}+\frac{3}{2}\mathrm{tr}\underline{\chi}\check{\sigma}
     =-\mathrm{div}*\beta+\cdots,
   \end{align*}
   where there are no terms containing $\underline{\alpha}$ or quadratic terms in $\mathrm{tr}\underline{\chi},\hat{\underline{\chi}}$ or $\underline{\omega}$. Hence, every term on the right hand side is integrable in $u$ direction.

   Here $(K,\check{\sigma})$ are regular, and will be bounded in the $\|(K,\check{\sigma})\|_{L_{\underline{u}}^\infty L_u^2 L^2(S)}$
   norm. For this system of equations, no $\alpha$ or $\underline{\alpha}$ is needed. However, $\underline{\beta}$ is still singular, but will be controlled in $\|f(u)\underline{\beta}\|_{L_{\underline{u}}^\infty L_u^2 L^2(S)}$ norm. Accordingly, we can decompose the renormalized curvature components $\check{\Psi}\in\{K,\check{\sigma},\b,\bb\}$ into $\check{\Psi}_H\in\{(K,\check{\sigma}),\b\},\ \check{\Psi}_{\Hb}\in\{\bb,(K,\check{\sigma})\}$, and rewrite the energy estimates formally, omitting the $f(u)$ weight in the pairing $(K,\check{\sigma}),\ \bb$, as
   \begin{align}
   \label{Eq1.3.1}
     &\int_H\check{\Psi}_H^2+\int_{\underline{H}}\check{\Psi}_{\Hb}^2\le\text{Data}+\iint \Gamma\check{\Psi}\check{\Psi},
   \end{align}

   A priori, the degenerate $L^2$ estimates may not control the error terms. To carry out the argument using these weak estimates require the null structure of the vacuum Einstein equations. For example, as in \cite{luk2017weak}, in the energy estimates for the singular component $\underline{\beta}$,
   \begin{align*}
     \|f(u)\underline{\beta}\|_{L^2(\underline{H})}^2 &\le \text{Data}+\|f(u)^2(\underline{\beta}\psi_H\underline{\beta}
     +\underline{\beta}\psi_{\underline{H}}\beta+\underline{\beta}\psi K)\|_{L_u^1L_{\underline{u}}^1L^1(S)}.
   \end{align*}
   It suffices to note that $\psi_{\underline{H}}$, while singular, can be shown to be small after integrating along the $u$ direction. Thus, the three terms can be bounded using Gronwall's inequality. Notice that if other combinations of curvature terms and Ricci coefficients such as $\underline{\beta}\psi_{\underline{H}}\underline{\beta},\underline{\beta}\psi_H\underline{\beta}$ or $\underline{\beta}\psi_{\underline{H}}K$ appear in the error terms, the degenerate energy will not be strong enough to close the estimates.

   To close all the estimates, we need to commute also with higher derivatives. As in \cite{https://doi.org/10.1002/cpa.21531}, \cite{Luk2013NonlinearIO}, \cite{luk2017weak}, commuting with partial derivatives $\epsilon\nabla_4$ and angular derivatives will not introduce terms that are more singular. Moreover, the null structure of the estimates is also preserved under these commutations.

   Similar to \cite{https://doi.org/10.1002/cpa.21531}, \cite{Luk2013NonlinearIO}, \cite{luk2017weak}, the renormalization introduces error terms in the energy estimates concerning Ricci coefficients with one more derivative compared to the curvature terms. To estimates these terms, we need top order  elliptic estimates on the tori in addition to the estimates via transport equations. A similar null structure appear in the elliptic estimates, allowing all the bounds to be closed.

   \subsubsection{Less Singular Nature of Fluid Variables}\label{Section1.3.2}
   We now discuss the estimates for the fluid variables in the presence of a weak null singularity.
   Denoting by $r$ a general fluid variable, one has the following energy estimates:
   \begin{align*}
     \int_Hrr+\int_{\underline{H}}rr\le\text{Data}+\iint \Gamma rr.
   \end{align*}
   In particular, the energy estimates allow us to control all components of the fluid variables without degeneration on the hypersurfaces $H$ and $\underline{H}$. This is due to the fact that while $H$ and $\underline{H}$ are null with respect to the spacetime metric, they are spacelike with respect to the acoustical metric. (Notice that this relies on the speed of sound being strictly small than the speed of light.) Moreover, we can take $\nabla_4$ and angular derivatives of the fluid variables because they are regular derivatives for the spacetime metric. Since these derivatives span the tangent space of $H$, Sobolev embedding allows us to conclude that the fluid variables are bounded up to the weak null singularity.

   On the other hand, notice that due to the singularity in the derivatives of the metric, the singular derivative $\nabla_3$ of the fluid variables blows up quantitatively as $f(u)^{-2}$. This can be proven directly from the fluid equation, since $\nabla_3$ is a non-characteristic derivative and the singularity is sourced by the singular metric terms on the right-hand side of the equation.



   The fact that the fluid variables are more regular reflects a more general phenomenon that can be seen in hyperbolic systems with multiple speeds. We refer the reader to the work of Speck \cite{Speck2017}, where a similar phenomenon is seen in the problem of shock formation for a system of two wave equations of different speeds. See also \cite{FritzJohnFormationofSingularities}, \cite{Luk2018}. In this problem, the speed of sound in the relativistic Euler equations is required to be strictly smaller than the speed of light, see Section \ref{Section5}.

   \subsubsection{Coupling of Geometry and Fluid}
   For the Einstein--Euler system, introduction of fluid variables gives rise to nonzero Ricci curvature terms, in contrast to the vacuum problem. Denote a general renormalized curvature component as $\check{\Psi}$. An adaptation of the $L^2$ energy estimates in \cite{luk2017weak} for the spacetime metric produces
   \begin{align*}
     &\int_H\check{\Psi}_H^2+\int_{\underline{H}}\check{\Psi}_{\Hb}^2\le\text{Data}+\iint \Gamma\check{\Psi}\check{\Psi}+\iint \check{\Psi} D\mathrm{Ric}
     +\iint \Gamma D\Gamma\check{\Psi}
     +\iint\Gamma\mathrm{Ric}\check{\Psi}.
   \end{align*}
   The $D\mathrm{Ric}$ term could be singular. In fact, this could involve $\nabla_3$ of fluid variables or a singular Ricci coefficient, both of which are merely in $L^1_u$. However, the singular terms only appear either when: (1) multiplied by a renormalized curvature component which can be controlled by the energy on $H$, and thus spares the $u$ integral for the singular term; (2) in the most singular energy estimates for $(K,\check{\sigma},\underline{\beta})$, where such a singularity is still consistent with the weight that we introduced.

   Following the renormalization of curvature components as in \cite{luk2017weak}, the energy estimates for the triple $(K,\check{\sigma},\beta)$ only involve regular derivatives, consisting of $\nabla_4$ and angular derivatives, of Ricci curvature. This is compatible with the less singular nature of $\beta$. In addition, we also adapt the proof to the pair $(\alpha,\beta)$ without further obstruction.
   
   For Ricci coefficients, the transport equations are
   \begin{align*}
     &\nabla_3\Gamma=\check{\Psi}+\Gamma\Gamma+\mathrm{Ric},\\
     &\nabla_4\Gamma=\check{\Psi}+\Gamma\Gamma+\mathrm{Ric}.
   \end{align*}
   Since the Ricci curvature terms are bounded, they do not pose any additional difficulties.

   \subsubsection{Null Coordinates and Commutators}


   In \cite{luk2017weak}, only projected angular covariant derivatives act as commutators. However, taking angular derivatives only is too restrictive for the spacetime integral arising in the energy current approach for fluid variables. Even if we were to commute only with angular derivatives, the commutator terms arising in the divergence of the fluid energy current involve all derivatives.

   Instead, we commute with regular derivatives $ Z $, consisting of $\epsilon\nabla_4$ and angular covariant derivatives, as the set of vector fields tangential to the null hypersurface $H$ where the weak null singularity occurs. Since $H$ is spacelike in acoustical metric and hence non-characteristic, we can recover $\nabla_3$ derivatives of the fluid variables using the equations.

   In \cite{luk2017weak}, when only angular commutations were used, one could control $N$ angular derivatives of any Ricci coefficients by $N-1$ angular derivatives of the renormalized curvature components using angular elliptic estimates. In our setting, however, because we use both angular derivatives and $\epsilon\nabla_4$ as commutators, at the top order, we will need to control the $\nabla_4$ derivatives in addition to the angular derivatives of the Ricci coefficients. In order to achieve this for the Ricci coefficients $\omega$ and $\underline{\omega}$, we need to normalize our null frames suitably, choosing $e_3 = \partial_{u} + b^A \partial_A$ and $e_4 = \Omega^{-2} \p_{\ub}$ (instead of, say, $e_3 = \Omega^{-1}(\partial_{u} + b^A \partial_A)$, $e_4 = \Omega^{-1} \p_{\ub}$). With such a normalization, $\omega = 0$ while $\nabla_4\underline{\omega}$ is admissible,
    see Section \ref{SingularFluidVariables}. The other Ricci coefficients are all easier because they behave well under transportation by $\nabla_4$. As a result, when proving bounds for the fluid velocity, we also have to decompose it with respect to a null frame defined in this manner.
    
   It is important to note that, for the Einstein--Euler system, we employ commutators built from projected covariant derivatives, instead of ambient covariant derivatives or Lie derivatives. The ambient covariant or Lie derivatives applied to fluid variables do not, in general, eliminate the singular components that arise in the regular directional derivatives of fluid variables. For example, $D_Ab^B\sim e_Ab^B+\underline{\chi}b^3+\chi b^4$, where $e_Ab^B$ is assumed to remain bounded as prescribed on the initial data, while $\underline{\chi}$ is singular. Consequently, such derivatives do not isolate the singular contribution coming from geometry. On the other hand, the difficulty with applying Lie derivatives arises in controlling the shift vector $b$ in the metric components. Differentiating the Euler equation with respect to Lie derivatives produces the term $e_A(b^B)$, in addition to derivatives of metric components which can be combined into Ricci coefficients. However, from the evolution equation $e_4b\sim(\eta-\etab)$, controlling $N$ derivatives of $e_A(b^B)$ requires controlling $N+1$ derivatives of Ricci coefficients, which prevents the estimates we expect. 
   
   In contrast, if we first decompose the fluid velocity with respect to the null frame, and then differentiate using projected covariant derivatives, we only generate Christoffel symbols which can be combined into Ricci coefficients, see Section \ref{SchematicEquations}. Moreover, singular Ricci components only appear in the singular directional derivatives, see Section \ref{ConstructionofInitialdata}. For this reason, we use projected covariant derivatives to differentiate the fluid part of the system.





   \subsection{Outline of the Proof}
   \begin{itemize}
     \item In Section \ref{Section3}, we define the energy and auxiliary norms.
     \item In Section \ref{Section3.1}, we reduce the main theorem to a priori estimates.
     \item In Section \ref{Section4}, we prove the a priori estimates for spacetime metric.
     \item In Section \ref{Section5}, we prove the a priori estimates for fluid variables.
     \item In Section \ref{Section6}, we complete the proof of the a priori estimates and the proof of continuous extendibility.
   \end{itemize}

   \section{Norms and Construction of Initial Data}\label{Section3}

   \subsection{Norms}\label{norms}
   As in \cite{luk2017weak}, we introduce the schematic notation
   \begin{align*}
     \psi_0\in\{b,\gamma,\gamma^{-1},\log\Omega\},\
     \psi\in\{\eta,\underline{\eta}\},\
     \psi_H\in\{\mathrm{tr}\chi,\hat{\chi},\omega\},\
     \psi_{\underline{H}}\in\{\mathrm{tr}\underline{\chi},\hat{\underline{\chi}},\underline{\omega}\}.
   \end{align*}
   We also introduce
   \begin{align*}
     r\in\{\slashed{v}^3,\ \slashed{v}^4,\ \slashed{v},\ \log\tau\}.
   \end{align*}
   where $\slashed{v}^3,\ \slashed{v}^4$ are scalar functions and $\slashed{v}=(\slashed{v}^1,\slashed{v}^2)$ is a tensor on the $2$-torus.
   Recall that for $A\in\{1,2\}$, the commutators are
   \begin{align*}
      Z \in\{\nabla_{A},\epsilon\nabla_4\},
     \slashed{\nabla}\in\{\nabla_A\}.
   \end{align*}
   Define for every $i\in\mathbb{N}$ and any tensor $\phi$, for $a$ multi-indices and $\|\|$ some norm,
   \begin{align*}
     &\| Z ^i\phi\|=\sum_{|a|=i}\| Z ^{a}\phi\|.
   \end{align*}
   Our norms will be of the form $L_u^pL_{\underline{u}}^qL^r(S)$, defined as
   \begin{align*}
     &\|\phi\|_{L_u^pL_{\underline{u}}^qL^r(S)}
     =\Big(\int\Big(\int\|\phi\|_{L^r(S)}^q d\underline{u}\Big)^{\frac{p}{q}} du\Big)^{\frac{1}{p}},
   \end{align*} 
   where $L_u^p$ and $L_{\underline{u}}^q$ are defined in measures $du$ and $d\underline{u}$ respectively and $L^r(S)$ is defined with respect to the surface measure $\mathrm{d}\sigma=\sqrt{|\det \gamma|}\mathrm{d}\theta^A\wedge\mathrm{d}\theta^B$ on $S_{u,\underline{u}}$, componentwise as
   \begin{align*}
     \| Z ^i\phi\|_{L^r(S_{u,\underline{u}})} &=\Big(\int_{S_{u,\underline{u}}}
     \sum_{I,J}| Z ^i(\phi_{\alpha_I}{}^{\alpha_J})|^r d\sigma
     \Big)^{\frac{1}{r}}.
   \end{align*}
   Define norms for the highest derivatives for the Ricci coefficients,
   \begin{align*}
     \tilde{\mathcal{O}}_{4,2}
     =&\sum_{i\le4}
     \Big(
     \epsilon^{-\frac{1}{2}}\|f(u) Z ^i\mathrm{tr}\underline{\chi}\|_{L_{u}^2 L_{\underline{u}}^\infty  L^2(S)}
     +\epsilon^{-\frac{1}{2}}\|f(u) Z ^i(\hat{\underline{\chi}},\underline{\omega})\|_{L_{\underline{u}}^\infty L_{u}^2 L^2(S)}
     +\epsilon^{-\frac{1}{2}}\| Z ^i(\psi,\psi_0)\|_{L_{\underline{u}}^\infty L_{u}^2 L^2(S)}
     \Big)
     \\
     &+\sum_{i\le4}
     \Big(
     \epsilon^{-\frac{1}{2}}\| Z ^i \mathrm{tr}\chi\|_{L_{\underline{u}}^2L_{u}^\infty     L^2(S)}
     +\epsilon^{-\frac{1}{2}}\| Z ^i (\hat{\chi},\omega)\|_{L_{u}^\infty L_{\underline{u}}^2 L^2(S)}
     +\epsilon^{-\frac{1}{2}}\|f(u) Z ^i (\psi,\psi_0)\|_{L_{u}^\infty L_{\underline{u}}^2 L^2(S)}
     \Big).
   \end{align*}
   Also, define norms for lower order derivatives for the Ricci coefficients,
   \begin{align*}
     \mathcal{O}_{i,p}
     =&\sum_{j\le i}\epsilon^{-\frac{1}{2}}\| Z ^j \psi_H\|_{ L_{\underline{u}}^2 L_{u}^\infty L^p(S)}
     +\sum_{j\le i}\epsilon^{-\frac{1}{2}}\|f(u) Z ^j \psi_{\underline{H}}\|_{ L_{u}^2 L_{\underline{u}}^\infty L^p(S)}
     \\
     &+\sum_{j\le i}\| Z ^j(\psi,\psi_0)\|_{L_{\underline{u}}^\infty L_{u}^\infty L^p(S)}.
   \end{align*}
   In addition, define the curvature norms for the curvature components,
   \begin{align*}
     \mathcal{R}_3=&\sum_{i\le3}
     \Big(
     \epsilon^{-\frac{1}{2}}\| Z ^i \beta\|_{L_{u}^\infty L_{\underline{u}}^2 L^2(S)}
     +\epsilon^{-\frac{1}{2}}\| Z ^i (K,\check{\sigma})\|_{L_{\underline{u}}^\infty L_{u}^2L^2(S)}
     \Big)
     \\
     &+\sum_{i\le3}
     \Big(
     \epsilon^{-\frac{1}{2}}\|f(u) Z ^i \underline{\beta}\|_{L_{\underline{u}}^\infty L_{u}^2L^2(S)}
     +\epsilon^{-\frac{1}{2}}\|f(u) Z ^i (K,\check{\sigma})\|_{L_{u}^\infty L_{\underline{u}}^2 L^2(S)}
     \Big)
     \\
     &+\sum_{i\le3}
     \Big(
     \| Z ^i\alpha\|_{L_u^\infty L_{\underline{u}}^2 L^2(S)}
     +\| Z ^i \beta\|_{L_{\underline{u}}^\infty L_{u}^2 L^2(S)}
     \Big).
   \end{align*}
   Also, define the corresponding norms for the initial data,
   \begin{align*}
     \mathcal{O}_{ini} =& \sum_{i\le3
     }
     \big(
     \| Z ^i(\psi,\psi_0)\|_{L_{\underline{u}}^\infty L^2(S_{\underline{u},0})}
     +\| Z ^i(\psi,\psi_0)\|_{L_{u}^\infty L^2(S_{0,u})}
     \\
     &+\epsilon^{-\frac{1}{2}}\| Z ^i \psi_H\|_{ L_{\underline{u}}^2 L^2(S_{\underline{u},0})}
     +\epsilon^{-\frac{1}{2}}\|f(u) Z ^i \psi_{\underline{H}}\|_{ L_{u}^2  L^2(S_{0,u})}
     \big)
     \\
     &+\| Z ^4\mathrm{tr}\chi\|_{L_{\underline{u}}^\infty L^2(S_{\underline{u},0})}
     +\|f(u)^2 Z ^4\mathrm{tr}\underline{\chi}\|_{L_{u}^\infty L^2(S_{0,u})}
     \\
     &+\| Z ^4\mathrm{tr}\underline{\chi}\|_{L_{\underline{u}}^\infty L^2(S_{\underline{u},0})}
     +\| Z ^4\mathrm{tr}\chi\|_{L_{u}^\infty L^2(S_{0,u})}
     \\
     &+\epsilon^{-\frac{1}{2}}\| Z ^4 (\hat{\chi},\omega)\|_{L_{\underline{u}}^2 L^2(S_{\underline{u},0})}
     +\epsilon^{-\frac{1}{2}}\|f(u) Z ^4(\psi,\psi_0)\|_{L_{\underline{u}}^2 L^2(S_{\underline{u},0})}
     \\
     &+\epsilon^{-\frac{1}{2}}\|f(u) Z ^4 (\hat{\underline{\chi}},\underline{\omega})\|_{ L_{u}^2 L^2(S_{0,u})}
     +\epsilon^{-\frac{1}{2}}\| Z ^4 (\psi,\psi_0)\|_{ L_{u}^2 L^2(S_{0,u})},
   \end{align*}
   and
   \begin{align*}
     \mathcal{R}_{ini}=&\sum_{|i|\le3}
     \Big(
     \epsilon^{-\frac{1}{2}}\| Z ^i\beta\|_{L_{\underline{u}}^2 L^2(S_{0,\underline{u}})}
     +\epsilon^{-\frac{1}{2}}\| Z ^i (K,\check{\sigma})\|_{L_{u}^2L^2(S_{u,0})}
     \Big)
     \\
     &+\sum_{i\le3}
     \Big(
     \epsilon^{-\frac{1}{2}}\|f(u) Z ^i \underline{\beta}\|_{L_{u}^2L^2(S_{u,0})}
     +\epsilon^{-\frac{1}{2}}\|f(u) Z ^i (K,\check{\sigma})\|_{L_{\underline{u}}^2 L^2(S_{0,\underline{u}})}
     \Big)
     \\
     &+\sum_{i\le3}
     \Big(
     \| Z ^i\alpha\|_{L_{\underline{u}}^2 L^2(S_{0,\underline{u}})}
     +\| Z ^i\beta\|_{L_{u}^2L^2(S_{u,0})}
     \Big).
   \end{align*}
   Similarly, define norms for the regular highest derivatives for the fluid variables,
   \begin{align*}
     \tilde{\mathcal{E}}_{i,2}
     =&
     \sum_{j\le i}\epsilon^{-\frac{1}{2}}\| Z ^j (\slashed{v},\log\tau)\|_{L_{\underline{u}}^\infty L_{u}^2L^2(S)}
     +\sum_{j\le i}\epsilon^{-\frac{1}{2}}\| Z ^j (\slashed{v},\log\tau)\|_{L_{u}^\infty L_{\underline{u}}^2L^2(S)},
   \end{align*}
   and norms with singular derivative $\nabla_3$ involved as:
   \begin{align*}
     \tilde{\mathcal{F}}_{4,2}=
     \sum_{i\le3}\|f(u)\nabla_3 Z ^i (\slashed{v},\log\tau)\|_{L_{\underline{u}}^2 L_{u}^2L^2(S)}.
   \end{align*}
   Also, define norms for lower order derivatives for the fluid variables,
   \begin{align*}
     \mathcal{E}_{2,\infty}
     =
     &\sum_{i\le2}
     \Big(
     \| Z ^i (\slashed{v},\log\tau)\|_{L_{\underline{u}}^\infty L_{u}^\infty L^\infty(S)}
     +\|f(u)^2\nabla_3 Z ^i (\slashed{v},\log\tau)\|_{L_{u}^\infty L_{\underline{u}}^\infty L^\infty(S)}
     \Big)
     .
   \end{align*}
   Finally, let $\mathcal{E}_{ini}$ denote the corresponding norms for the initial data,
   \begin{align*}
     \mathcal{E}_{ini}
     =&\sum_{i\le4}
     \Big(
     \epsilon^{-\frac{1}{2}}\| Z ^i(\slashed{v},\log\tau)\|_{L_{\underline{u}}^2 L^2(S_{\underline{u},0})}
     +\epsilon^{-\frac{1}{2}}\| Z ^i(\slashed{v},\log\tau)\|_{L_{u}^2 L^2(S_{0,u})}
     \Big)
     \\
     &+\sum_{i\le3}
     \Big(
     \|f(u)^2 \nabla_3 Z ^i (v,\log\tau)\|_{L_{\underline{u}}^2 L^2(S_{\underline{u},0})}
     +\|f(u)^2 \nabla_3 Z ^i (v,\log\tau)\|_{L_{u}^2 L^2(S_{0,u})}
     \Big)
     \\
     &+\sum_{i\le2}
     \Big(
     \| Z ^i(\slashed{v},\log\tau)\|_{L_{\underline{u}}^\infty L^\infty(S_{\underline{u},0})}
     +\| Z ^i(\slashed{v},\log\tau)\|_{L_{u}^\infty L^\infty(S_{0,u})}
     \Big)
     .
   \end{align*}
   \subsection{Commutation Formulae}\label{Section3.2}
   We shall need the following commutation formulae for $\slashed{\nabla},\ \nabla_3,\ \nabla_4$ for tensors on the torus, see for example \cite{Christodoulou:2008nj},
   \begin{prop}\label{Prop2.1}
     As for $\phi$ any covariant tensor on the torus,
     \begin{align*}
        [\nabla_4,\nabla_B]\phi_{A_1...A_r}
       =&(\etab_B+\zeta_B)\nabla_4\phi_{A_1...A_r}
       -\gamma^{CD}\chi_{BD}\nabla_C\phi_{A_1...A_r}
       \\
       &+\sum_i (\gamma^{CD}\chi_{A_iB}\etab_D
       -\gamma^{CD}\chi_{BD}\etab_{A_i}
       +\epsilon_{A_i}{}^C*\beta_B)\phi_{A_1...\hat{A}_iC...A_r},
       \\
        [\nabla_3,\nabla_B]\phi_{A_1...A_r}
       =&(\eta_B-\zeta_B)\nabla_3\phi_{A_1...A_r}
       -\gamma^{CD}\chib_{BD}\nabla_C\phi_{A_1...A_r}
       \\
       &+\sum_i (\gamma^{CD}\chib_{A_iB}\eta_D
       -\gamma^{CD}\chib_{BD}\eta_{A_i}
       +\epsilon_{A_i}{}^C*\bb_B)\phi_{A_1...\hat{A}_iC...A_r},
       \\
       [\nabla_3,\nabla_4]\phi_{A_1...A_r}
       =&-2\om\nabla_3\phi_{A_1...A_r}+2\omb\nabla_4\phi_{A_1...A_r}
       +2\gamma^{BC}(\eta_B-\etab_B)\nabla_C\phi_{A_1...A_r}
       \\
       &+\sum_i(2\gamma^{BC}\etab_{A_i}\eta_B-\eta_{A_i}\etab_B
       +\epsilon_{A_iB}\sigma)\phi_{A_1...\hat{A}_iC...A_r},
       \\
       [\nab_A,\nab_B]\phi_{A_1...A_r}
       =&K(\gamma_{AA_i}U_{A_1...\hat{A}_iB...A_r}-\gamma_{BA_i}\phi_{A_1...\hat{A}_iA...A_r}).
     \end{align*}
   \end{prop}
   By schematic notation $a=_s\sum a_i$, we mean by $a=\sum c_i a_i+c_0$ for some $c_i=C(\Delta_1,\mathcal{E}_{2,\infty})$. In particular, by Sobolev embeddings, $ Z \psi,  Z ^2r$ are of this kind.
   By induction and schematic Codazzi equation \eqref{Eq2.9} - \eqref{Eq2.10} and \eqref{Eq2.41},
   \begin{align*}
     &Z^i\beta=_sZ^i\nabla\psi_H+\sum_{i_1+i_2\le i}Z^{i_1}\psi Z^{i_2}\psi_H,
     \\
     &Z^i\bb=_sZ^i\nabla\psi_{\Hb}+\sum_{i_1+i_2\le i}Z^{i_1}\psi Z^{i_2}\psi_{\Hb},
     \\
     &Z^i\sigma=Z^i\nabla\psi+\sum_{i_1+i_2\le i}Z^{i_1}\psi_H Z^{i_2}\psi_{\Hb}.
   \end{align*}
   we obtain the schematic formulae for commutations, see for example \cite{dafermos2017interiordynamicalvacuumblack},
   \begin{prop}\label{Prop2.2}
     Applying $ Z ^i$ obeys the following schematic equations
    \begin{align*}
     &\nabla_4 Z ^i\phi=_s\sum_{ i_1+ i_2+ i_3\le i}  Z ^{i_1}\psi^{i_2} Z ^{i_3}\nabla_4\phi
     +\sum_{ i_1+ i_2+ i_3+ i_4\le i}  Z ^{i_1}\psi^{i_2} Z ^{i_3}\psi_H Z ^{i_4}\phi,
     \\
     &\nabla_3 Z ^i\phi=_s\sum_{ i_1+ i_2+ i_3\le i}  Z ^{i_1}\psi^{i_2} Z ^{i_3}\nabla_3\phi
     +\sum_{i_1+i_2+i_3+i_4\le i}  Z ^{i_1}(\psi+\psi_H)^{i_2} Z ^{i_3}\psi_{\underline{H}} Z ^{i_4}\phi,
     \\
     &\nabla_A Z ^i\phi=_s\sum_{ i_1+ i_2+ i_3\le i}  Z ^{i_1}\psi^{i_2} Z ^{i_3}\nabla_A\phi
     +\sum_{i_1+i_2\le i-1} Z^{i_1}K Z^{i_2}\phi
     +\sum_{ i_1+ i_2+ i_3+ i_4\le i}  Z ^{i_1}\psi^{i_2} Z ^{i_3}\psi_H Z ^{i_4}\phi.
   \end{align*}
   \end{prop}
   As compared with \cite{luk2017weak}, we have extra terms when commutating with $\nabla_3$ and $\nabla_A$ equations, which arise from the extra commutator $\ep\nabla_4$. In Section \ref{Section4}, we will focus on estimating the extra terms
   \begin{align*}
     &\sum_{i_1+i_2+i_3+i_4\le i}  Z ^{i_1}\psi_H^{i_2} Z ^{i_3}\psi_{\underline{H}} Z ^{i_4}\phi
   \end{align*}
   in $\nabla_3Z^i\phi$ equations and
   \begin{align*}
     &\sum_{ i_1+ i_2+ i_3+ i_4\le i}  Z ^{i_1}\psi^{i_2} Z ^{i_3}\psi_H Z ^{i_4}\phi
   \end{align*}
   in $\nabla_AZ^i\phi$ equations.
   
   \subsection{Construction of the Initial data}\label{ConstructionofInitialdata}
   \begin{prop}\label{Prop3.1}
     There exists a family of initial data prescribed on $H_0\cap\{0\le \underline{u}<\underline{u}_*\}$ and $\underline{H}_0\cap\{0\le u<u_*\}$ satisfying all the assumptions in Theorem \ref{SmoothExist}, including:
     \begin{itemize}
     \item The metric components $(\gamma,\Omega,b)$
     prescribed on $H_0\cap\{0\le \underline{u}<\underline{u}_*\}$ and $\underline{H}_0\cap\{0\le u<u_*\}$ satisfy the following:
     \begin{align*}
     &| (\p_{\ub},\slashed{\p}) ^i\gamma|\Big{|}_{S_{0,0}}\le D,\
     0<d\le\det\gamma_0\le D,\
     \Omega|_{H_0\cup\underline{H}_0}=1,\
     b^A|_{\underline{H}_0}=0.
     \end{align*}
     \item The Ricci coefficients $(\chi,\underline{\chi},\zeta)$ defined as
     \begin{align*}
       &\Omega\chi=\frac{1}{2}\partial_{\ub}\gamma,\
       \Omega^{-1}\underline{\chi}=\frac{1}{2}\partial_u\gamma,\
       \zeta=\frac{1}{4}\partial_ub,
     \end{align*}
     satisfying the constraints on $\underline{H}_0$
     \begin{align*}
     &\slashed{\mathcal{L}}_{e_3}\mathrm{tr}\underline{\chi}
     =-\frac{1}{2}(\mathrm{tr}\underline{\chi})^2-|\hat{\underline{\chi}}|^2-\mathrm{Ric}_{33},
     \\
      &\slashed{\mathcal{L}}_{e_3}\mathrm{tr}\chi
      +\mathrm{tr}\underline{\chi}\mathrm{tr}\chi
      =-2K-2\mathrm{div}\zeta+2|\zeta|^2-\mathrm{Ric}_{34},
      \\
      &\slashed{\mathcal{L}}_{e_3}\zeta+\mathrm{tr}\underline{\chi}\zeta
      =\slashed{\mathrm{div}}\underline{\chi}-\nabla\mathrm{tr}\underline{\chi}+\mathrm{Ric}_{3B},
     \end{align*}
     and the constraints on $H_0$
     \begin{align*}
     &\slashed{\mathcal{L}}_{e_4}\mathrm{tr}\chi
     =-\frac{1}{2}(\mathrm{tr}\chi)^2-|\hat{\chi}|^2-\mathrm{Ric}_{44},
     \\
      &\slashed{\mathcal{L}}_{e_4}\mathrm{tr}\underline{\chi}
      +\mathrm{tr}\chi\mathrm{tr}\underline{\chi}
      =-2K-2\mathrm{div}\zeta+2|\zeta|^2-\mathrm{Ric}_{43},
      \\
      &\slashed{\mathcal{L}}_{e_4}\zeta+\mathrm{tr}\chi\zeta
      =\slashed{\mathrm{div}}\chi-\nabla\mathrm{tr}\chi+\mathrm{Ric}_{4B},
     \end{align*}
     and the bounds
     \begin{align*}
     &\sum_{i\le 5}\| Z ^i\chi\|_{L^\infty_{\underline{u}}L^2(S_{0,\underline{u}})}
     +\sum_{i\le4}\| Z ^i(\zeta,\mathrm{tr}\underline{\chi})\|_{L^\infty_{\underline{u}}L^2(S_{0,\underline{u}})}
     \Big{|}_{H_0}\le D,
     \\
     &\sum_{i\le 5}\|f(u)^2 Z ^i\underline{\chi}\|_{L^\infty_{u}L^2(S_{0,u})}
     +\sum_{i\le4}\| Z ^i(\zeta,\mathrm{tr}\chi)\|_{L^\infty_{u}L^2(S_{0,u})}
     \Big{|}_{\underline{H}_0}\le D,
   \end{align*}
   with an exact blow-up rate
   \begin{align*}
     &|\hat{\underline{\chi}}|(u,0,\theta)\sim f(u)^{-2}.
   \end{align*}
    \item The fluid initial data
     $(v,\log\tau)$ prescribed on $H_0\cap\{0\le \underline{u}<\underline{u}_*\}$ and $(v,\log\tau)$ prescribed on $\underline{H}_0\cap\{0\le u<u_*\}$ satisfying the following:
     \begin{align*}
     &\sum_{i\le 4}\|\nabla_3 Z ^i (\slashed{v},\log\tau)\|_{L^\infty_{\underline{u}}L^2(S_{0,\underline{u}})}
     \Big{|}_{H_0}\le D,
     \\
     &\sum_{i\le 5}\| Z ^i (\slashed{v},\log\tau)\|_{L^\infty_{\underline{u}}L^2(S_{0,\underline{u}})}
     \Big{|}_{H_0}\le D,
     \\
     &\sum_{i\le 4}\|f(u)^2\nabla_3 Z ^i(\slashed{v},\log\tau)\|_{L^\infty_{u}L^2(S_{0,u})}
     \Big{|}_{\underline{H}_0}\le D,
     \\
     &\sum_{i\le 5}\| Z ^i(\slashed{v},\log\tau)\|_{L^\infty_{u}L^2(S_{0,u})}
     \Big{|}_{\underline{H}_0}\le D.
     \end{align*}
     where the transversal projected covariant derivatives $(\nabla_3\slashed{v},\nabla_3\log\tau)$ on $H_0\cap\{0\le \underline{u}<\underline{u}_*\}$ and $(\nabla_4\slashed{v},\nabla_4\log\tau)$ on $\underline{H}_0\cap\{0\le u<u_*\}$ are defined using the fluid equations.
     \end{itemize}
   \end{prop}
   \begin{proof}
   We follow the approach as in 
   \cite{Christodoulou:2008nj}.
   Let $(\theta^1,\theta^2)$ be the standard local coordinates on torus $S_{0,0}\simeq\mathbb{R}^2/\mathbb{Z}^2$. 
   
   We begin with the initial data on $\underline{H}_0$, which is the harder case because of the presence of singularity. On $\underline{H}_0$, we set $\Omega=1, b=0$ and therefore $e_3=\partial_{u}$. We construct a metric on $\underline{H}_0$ in $(u,\theta_1,\theta_2)$ coordinates as
   \begin{align*}
     \gamma_{AB}=\Phi^2\hat{\gamma}_{AB},\
     \hat{\gamma}_{AB}=\exp\Psi.
   \end{align*}
   The freely prescribable data will be traceless $2\times2$ matrices $\Psi$ on $H_0\cup\underline{H}_0$
    such that for sufficiently large $N$,
   \begin{align*}
     \sum_{i\le N}
     |(\slashed{\p},\ep\p_{\ub})^i\Psi|
     \lesssim 1,\
     \sum_{i\le N}|(\slashed{\p},\ep\p_{\ub})^i\partial_{u}\Psi|
     \lesssim f(u)^{-2},\
     |\partial_{u}\Psi|\sim f(u)^{-2},
   \end{align*}
   and smooth functions $\Phi,\nabla_3\Phi,\nabla_4\Phi,\zeta_A,\underline{\omega}$ prescribed on $S_{0,0}$, such that
   \begin{align*}
      \Phi|_{S_{0,0}}=1,\
      \sum_{i\le N-1}|\slashed{\p}^i(\partial_{u}\Phi,\p_{\ub}\Phi,\zeta_A,\underline{\omega})|\Big|_{S_{0,0}}
      \lesssim1,
    \end{align*}
    and $(v^\nu,\log\tau)$, $(\p_{\ub} v^\nu,\p_{\ub}\log\tau)$ for $\nu\in\{1,2,4\}$ 
    on $S_{0,0}$ such that
    \begin{align*}
     \sum_{i\le N}|(\slashed{\p},\ep\p_{\ub})^i(v^\nu,\log\tau)|\Big|_{S_{0,0}}\lesssim 1,\
     \sum_{i\le N-1}|(\slashed{\p},\ep\p_{\ub})^i(\p_{u}v^\nu,\p_{u}\log\tau
     )|\Big|_{S_{0,0}}\lesssim 1,\ |v^4|\sim1,
   \end{align*}
    and smooth bounded functions $q^0,q^A,q^4$ on $\underline{H}_0$ for $A\in\{1,2\}$, such that
    \begin{align*}
      \sum_{i\le N-1}|\slashed{\p}^i(q^0,q^A,q^4)|\lesssim1.
    \end{align*}
   We derive an initial data for each set of freely prescribable data.
   For the fluid variables, 
   note that the normalization of fluid velocity implies
   \begin{align*}
     v^3=(2v^4)^{-1}(\gamma_{AB}v^Av^B+1),
   \end{align*}
   and thus determines $v^3, \p_{\ub}v^3$ on $S_{0,0}$.
   Also, the constraint equations from Einstein--Euler system can be translated to transport equations as follows, for $\alpha', \nu'\in\{A,B,3\}$,
    \begin{align}
    \begin{split}
    &  D_4v^4=(v^4(1-p'))^{-1}(-v^{\nu'} D_{\nu'} v^4-v^4 p'v^{\nu'} D_{\nu'} L
    -p'\frac{1}{2}D_3L
     +p'v^4 (v^{\nu'} D_{\nu'} L+ D_{\nu'} v^{\nu'})),
     \\
     &  D_4L=(v^4)^{-1}(-v^{\nu'} D_{\nu'} L- D_{\nu'} v^{\nu'}- D_4v^4),
     \\
     &  D_4v^{\alpha'}
     =(v^4)^{-1}(-v^{\nu'} D_{\nu'} v^{\alpha'}-v^{\alpha'} p'v^{\nu'} D_{\nu'} L- p' (1-\delta^{3}_{\alpha'})D^{\alpha'}   L-v^{\alpha'} p'v^4 D_4L+\frac{1}{2} p'\delta^{\alpha'}_3 D_4 L),
     \end{split}
     \label{Eq3.01}
   \end{align}
   or equivalently for $\alpha, \nu\in\{A,B,4\}$,
   \begin{align}\label{Eq3.023}
    \begin{split}
     &  D_3v^3=(v^3(1-p'))^{-1}(-v^\nu D_\nu v^3-v^3 p'v^\nu D_\nu L
     -\frac{1}{2}p'D_4L
     -p'v^3(-v^\nu D_\nu L- D_\nu v^\nu)),
     \\
     &  D_3L=(v^3)^{-1}(-v^\nu D_\nu L- D_\nu v^\nu- D_3v^3),
     \\
     &  D_3v^\alpha=(v^3)^{-1}(-v^\nu D_\nu v^\alpha-v^\alpha p'v^\nu D_\nu L- p' (1-\delta^{4}_\alpha)D^\a L-v^\alpha p'v^3 D_3L+\frac{1}{2} p'\delta^{\alpha}_4 D_3 L).
     \end{split}
   \end{align}
   For definition of $L$ see Section \ref{CommutingFluids}. For the proof, see Proposition \ref{SingularEnergy} for details. Note that $\Phi,\nabla_3\Phi,\nabla_4\Phi,\zeta_A,\underline{\omega}$ on $S_{0,0}$ and $\Psi$ on $H_0\cup\underline{H}_0$ determines all Christoffel symbols of the spacetime metric, and thus defines \eqref{Eq3.023} on $S_{0,0}$. We therefore arrive at $\nabla_3v$ on $S_{0,0}$.
    
    Next, we state the transport equations on $H_0$. As for $\Phi,\chi,\zeta,\omega$, we shall need
   \begin{align}
       &\partial_{u}^2\Phi
      +\frac{1}{8}\hat{\gamma}^{AC}\hat{\gamma}^{BD} \partial_{u}\hat{\gamma}_{AB}
      \partial_{u}\hat{\gamma}_{CD}\Phi+\frac{\Phi}{2}\Omega^2(\tau+p)v_3v_3=0,
      \label{Eq3.02}
      \\
      \label{Eq3.03}
      &\slashed{\mathcal{L}}_{e_3}\mathrm{tr}\chi
      +\mathrm{tr}\underline{\chi}\mathrm{tr}\chi
      =-2K-2\mathrm{div}\zeta+2|\zeta|^2-\mathrm{Ric}_{34},
      \\
      &\slashed{\mathcal{L}}_{e_3}\zeta+\mathrm{tr}\underline{\chi}\zeta
      =\slashed{\mathrm{div}}\underline{\chi}-\nabla\mathrm{tr}\underline{\chi}+\mathrm{Ric}_{3B},
      \label{Eq3.020}
    \end{align}
    where, given $b|_{\underline{H}_0}=0$, \eqref{Eq3.02} is equivalent to the constraint equation,
    \begin{align*}
     \slashed{\mathcal{L}}_{e_3}\mathrm{tr}\underline{\chi}
     =-\frac{1}{2}(\mathrm{tr}\underline{\chi})^2-|\hat{\underline{\chi}}|^2-\mathrm{Ric}_{33}.
   \end{align*}
   As for $(v,L)$, we need
   \begin{align}\label{Eq3.022}
   \begin{split}
     &\partial_{u}L=((1-p')v^3v^4+\frac{3}{2}p')^{-1}(v^4v^A\eta_A+\underline{\chi}_{AB}v^Av^B
     -v^3v^4\mathrm{tr}\underline{\chi})+q^0,
     \\
     &\partial_{u}v^A = -2\underline{\chi}^A{}_{B}v^B-2\eta^Av^4-v^Ap'\p_u L+q^A,
     \\
     &\partial_{u}v^3 = -2\underline{\omega} v^3+\eta_Av^A-v^3\p_u L
     -v^3\mathrm{tr}\underline{\chi}+q^3,
     \\
     &v^4=(2v^3)^{-1}(\gamma_{AB}v^Av^B+1),
     \end{split}
   \end{align}
   where $\underline{\chi}$ is determined by $\Phi$ and its tangential derivatives,
   \begin{align*}
     \Omega^{-1}\underline{\chi}_{AB}
     =\Phi^2\partial_{u}\hat{\gamma}_{AB}
      +\frac{2}{\Phi}\partial_{u}\Phi \gamma_{AB}.
   \end{align*}
    We prove that the ODEs \eqref{Eq3.02} 
    - \eqref{Eq3.022} give rise to a bounded solution $\Phi$ and $(\slashed{v},\log\tau)$. Since $\Omega=1,\ b=0,\ \omega=0$ on $\underline{H}_0$, we have $e_3\Omega=0$ and thus $\zeta=\eta=-\underline{\eta},\ \underline{\omega}=0$, and
    \begin{align*}
     \Omega\hat{\chi}_{AB}=\frac{1}{2}\Phi^2\partial_{\ub}\hat{\gamma}_{AB},\
     \Omega\mathrm{tr}\chi=\frac{2}{\Phi}\partial_{\ub}\Phi.
   \end{align*}
    Thus, a priori estimates
    \begin{align*}
      c\le  |\Phi|\le C,\
      |\zeta|,|\mathrm{tr}\chi|
      ,|(v,\log\tau)|\le C,\
      |\p_{u}\Phi|\le Cf(u)^{-2},
    \end{align*}
    implies that $|\mathrm{tr}\underline{\chi}|\le C^2f(u)^{-2}$,\ $|\hat{\underline{\chi}}|\sim f(u)^{-2}$, $|\hat{\chi}|\le C^2$ and thus
    \begin{align*}
      &|(v,\log\tau)(u,0)
      -(v,\log\tau)(0,0)|
      \le C^2\int_0^{u} f(u')^{-2}
      \le C^2\epsilon^2,
      \\
      &|\partial_u\Phi(u,0)-\partial_u\Phi(0,0)|
      \le C^3\int_0^{u} f(u')^{-4}
      \le \epsilon^2 C^3 f(u)^{-2},
      \\
      &|e^{\int^u_0\mathrm{tr}\underline{\chi}}\zeta(u,0)-\zeta(0,0)|
      \le C e^{C\int_0^{u} f(u')^{-2}}\int_0^{u} f(u')^{-2}
      \le Ce^\epsilon\epsilon^2
      \\
      &|e^{\int^u_0\mathrm{tr}\underline{\chi}}\mathrm{tr}\chi(u,0)
      -\mathrm{tr}\chi(u,0)|
      \le C e^{C\int_0^{u} f(u')^{-2}}\int_0^{u} f(u')^{-2}
      \le Ce^\epsilon\epsilon^2.
    \end{align*}
    The a priori estimates are improved for $C$ sufficiently large as compared to data on $S_{0,0}$ and time $\epsilon$ accordingly small. Thus, bootstrap arguments implies existence of a solution satisfying the a priori estimates. In particular, Christoffel symbols $\psi,\psi_H$ are bounded. Also, \eqref{Eq3.022} further implies 
    \begin{align*}
      |(\partial_{u}v,\partial_{u}\log\tau)|\lesssim f(u)^{-2}.
    \end{align*}
    In fact, if we further impose
    \begin{align}\label{Eq2.7001}
      &\sum_{i\le N}|(\slashed{\p},\ep\p_{\ub})^i\partial_u\partial_{u}\Psi|
     \lesssim u^{-1}f(u)^{-2},\
     |\partial_{u}\partial_u\Psi|\sim u^{-1}f(u)^{-2},
    \end{align}
    and refer to \eqref{Eq2.4}, we will arrive at
    \begin{align*}
       \underline{\alpha}_{AB}(u,0)
        =_s \nabla_3\hat{\underline{\chi}}_{AB}(u,0)
       +\underline{\chi}\hat{\underline{\chi}}(u,0)
       +\underline{\omega}\hat{\underline{\chi}}(u,0).
    \end{align*}
    Therefore, by assigning $\hat{\gamma}$ as a diagonal matrix, we have
    \begin{align*}
       &|\underline{\alpha}_{AA}|(u,0)\sim |\p_u\hat{\underline{\chi}}_{AA}|(u,0)
       \sim |\Phi^2\p_u\p_u\hat{\gamma}_{AA}|(u,0)
        \sim u^{-1}f(u)^{-2},\ A\in\{1,2\},
        \\
        &|\underline{\alpha}_{12}|(u,0)  \lesssim |\underline{\chi}\hat{\underline{\chi}}_{12}|(u,0)
        +|\underline{\omega}\hat{\underline{\chi}}_{12}|(u,0)
       \sim f(u)^{-4}.
    \end{align*}
    Thus, the resulting family of initial data satisfies \eqref{CurvatureBlowup} in addition.
    This observation will be used in the proof of Remark \ref{Rem1.1} in Section \ref{Section6}.
    
   Next, we check that the bounds in Theorem \ref{SmoothExist} hold.
   Note that
      \begin{align*}
     & \partial_\alpha v^\beta
     =D_\alpha v^\beta +g(D_\alpha e^\beta,e_\gamma) v^\gamma.
   \end{align*}
   and the only other terms in \eqref{Eq3.01} involving singular Christoffel symbols $\psi_{\underline{H}}$ are
  \begin{align*}
    &D_Av^B=e_Av^B-g(\nabla_Ae^B,e_C)v^C+\chi^B{}_Av^4+\underline{\chi}^B{}_Av^3,
    \\
    &D_Av^4=e_Av^4+\frac{1}{2}\underline{\chi}_{AB}v^B.
  \end{align*}
   According to \eqref{Eq3.022}, 
   the singular Christoffel symbols $\psi_{\underline{H}}\in\{\mathrm{tr}\underline{\chi},\hat{\underline{\chi}}, \underline{\omega}\}$ arising in \eqref{Eq3.01} cancel up. 
   Thus, $(\partial_{\ub}v,\partial_{\ub}\log\tau)$ derived from \eqref{Eq3.022} is bounded on $\underline{H}_0$. Furthermore, as for $\slashed{v}$,
   \begin{align*}
     &\nabla_4\slashed{v}^A=e_4\slashed{v}^A-g(\nabla_4e^A,e_\gamma)\slashed{v}^\gamma,\
     \nabla_4\slashed{v}^3=e_4\slashed{v}^3,\
     \nabla_4\slashed{v}^4=e_4\slashed{v}^4,
   \end{align*}
   implies that $(\nabla_4\slashed{v},\nabla_4\log\tau)$ is bounded on $\Hb_0$ as a consequence of boundedness of $\p_{\ub}r$.
    
    For higher order derivatives, we commute \eqref{Eq3.01} and initial data \eqref{Eq3.022} inductively with $\partial_{\ub}$ to derive $\partial_{\ub}^ir$.
    Following \cite{luk2017weak}, commuting the transport equations \eqref{Eq3.02},
     \eqref{Eq3.03}, \eqref{Eq3.01} with $ Z ^i$
     imply the bounds simultaneously for $ Z ^i (v,\log\tau)$ and $ Z ^i(\partial_{\nu}\Phi,\zeta,\mathrm{tr}\chi)$.
     
     It is worth noting that an extra derivative of the metric component $b$ is required for estimating $\p_uv^A$ as given rise by the formula
     \begin{align*}
       &D_3v^A=(\partial_u+b^B\partial_B)v^A-(\chib^{A}{}_{B}-D_Ab^B\Om^{-1})v^B.
     \end{align*}
     This will not be a problem by applying \eqref{Eq2.17} and utilizing sufficiently large number of derivatives prescribed for $\eta,\ \etab$ on $H_0\cup\Hb_0$ in the assumption.
    
    By the flexibility in choosing the traceless matrices $\Psi$ and the fluid data $q^0,q^A$, we conclude that there exists a family of admissible initial data satisfying all the assumptions in Theorem \ref{SmoothExist} with a singularity present. 
   
   Similarly for solving $\Phi$ and $(v,\log\tau)$ with the initial data for the spacetime metric on $H_0$ satisfying the bounds
    \begin{align*}
     \sum_{i\le N}| (\slashed{\p},\ep\p_{\ub}) ^i\Psi|\lesssim 1,\
     \sum_{i\le N-1}| (\slashed{\p},\ep\p_{\ub})^i\p_{\ub}\Psi|\lesssim 1,\
     |\p_{\ub}\Psi|\sim 1,
    \end{align*}
    and
    \begin{align*}
      \Phi|_{S_{0,0}}=1,\
      \sum_{i\le N-1}|(\slashed{\p},\ep\p_{\ub})^i(\partial_{u}\Phi,\p_{\ub}\Phi,\zeta_A,\underline{\omega})|\Big|_{S_{0,0}}
      \lesssim1,
    \end{align*}
    and
    \begin{align*}
     \sum_{i\le N}|\slashed{\p}^i(v,\log\tau)|\Big|_{S_{0,0}}\lesssim 1,
     \sum_{i\le N-1}|\slashed{\p}^i(\p_{\ub}v,\p_{\ub}\log\tau,
     \partial_{u}v,\partial_{u}\log\tau)|\Big|_{S_{0,0}}\lesssim 1.
   \end{align*}
   Here $\Omega=1,\omega=0$ on $H_0$ implies that $\zeta
   =\eta
   =-\underline{\eta}$. In order to control all Christoffel symbols, we apply similar proof to an extra equation for $\underline{\omega}$,
   \begin{align*}
     &\slashed{\mathcal{L}}_{e_4}\underline{\omega}
      =\frac{3}{2}|\zeta|^2-\frac{1}{2}K+\frac{1}{4}\hat{\chi}\cdot\hat{\underline{\chi}}
      -\frac{1}{8}\mathrm{tr}\chi\mathrm{tr}\underline{\chi}+\frac{1}{4}R+\frac{3}{4}\mathrm{Ric}_{34}.
   \end{align*}
   As for checking the conditions in Theorem \ref{SmoothExist}, 
   we read from \eqref{Eq3.01} that $(\partial_{u}v,\partial_{u}\log\tau)$ remain bounded on $H_0$. In particular, there is no singularity on $H_0$.
    \end{proof}
    
    \section{Reduction of Theorem 2 to a priori estimates}\label{Section3.1}

   We now turn to the proof of Theorem \ref{SmoothExist}. We shall need a local existence result for the characteristic initial value problem, see for example \cite{doi:10.1098/rspa.1990.0009}.
   \begin{prop}\label{Prop3.01}
     Assume that the initial data for the characteristic initial value problem satisfy the assumptions of Theorem \ref{SmoothExist} with $\epsilon$ sufficiently small. Then there exists an Einsteinian development $(U,g,v,\tau)$ satisfying the Einstein--Euler system in a neighborhood $U$ of $S_{0,0}$, obeying the initial conditions in $(H_u\cup H_{\underline{u}})\cap U$. 
   \end{prop}
   By solving the eikonal equation in $U$, there exists a smooth local coordinate chart $(u,\underline{u},\theta)$ within $(H_0\cap\{0\le \underline{u}<\underline{u}_*\})\cap(\underline{H}_0\cap\{0\le u<u_*\})\subset U$ as a neighborhood of $S_{0,0}$.
   
   We shall also need a local existence result for the Cauchy initial value problem, see for example \cite{doi:10.1142/S0217732314502058}.
   \begin{prop}\label{Prop3.02}
   Let $s>\frac{3}{2}+2$, and let $(\Sigma,g_0,K,\tau,v)$ be an initial data set prescribed on a compact spacelike hypersurface $\Sigma$ for the Cauchy initial value problem for the Einstein--Euler system, such that $g_0\in H^{s+1}$, $ K,\tau,v\in H^s$, and $p'<1$. Then there exists an Einsteinian development $(\Sigma\times[0,T],g,v,\tau)$ satisfying the Einstein--Euler system, obeying the initial conditions, such that $g\in C^0([0,T],H^{s+1}(\Sigma))\cap C^1([0,T],H^{s}(\Sigma))\cap C^2([0,T],H^{s-1}(\Sigma))$ and $v,\tau\in C^0([0,T],H^{s}(\Sigma))\cap C^1([0,T],H^{s-1}(\Sigma))$.
   \end{prop}
   By solving the eikonal equation in $U$, there exists a smooth local coordinate chart $(u,\underline{u},\theta)$ within the determining future neighborhood of $\Sigma$.
   
   Next, we will establish a priori estimates for the geometric  and fluid quantities.
   \begin{thm}
     \label{AprioriEstimates}
     Assume that the initial data for the characteristic initial value problem satisfy the assumptions of Theorem \ref{SmoothExist} with $\epsilon$ sufficiently small. We also assume that the solution exists in $D_{u_*,\underline{u}_*}\cap\{u+\underline{u}< \delta\}$ where $0<\delta\le u_*+\underline{u}_*$. Then
     \begin{align}
       &\mathcal{O}_{ini}+\mathcal{R}_{ini}+\mathcal{E}_{ini}\le\tilde{D},
     \end{align}
     for some $\tilde{D}$ depending only on $D$ and $d$. Moreover, for $([0,u_*)\times[0,\underline{u}_*)\times\mathbb{T}^2)\cap \{u+\underline{u}=\delta\}$ there exists $B$ depending only on $D$ and $d$ such that
     \begin{equation}\label{Theorem4(2)}
       \mathcal{O}_{3,2}+\tilde{\mathcal{O}}_{4,2}+\mathcal{R}_3
       +\mathcal{E}_{2,\infty}+\tilde{\mathcal{E}}_{4,2}+\tilde{\mathcal{F}}_{4,2}
     \le B.
     \end{equation}
   \end{thm}
   \begin{proof}
     [Proof of Theorem \ref{SmoothExist}]
     Denote by $U_\delta=D_{u_*,\underline{u}_*}\cap\{u+\underline{u}< \delta\}$ and $\Sigma_\delta=([0,u_*)\times[0,\underline{u}_*)\times\mathbb{T}^2)\cap \{u+\underline{u}=\delta\}$.
     We note that estimates \eqref{Theorem4(2)} implies the assumptions of Theorem \ref{SmoothExist} and thus Proposition \ref{Prop3.01} on every null hypersurface $H_u$ and $\underline{H}_{\underline{u}}$, where $0<u,\underline{u}<\delta$. Also, equations \eqref{Eq2.17} - \eqref{Eq2.22} and equations \eqref{Eq2.1} - \eqref{Eq2.16} imply that the Christoffel symbols and curvature components in estimates \eqref{Theorem4(2)} controls $\partial_tg$ and $\partial_t^2g$, respectively. Thus, estimates \eqref{Eq3.1} also implies the assumptions of Proposition \ref{Prop3.02} on every Cauchy hypersurface $\Sigma_{\delta_0}$, where $0<\delta_0<\delta$.
     
     We prove by contradiction that the domain of existence contains any compact subset $A_\delta=([0,u_*)\times[0,\underline{u}_*)\times\mathbb{T}^2)\cap \{u+\underline{u}\le \delta\}$. To start, Proposition \ref{Prop3.01} implies existence in some $A_\delta$. Suppose there exists a supremum $\delta_*<u_*+\underline{u}_*$ of all permitted $\delta$, then Theorem \ref{AprioriEstimates} ensures that an initial data set can be defined on the compact spacelike hypersurface $\Sigma_*=([0,u_*)\times[0,\underline{u}_*)\times\mathbb{T}^2)\cap \{u+\underline{u}=\delta_*\}$ as a strong $H^s$ limit of the solution. Thus, Proposition \ref{Prop3.02} applied to $\Sigma_*$ derives existence in some determining future neighborhood of $\Sigma_*$. Combining with existence within a neighborhood of $S_{0,\delta_*}$ by Proposition \ref{Prop3.01} applied to $S_{0,\delta_*}$ if $\delta_*<\underline{u}_*$ and a neighborhood of $S_{\delta_*,0}$ by applying to $S_{\delta_*,0}$ if $\delta_*<u_*$, we derive a development containing $A_{\delta_0}$ for some $\delta_0>\delta_*$. This contradiction implies that $\delta_*=u_*+\underline{u}_*$ and thus Theorem \ref{SmoothExist} is proved.
   \end{proof}
    In Section \ref{Section4} and Section \ref{Section5}, we prove Theorem \ref{AprioriEstimates} for the case of $\delta=u_*+\underline{u}_*$, where $U_\delta=D_{u_*,\underline{u}_*}$.
     The proof also applies to general $\delta$ with slight modifications, see Section \ref{Section6}. We prove the estimates for the spacetime metric on every compact region $[0,u_*']\times[0,\underline{u}_*']\times\mathbb{T}^2$
    where $u_*'\le u_*,\ \underline{u}_*'\le \underline{u}_*$, stated as follows,
   \begin{align}
   \begin{split}
     &\mathcal{O}_{3,2}\le C(\mathcal{O}_{ini})+\epsilon^{\frac{1}{2}} C(\Delta_1,\mathcal{E}_{2,\infty})(1+\mathcal{R}_3+\tilde{\mathcal{O}}_{4,2}
       +\tilde{\mathcal{E}}_{4,2}),
     \\
     &\tilde{\mathcal{O}}_{4,2}
     \le C(\mathcal{O}_{ini})(1+\mathcal{R}_3
     + C(\Delta_1,\mathcal{E}_{2,\infty})
     (\epsilon^{\frac{1}{2}}(1+\mathcal{R}_3+\tilde{\mathcal{O}}_{4,2}^2)
     +\tilde{\mathcal{E}}_{4,2}+\tilde{\mathcal{F}}_{4,2})),
     \\
     &\mathcal{R}_3
     \le C(\mathcal{O}_{ini},\mathcal{R}_{ini})
     (1+\epsilon^{\frac{1}{2}} C(\Delta_1,\mathcal{E}_{2,\infty})(1+\mathcal{R}_3+\tilde{\mathcal{O}}_{4,2}
     +\tilde{\mathcal{E}}_{4,2}+\tilde{\mathcal{F}}_{4,2})),
     \end{split}
   \label{Eq3.1}
   \end{align}
   and for fluid variables
   \begin{align}
     \begin{split}
     &(1-(p')^{\frac{1}{3}})\tilde{\mathcal{E}}_{4,2}^2
     \le C(\mathcal{E}_{ini})(1
     + \epsilon C(\Delta_1,\mathcal{E}_{2,\infty})),
     \\
     &\tilde{\mathcal{F}}_{4,2} \le \ep^{\f{1}{2}}C(\Delta_1, \mathcal{E}_{2,\infty})(1+\tilde{\mathcal{E}}_{4,2}+\mathcal{O}_{3,2}),
     \\
     &\mathcal{E}_{2,\infty}\le C(\Delta_1)(1+\tilde{\mathcal{E}}_{4,2}^3),
     \end{split}
     \label{Eq3.001}
   \end{align}
   under the bootstrap assumption
   \begin{align}
     \begin{split}
     &\mathcal{O}_{1,\infty}+\mathcal{O}_{2,4}+\mathcal{O}_{3,2}\le \Delta_1,\
     \mathcal{R}_3 \le \Delta_1,\
     \tilde{\mathcal{O}}_{4,2}\le 2\Delta_1^{2},\
     \tilde{\mathcal{E}}_{4,2}\le \Delta_1,
     \end{split}
     \label{Eq3.2}
   \end{align}
   with $\Delta_{1}=\Delta_{1}(\mathcal{O}_{ini},\mathcal{R}_{ini},\mathcal{E}_{ini})$ to be determined. According to the proof as in Section \ref{Section4} and Section \ref{Section5}, the constants of the form $C(\Lambda)$ are polynomials of $\Lambda$ of order less than $10$. To prove \eqref{Theorem4(2)},
   we take some constant $M\gg1$ such that the constants in \eqref{Eq3.1}, \eqref{Eq3.001} of the form $C(\mathcal{O}_{ini},\mathcal{R}_{ini},\mathcal{E}_{ini})$ are dominated by $M$ and the constants of the form $C(\Lambda)$ are controlled by $M\Lambda^{10}$. We also take $\Delta_{1}= \max\{12,2(1-(p')^{\frac{1}{3}})^{-\frac{1}{2}}\}M^5$. Then bootstrap arguments imply that, with data dominated by $M$ and $\epsilon=(16\Delta_1^{900})^{-1}$, by deducing step by step,
   \begin{align}\label{Eq6.1}
   \begin{split}
     &\mathcal{E}_{2,\infty}\le \Delta_1^{10}(1+\Delta_1^{3})\le 2\Delta_1^{13},
     \\
     &\tilde{\mathcal{E}}_{4,2}\le (1-(p')^{\frac{1}{3}})^{-\frac{1}{2}}( M+\epsilon^{\frac{1}{2}}
      2\Delta_1^{130})\le 2(1-(p')^{\frac{1}{3}})^{-\frac{1}{2}}M
      \le\frac{1}{2}\Delta_1,
     \\
     &\tilde{\mathcal{F}}_{4,2}\le 3\ep^{\f{1}{2}}\Delta_1^{150},
     \\
     &\tilde{\mathcal{O}}_{4,2}\le M(1+\Delta_1)+\epsilon^{\frac{1}{2}}\Delta_1^{300}
     \le 2M\Delta_1 \le \Delta_1^2,
     \\
     &\mathcal{O}_{3,2},\mathcal{R}_3\le M+\epsilon^{\frac{1}{2}}\Delta_1^{400}\le 2M.
     \end{split}
   \end{align}
   Apply Sobolev embedding \cite[Lemma 5.1, Lemma 5.2]{Christodoulou:2008nj} with bounded isoperimetric constant on every $S_{u,\underline{u}}$, as seen from Proposition \ref{Prop3.7}, this implies
   \begin{align}
     \mathcal{O}_{1,\infty}+\mathcal{O}_{2,4}+\mathcal{O}_{3,2}\le \frac{1}{2}\Delta_1.
   \end{align}
   Therefore, these estimates imply \eqref{Theorem4(2)} in the designated area.
   
   It thus remains to prove \eqref{Eq3.1} and \eqref{Eq3.001}. In the follows, \eqref{Eq3.1} will be proven in Section \ref{Section4}, and \eqref{Eq3.001} will be proven in Section \ref{Section5}.
   
   \section{Estimates for the Spacetime Metric}\label{Section4}
   

   \subsection{General Propositions}
   We prove general propositions for obtaining bounds from the covariant null transport equations. As in \cite{luk2017weak},
   \begin{prop}\label{PropTransport}
     There exists $\epsilon_0=\epsilon_0(\Delta_1)>0$ such that for all $\epsilon\le\epsilon_0$ and for every $2\le p\le \infty$, we have
     \begin{align*}
       \sum_{i\le a}\| Z ^i\phi\|_{L^p(S_{u,\underline{u}})}
       &\le C(\sum_{i\le a}\| Z ^i\phi\|_{L^p(S_{u,\underline{u}'})}
       +\sum_{i\le a}\int_{\underline{u}'}^{\underline{u}}\|\nabla_4 Z ^i\phi\|_{L^p(S_{u,\underline{u}''})}
       d\underline{u}''),
       \\
       \sum_{i\le a}\| Z ^i\phi\|_{L^p(S_{u,\underline{u}})}
       &\le C(\sum_{i\le a}\| Z ^i\phi\|_{L^p(S_{u',\underline{u}})}
       +\sum_{i\le a}\int_{u'}^{u}\|\nabla_3 Z ^i\phi\|_{L^p(S_{u'',\underline{u}})}
       du''),
     \end{align*}
     for any tensor $\phi$ tangential to the torus $S_{u,\underline{u}}$.
   \end{prop}
   \begin{proof}
     According to \cite{luk2017weak}, for surface measure $\mathrm{d}\sigma=\sqrt{|\det \gamma|}\mathrm{d}\theta^A\wedge\mathrm{d}\theta^B$ on $S_{u,\underline{u}}$, construction of $\nabla$ as the projected metric connection implies that $\nabla\gamma=0$. Therefore
     \begin{align*}
       &e_4\int_{S_{u,\underline{u}}}f
       =\int_{S_{u,\underline{u}}}\nabla_4f\mathrm{d}\sigma.
     \end{align*}
     Hence, by taking $f=| Z ^i\phi|^p$, we have
     \begin{align*}
       &\| Z ^i\phi\|_{L^p(S_{u,\underline{u}})}^p
       =\| Z ^i\phi\|_{L^p(S_{u,\underline{u}'})}^p
       +\int_{\underline{u}'}^{\underline{u}}
       \int_{S_{u,\underline{u}''}}p| Z ^i\phi|^{p-2}
       Z ^i\phi\nabla_4 Z ^i\phi\mathrm{d}\sigma
       d\underline{u}''.
     \end{align*}
     Thus proposition holds by H\"older's inequality and Gr\"onwall's inequality. Similarly we have
     \begin{align*}
       &e_3\int_{S_{u,\underline{u}}}f
       =\int_{S_{u,\underline{u}}}\nabla_3f
       \mathrm{d}\sigma,
     \end{align*}
     Hence, by taking $f=| Z ^i\phi|^p$, the second inequality follows similarly.
   \end{proof}
   \begin{prop}\label{PropElliptic}
     Let $\phi$ be a symmetric $r$ covariant tensor on a $2$-torus $(\mathbb{T}^2,\gamma)$. Recall that
     \begin{align*}
       &(\mathrm{div}\phi)_{A_1...A_r}=\nabla^B\phi_{BA_1...A_r},\
       (\mathrm{curl}\phi)_{A_1...A_r}=\slashed{\epsilon}^{BC}\nabla_B\phi_{CA_1...A_r}.
     \end{align*}
     Suppose $\phi$ satisfies
     \begin{align*}
        &\mathrm{div}\phi=f,\quad
        \mathrm{curl}\phi=g.
     \end{align*}
     Suppose also that
     \begin{align*}
        &E=\sum_{i\le2}\| Z ^i (\psi_0,\psi,\psi_H)\|_{L^\infty(S)}
        +\sum_{i\le3}\| Z ^i (\psi_0,\psi,\psi_H)\|_{L^4(S)}
        +\sum_{i\le4}\| Z ^i (\psi_0,\psi,\psi_H)\|_{L^2(S)}<\infty.
     \end{align*}
     There exists a constant $C_E$ depending only on $E$ such that
     \begin{align*}
       &
        \sum_{i\le3, j\le1}\| Z ^i\slashed{\nabla}^j \phi\|_{L^2(S)}
       \le C_E\sum_{i\le3}\| Z ^i (f,g
       ,\phi)\|_{L^2(S)}.
     \end{align*}
   \end{prop}
   \begin{proof}
     Upon integration by parts and $\nabla\gamma=0$, it follows that
     \begin{align*}
       \int_{S_{u,\underline{u}}} \varphi\nabla_C \psi\mathrm{d}\sigma
       =\int_{S_{u,\underline{u}}} \nabla_C \varphi \psi\mathrm{d}\sigma.
     \end{align*}
     By algebraic transformation,
     \begin{align*}
       &\int_{S_{u,\underline{u}}} |f|^2+|g|^2\mathrm{d}\sigma
       =\int_{S_{u,\underline{u}}} |\slashed{\nabla}\phi|^2\mathrm{d}\sigma
     \end{align*}
     For higher order derivatives, commuting the formula of $f,g$ by $ Z ^i$, then
     \begin{align*}
       &[\gamma^{AB}\nabla_A, Z ^i]\phi_{B\cdot}
       =_s\sum_{i=i_1+i_2} Z ^{i_1}\gamma^{AB}Z ^{i_2}\nabla_A\phi_{B\cdot}
       +\gamma^{AB}[\nabla_A, Z ^i]\phi_{B\cdot}\ ,
       \\
       &[\slashed{\epsilon}^{AB}\nabla_A, Z ^i]\phi_{B\cdot}
       =_s\sum_{i=i_1+i_2} Z ^{i_1}\slashed{\epsilon}^{AB}Z ^{i_2}\nabla_A\phi_{B\cdot}
       +\slashed{\epsilon}^{AB}[\nabla_A, Z ^i]\phi_{B\cdot}\ .
     \end{align*}
     Apply Proposition \ref{Prop2.2} to commuted $ Z ^if, Z ^ig$, then up to a constant $C_E$ depending on $E$,
     \begin{align*}
       &|\int_{S_{u,\underline{u}}} | Z ^if|^2+| Z ^ig|^2\mathrm{d}\sigma
       -\int_{S_{u,\underline{u}}} | Z ^i\slashed{\nabla}\phi|^2\mathrm{d}\sigma|
       \\
       \le&\ C_E \sum_{i=i_1+i_2+i_3, j, k\le1}\| Z ^{i_1}\slashed{\nabla}^j \phi  Z ^{i_2}\phi Z ^{i_3}\slashed{\nabla}^k(\psi_0,\psi,\psi_H)\|_{L^1(S)}
       \\
       \le&\ C_E (\sum_{i=i_1+i_2, j\le1}\| Z ^{i_1}\slashed{\nabla}^j \phi\|_{L^2(S)} \| Z ^{i_2}\phi\|_{L^2(S)}
       \\
       &+\sum_{i=i_1+i_2+i_3, i_3\ge2, j,k\le1}\| Z ^{i_1}\slashed{\nabla}^j \phi\|_{L^2(S)} \| Z ^{i_2}\phi\|_{L^\infty(S)} \| Z ^{i_3}\slashed{\nabla}^k(\psi_0,\psi,\psi_H)\|_{L^2(S)})
       \\
       \le&\
       \frac{1}{100}\sum_{i_1\le i, j\le1}\| Z ^{i_1}\slashed{\nabla}^j \phi\|_{L^2(S)}^2
       +C_E (\sum_{i_2\le i-1} \| Z ^{i_2}\phi\|_{L^2(S)}^2
       \\
       &+\sum_{i_2+i_3\le i,i_3\ge2, j,k\le1}\| Z ^{i_2}\phi\|_{L^\infty(S)}^2 \| Z ^{i_3}\slashed{\nabla}^k(\psi_0,\psi,\psi_H)\|_{L^2(S)}^2).
     \end{align*}
     Note that curvature on $\mathbb{T}^2$ is bounded by $E$, Sobolev embedding on $\mathbb{T}^2$ applied to $ Z ^i\slashed{\nabla}^j \phi$ concludes the proof.
   \end{proof}
   Also, we refer to a standard Sobolev embedding result.
   \begin{prop}\label{PropSobolev}
     Let $\phi$ be a function on $D_{u,\underline{u}}$, $0<u,\underline{u}\le\epsilon$, there exists a constant $C(\Delta_1) $ depending on $\Delta_1$ such that
     \begin{align*}
       \|\phi\|_{L^\infty(D_{u,\underline{u}})}
     &\le C(\Delta_1) \epsilon^{-\frac{1}{2}}\sum_{|\alpha|\le2}\| Z ^\alpha\phi\|_{L_{u}^\infty L_{\underline{u}}^2 L^2(S)}.
     \end{align*}
   \end{prop}
   \begin{proof}
     Apply the interpolation theorem by Lions and Peetre to Sobolev embedding in $\mathbb{T}^2$, see for example \cite[(5.35)]{Christodoulou:2008nj}, where the constant depend on $\Delta_1$ which bounds the volume form,
     \begin{align*}
       &\|\phi\|_{L^\infty(D_{u,\underline{u}})}\le C(\Delta_1,\lambda) \sum_{|\alpha|\le2}\|\slashed{\nabla}^\alpha\phi\|_{L_{u}^\infty L_{\underline{u}}^\infty L^{1+\lambda}(S)},
     \end{align*}
     and
     \begin{align*}
       &\|\phi\|_{L^\infty(D_{u,\underline{u}})}\le C(\Delta_1,\lambda)\epsilon^{-\frac{1}{1+\lambda}} \sum_{|\alpha|\le1,|\beta|\le1}\|(\epsilon\nabla_4)^\alpha\slashed{\nabla}^\beta\phi\|_{L_{u}^\infty L_{\underline{u}}^{1+\lambda} L^{2+\lambda}(S)}.
     \end{align*}
     Take specifically $\lambda=\delta^2$ and apply H\"older's inequality, then
     \begin{align*}
       \|\phi\|_{L^\infty(D_{u,\underline{u}})}
       &\le C(\Delta_1,\delta)\epsilon^{-\frac{1}{1+\delta}}  \sum_{|\alpha|\le2}\| Z ^\alpha\phi\|_{L_{u}^\infty L_{\underline{u}}^{1+\delta} L^2(S)}.
     \end{align*}
     Apply H\"older's inequality again, we arrive at the inequality in the proposition.
   \end{proof}

   \subsection{Schematic Equations}\label{SchematicEquations}
   We derive the schematic equations, following \cite{luk2017weak}, from the equations in Appendix \ref{Equations} with the help of Proposition \ref{Prop2.2}. Recall that by $=_s$ we mean by suppressing all coefficients of size $C(\Delta_1,\mathcal{E}_{2,\infty})$, which are bounded by bootstrap assumptions. In the following, denote by $\Gamma$ a general Ricci coefficient perpendicular to the torus, then from \eqref{EqA.01} - \eqref{EqA.07} schematically we have
   \begin{align}\label{Eq4.0}
     \Gamma_{\alpha\cdot}^\beta
     =D_\alpha e^\beta-\nabla_\alpha (e^A\delta_A{}^\beta)
     =_s\psi+\psi_H+\psi_{\underline{H}}. &
   \end{align}
   Following \cite{Taylor2017}, denote by $\slashed{\mathrm{Ric}}_{\cdot \cdot}=\slashed{\mathrm{Ric}}_{AB}$ the restricted $(0,2)$-tensors on 
   the tori $S_{u,\underline{u}}$. Similarly, let $\slashed{\mathrm{Ric}}_{3\cdot},\slashed{\mathrm{Ric}}_{4,\cdot}$ be the $(0,1)$-tensors and $\slashed{\mathrm{Ric}}_{33},\slashed{\mathrm{Ric}}_{44},\slashed{\mathrm{Ric}}_{34}$ be the functions on tori $S_{u,\underline{u}}$. 
   We denote them by $\slashed{\mathrm{Ric}}$.
   Equation \eqref{Eq1} implies, for $ i\le4$, schematically
   \begin{align}
     &
      Z ^i \slashed{\mathrm{Ric}}=_s \sum_{ i_1\le i} Z ^{i_1}r
     +\sum_{ i_2\le i}  Z ^{i_2}\psi_0,
     \label{Eq3.3.1}
   \end{align}
   and for $ i\le 3$, schematically
   \begin{align}
      Z ^i \nabla\slashed{\mathrm{Ric}}=_s \sum_{ i_1+ i_2\le i}\nabla Z ^{i_1}r\cdot  Z ^{i_2}(r+\psi_0)
     +\sum_{ i_1+ i_2\le i} \nabla Z ^{i_1}\psi_0\cdot  Z ^{i_2}r.
     \label{Eq3.3.2}
   \end{align}
   As for metric components $\psi_0$, we derive from Equations \eqref{Eq2.17} - \eqref{Eq2.22} as
   \begin{align}
     \begin{split}
        &\nabla_4 Z ^i b
        =_s\sum_{ i_1+ i_2+ i_3+ i_4\le  i}Z ^{i_1}\psi^{i_2} Z ^{i_3}\psi  Z ^{i_4}\psi_H,
     \end{split}
     \label{Eq3.3}
     \\
     \begin{split}
        &\nabla_4 Z ^i \gamma
        =_s\sum_{ i_1+ i_2+ i_3+ i_4\le  i}Z ^{i_1}\psi^{i_2} Z ^{i_3}\psi  Z ^{i_4}\psi_H,
     \end{split}
     \label{Eq3.4}
     \\
     \begin{split}
        &\slashed{\nabla} Z ^i\log\Omega
        =_s
        \sum_{ i_1+ i_2+ i_3+ i_4\le  i}Z ^{i_1}\psi^{i_2} Z ^{i_3}\psi  Z ^{i_4}\psi_H,
     \end{split}
     \label{Eq3.5}
     \\
     \begin{split}
        &\nabla_4 Z ^i\log\Omega
        =_s
        \sum_{ i_1+ i_2+ i_3+ i_4\le  i}Z ^{i_1}\psi^{i_2} Z ^{i_3}\psi  Z ^{i_4}\psi_H,
     \end{split}
     \label{Eq3.5.1}
     \\
     \begin{split}
        &\nabla_3 Z ^i\log\Omega
        =_s 
        \sum_{ i_1+ i_2+ i_3+ i_4\le  i}Z ^{i_1}\psi^{i_2} Z ^{i_3}(\psi+\psi_H)  Z ^{i_4}\psi_{\Hb}.
     \end{split}
     \label{Eq3.6}
   \end{align}
   For Ricci coefficients, we have the following schematic equations:
   \begin{align}
     \nabla_4 Z ^i\eta
     =_s&\sum_{ i_1+ i_2+ i_3+ i_4\le  i}  Z ^{i_1}\psi^{i_2} Z ^{i_3}\beta
     +
      Z ^{i_1}\psi^{i_2} Z ^{i_3}\psi  Z ^{i_4}\psi_H
     + Z ^i\slashed{\mathrm{Ric}},
     \label{Eq3.7}
     \\
     \nabla_3 Z ^i\underline{\eta}
     =_s&\sum_{ i_1+ i_2+ i_3+ i_4\le  i}  Z ^{i_1}\psi^{i_2} Z ^{i_3}\underline{\beta}
     +
      Z ^{i_1}\psi^{i_2} Z ^{i_3}\psi  Z ^{i_4}\psi_{\underline{H}}
     + Z ^i\slashed{\mathrm{Ric}},
     \label{Eq3.8}
     \\
     \nabla_4 Z ^i\psi_{\underline{H}}
     =_s& Z ^iK+ Z ^i\slashed{\nabla}\underline{\eta}+\sum_{i_1+i_2=i}  Z ^{i_1}\psi\cdot  Z ^{i_2}\psi+\sum_{ i_1+ i_2\le  i} Z ^{i_1}(\psi+\psi_H) Z ^{i_2}\psi_{\underline{H}}
     + Z ^i\slashed{\mathrm{Ric}},
     \label{Eq3.9}
     \\
     \nabla_3 Z ^i\psi_{H}
     =_s& Z ^iK+ Z ^i\slashed{\nabla} \eta+\sum_{i_1+i_2=i}  Z ^{i_1}\psi\cdot  Z ^{i_2}\psi+\sum_{ i_1+ i_2\le  i}  Z ^{i_1}(\psi+\psi_H) Z ^{i_2}\psi_{\underline{H}}
     + Z ^i\slashed{\mathrm{Ric}},
     \label{Eq3.10}
   \\
     \nabla_4 Z ^i\mathrm{tr}\chi
     =_s&\sum_{ i_1+ i_2+ i_3+ i_4\le  i} Z ^{i_1}\psi^{i_2} Z ^{i_3}\psi_H Z ^{i_4}\psi_H+ Z ^i\slashed{\mathrm{Ric}},
     \label{Eq3.11}
     \\
     \nabla_3 Z ^i\mathrm{tr}\underline{\chi}
     =_s&\sum_{ i_1+ i_2+ i_3+ i_4\le  i} Z ^{i_1}\psi^{i_2} Z ^{i_3}(\psi+\psi_H+\psi_{\underline{H}}) Z ^{i_4}\psi_{\underline{H}}
     + Z ^i\slashed{\mathrm{Ric}},
     \label{Eq3.12}
     \\
     \nabla_4 Z ^i\hat{\chi}
     =_s&\sum_{ i_1+ i_2+ i_3+ i_4\le  i}  Z ^{i_1}\psi^{i_2} Z ^{i_3}\psi_H Z ^{i_4}\psi_H
     +\sum_{ i_1+ i_2+ i_3\le  i}  Z ^{i_1}\psi^{i_2} Z ^{i_3}\alpha,
     \label{Eq3.13}
     \\
     \nabla_3 Z ^i\hat{\underline{\chi}}
     =_s&\sum_{ i_1+ i_2+ i_3+ i_4\le  i} Z ^{i_1}\psi^{i_2} Z ^{i_3}(\psi+\psi_H+\psi_{\underline{H}}) Z ^{i_4}\psi_{\underline{H}}
     +\sum_{ i_1+ i_2+ i_3\le  i}  Z ^{i_1}\psi^{i_2} Z ^{i_3}\underline{\alpha},
     \label{Eq3.14}
     \\
     \mathrm{div} Z ^i\hat{\chi}
     =_s&\frac{1}{2} Z ^i\slashed{\nabla}\mathrm{tr}\chi- Z ^i\beta
     +\sum_{ i_1+ i_2+ i_3+ i_4\le  i} Z ^{i_1}\psi^{i_2} Z ^{i_3}(\psi+\psi_H) Z ^{i_4}\psi_H
     + Z ^i\slashed{\mathrm{Ric}},
     \label{Eq3.15}
     \\
     \mathrm{div} Z ^i\hat{\underline{\chi}}
     =_s&\frac{1}{2} Z ^i\slashed{\nabla}\mathrm{tr}\underline{\chi}- Z ^i\underline{\beta}
     +\sum_{ i_1+ i_2+ i_3+ i_4\le  i} Z ^{i_1}\psi^{i_2} Z ^{i_3}(\psi+\psi_H) Z ^{i_4}\psi_{\Hb}
     + Z ^i\slashed{\mathrm{Ric}},
     \label{Eq3.16}
     \\
     \nabla_4 Z ^i\underline{\omega}
     =_s&
      Z ^iK+
      \sum_{ i_1+ i_2+ i_3+ i_4\le  i} Z ^{i_1}\psi^{i_2} Z ^{i_3}(\psi+\psi_H) Z ^{i_4}(\psi_H+\psi_{\Hb})
     + Z ^i\slashed{\mathrm{Ric}},
     \label{Eq3.17}
     \\
     0=\nabla_3 Z ^i\omega
     =_s& Z ^iK+\sum_{ i_1+ i_2+ i_3+ i_4\le  i} Z ^{i_1}\psi^{i_2} Z ^{i_3}(\psi+\psi_H) Z ^{i_4}\psi_{\Hb}
     + Z ^i\slashed{\mathrm{Ric}}.
     \label{Eq3.18}
   \end{align}
   For curvature terms,
   we have the following equations:
   \begin{align}
   \begin{split}
     \nabla_3& Z ^i\beta+\slashed{\nabla} Z ^i K
     -*\slashed{\nabla} Z ^i\check{\sigma}
     =_s\sum_{ i_1+ i_2\le1,\  i_3+ i_4+ i_5+ i_6\le  i}
      Z ^{i_3}\slashed{\nabla}^{i_1}\psi_{\underline{H}} Z ^{i_4}\psi^{i_5}
      Z ^{i_6}\slashed{\nabla}^{i_2}\psi_H
     \\
     &+\sum_{ i_1+ i_2+ i_3\le  i-1}\psi^{i_1} Z ^{i_2}K Z ^{i_3}(K,\check{\sigma})
     \\
     &+\sum_{ i_1+ i_2+ i_3+ i_4\le i}  Z ^{i_1}\psi_{H}^{i_2} Z ^{i_3}\psi_{\Hb} Z ^{i_4}\beta
     +\sum_{ i_1+ i_2+ i_3+ i_4\le i}  Z ^{i_1}\psi^{i_2} Z ^{i_3}(\psi+\psi_H) Z ^{i_4}(K,\sigma)
     \\
     &+ Z ^i((\slashed{\nabla},\nabla_3,\nabla_4)\slashed{\mathrm{Ric}}+\Gamma\slashed{\mathrm{Ric}}),
     \end{split}
     \label{Eq3.19}
     \end{align}
     \begin{align}
     \begin{split}
     \nabla_4& Z ^i\underline{\beta}-\slashed{\nabla} Z ^i K
     -*\slashed{\nabla} Z ^i\check{\sigma}
     =_s\sum_{ i_1+ i_2\le1,\  i_3+ i_4+ i_5+ i_6\le  i}
      Z ^{i_3}\slashed{\nabla}^{i_1}\psi_{H} Z ^{i_4}\psi^{i_5}
      Z ^{i_6}\slashed{\nabla}^{i_2}\psi_{\underline{H}}
     \\
     &+\sum_{ i_1+ i_2+ i_3\le  i-1}\psi^{i_1} Z ^{i_2}K Z ^{i_3}(K,\check{\sigma})
     +\sum_{ i_1+ i_2+ i_3+ i_4\le i}  Z ^{i_1}\psi^{i_2} Z ^{i_3}(\psi+\psi_H) Z ^{i_4}(K,\sigma)
     \\
     &+ Z ^i((\slashed{\nabla},\nabla_3,\nabla_4)\slashed{\mathrm{Ric}}+\Gamma\slashed{\mathrm{Ric}}),
     \end{split}
     \label{Eq3.20}
     \\
     \begin{split}
     \nabla_3& Z ^i\underline{\beta}+\slashed{\mathrm{div}} Z ^i \underline{\alpha}
     =_s
     \sum_{ i_1+ i_2+ i_3+ i_4\le i}  Z ^{i_1}\psi_H^{i_2} Z ^{i_3}\psi_{\Hb} Z ^{i_4}\bb
     +\sum_{ i_1+ i_2+ i_3+ i_4\le i}  Z ^{i_1}\psi^{i_2} Z ^{i_3}(\psi+\psi_H) Z ^{i_4}\ab
     \\
     &+ Z ^i((\slashed{\nabla},\nabla_3,\nabla_4)\slashed{\mathrm{Ric}}+\Gamma\slashed{\mathrm{Ric}}),
     \end{split}
     \label{Eq3.21}
     \\
     \begin{split}
     \nabla_4& Z ^i\beta-\slashed{\mathrm{div}} Z ^i \alpha
     =_s
     \sum_{ i_1+ i_2\le  i} Z ^{i_1}\psi_{H} Z ^{i_2}\beta
     +\sum_{ i_1+ i_2+ i_3+ i_4\le i}  Z ^{i_1}\psi^{i_2} Z ^{i_3}(\psi+\psi_H) Z ^{i_4}\alpha
     \\
     &+ Z ^i( (\slashed{\nabla},\nabla_4) \slashed{\mathrm{Ric}}+(\psi+\psi_H) \slashed{\mathrm{Ric}}),
     \end{split}
     \label{Eq3.22}
     \\
     \begin{split}
     \nabla_4& Z ^iK+\mathrm{div} Z ^i\beta
     =_s\sum_{ i_1+ i_2+ i_3+ i_4\le  i}  Z ^{i_1}\psi^{i_2} Z ^{i_3}\psi_H Z ^{i_4}(K,\check{\sigma})
     \\
     &+\sum_{ i_1+ i_2+ i_3+ i_4\le  i}  Z ^{i_1}\psi^{i_2} Z ^{i_3}(\psi+\psi_H)  Z ^{i_4}\slashed{\nabla}\psi_H
     + Z ^i( (\slashed{\nabla},\nabla_4)\slashed{\mathrm{Ric}}+(\psi+\psi_H)\slashed{\mathrm{Ric}}),
     \end{split}
     \label{Eq3.23}
     \\
     \begin{split}
     \nabla_4& Z ^i(-\mathrm{div}\eta+K)
     =_s\sum_{ i_1+ i_2+ i_3+ i_4\le  i}  Z ^{i_1}\psi^{i_2} Z ^{i_3}\psi_H Z ^{i_4}(K,\check{\sigma})
     \\
     &
     +\sum_{ i_1+ i_2+ i_3+ i_4\le  i}  Z ^{i_1}\psi^{i_2} Z ^{i_3}\psi 
      Z ^{i_4}\slashed{\nabla}\psi_H
     + Z ^i( (\slashed{\nabla},\nabla_4)\slashed{\mathrm{Ric}}+(\psi+\psi_H)\slashed{\mathrm{Ric}}),
     \end{split}
     \label{Eq3.24}
     \\
     \begin{split}
     \nabla_3& Z ^iK-\mathrm{div} Z ^i\underline{\beta}
     =_s\sum_{ i_1+ i_2+ i_3+ i_4\le  i} Z ^{i_1}\psi^{i_2} Z ^{i_3}\psi_{\underline{H}} Z ^{i_4}(K,\check{\sigma})
     \\
     &+\sum_{ i_1+ i_2+ i_3+ i_4\le  i}  Z ^{i_1}\psi^{i_2} Z ^{i_3}(\psi+\psi_H)  Z ^{i_4}\slashed{\nabla}\psi_{\underline{H}}
     +\sum_{ i_1+ i_2+ i_3+ i_4\le i}  Z ^{i_1}\psi_H^{i_2} Z ^{i_3}\psi_{\Hb} Z ^{i_4}K
     \\
     &
     + Z ^i((\slashed{\nabla},\nabla_3,\nabla_4)\slashed{\mathrm{Ric}}+\Gamma\slashed{\mathrm{Ric}}),
     \end{split}
     \label{Eq3.25}
     \\
     \begin{split}
     \nabla_3& Z ^i(-\mathrm{div}\underline{\eta}+K)
     =_s\sum_{ i_1+ i_2+ i_3+ i_4\le  i} Z ^{i_1}\psi^{i_2} Z ^{i_3}\psi_{\underline{H}} Z ^{i_4}(K,\check{\sigma})
     \\
     &+\sum_{ i_1+ i_2+ i_3+ i_4\le  i}  Z ^{i_1}\psi^{i_2} Z ^{i_3}(\psi+\psi_H) Z ^{i_4}\slashed{\nabla}\psi_{\underline{H}}
     +\sum_{ i_1+ i_2+ i_3+ i_4\le i}  Z ^{i_1}\psi_H^{i_2} Z ^{i_3}\psi_{\Hb} Z ^{i_4}(K,\slashed{\nabla}\psi)
     \\
     &+ Z ^i((\slashed{\nabla},\nabla_3,\nabla_4)\slashed{\mathrm{Ric}}+\Gamma\slashed{\mathrm{Ric}}),
     \end{split}
     \label{Eq3.26}
     \\
     \begin{split}
     \nabla_3& Z ^i\check{\sigma}+\mathrm{div}* Z ^i\underline{\beta}
     =_s\sum_{ i_1+ i_2+ i_3+ i_4\le  i} Z ^{i_1}\psi^{i_2} Z ^{i_3}\psi_{\underline{H}} Z ^{i_4}(K,\check{\sigma})
     \\
     &+\sum_{ i_1+ i_2+ i_3+ i_4\le  i}  Z ^{i_1}\psi^{i_2} Z ^{i_3}(\psi+\psi_H)  Z ^{i_4}\slashed{\nabla}\psi_{\underline{H}}
     +\sum_{ i_1+ i_2+ i_3+ i_4\le i}  Z ^{i_1}\psi_H^{i_2} Z ^{i_3}\psi_{\Hb} Z ^{i_4}\check{\sigma}
     \\
     &+ Z ^i((\slashed{\nabla},\nabla_3,\nabla_4)\slashed{\mathrm{Ric}}+\Gamma\slashed{\mathrm{Ric}}),
     \end{split}
     \label{Eq3.27}
     \\
     \begin{split}
     \nabla_4& Z ^i\check{\sigma}+\mathrm{div}* Z ^i\beta
     =_s\sum_{ i_1+ i_2+ i_3+ i_4\le  i}  Z ^{i_1}\psi^{i_2} Z ^{i_3}\psi_{H} Z ^{i_4}(K,\check{\sigma})
     \\
     &+\sum_{ i_1+ i_2+ i_3+ i_4\le  i}  Z ^{i_1}\psi^{i_2} Z ^{i_3}(\psi+\psi_H)  Z ^{i_4}\slashed{\nabla}\psi_{H}
     + Z ^i(\slashed{\nabla}\slashed{\mathrm{Ric}}+(\psi+\psi_H)\slashed{\mathrm{Ric}}),
     \end{split}
     \label{Eq3.28}
     \end{align}
     \begin{align}
     \begin{split}
     \nabla_3& Z ^i\alpha-\slashed{\nabla}\hat{\otimes} Z ^i\beta
     =_s\sum_{ i_1+ i_2+ i_3+ i_4\le  i}  Z ^{i_1}(\psi+\psi_H)^{i_2} Z ^{i_3}(\psi+\psi_{\underline{H}})
      Z ^{i_4}\alpha
     + Z ^{i_1}\psi^{i_2} Z ^{i_3}\psi_{H}
      Z ^{i_4}
     (K,\check{\sigma})
     \\
     &+\sum_{ i_1+ i_2+ i_3+ i_4\le i}  Z ^{i_1}\psi^{i_2} Z ^{i_3}\psi_H Z ^{i_4}\beta
     +
      Z ^{i_1}\psi_H Z ^{i_2}\psi_H  Z ^{i_3}\psi_{\underline{H}}
      \\
     &+ Z ^i((\slashed{\nabla},\nabla_3,\nabla_4)\slashed{\mathrm{Ric}}+\Gamma\slashed{\mathrm{Ric}}),
     \end{split}
     \label{Eq3.29}
     \\
     \begin{split}
     \nabla_4& Z ^i\underline{\alpha}-\slashed{\nabla}\hat{\otimes} Z ^i\underline{\beta}
     =_s\sum_{ i_1+ i_2+ i_3+ i_4\le  i}  Z ^{i_1}\psi^{i_2} Z ^{i_3}\psi_{H}
      Z ^{i_4}\underline{\alpha}
     + Z ^{i_1}\psi^{i_2} Z ^{i_3}\psi_{\underline{H}}
      Z ^{i_4}
     (K,\check{\sigma})
     \\
     &+\sum_{ i_1+ i_2+ i_3+ i_4\le i}  Z ^{i_1}\psi^{i_2} Z ^{i_3}\psi_H Z ^{i_4}\bb
     +
      Z ^{i_1}\psi_{\underline{H}} Z ^{i_2}\psi_{\underline{H}}  Z ^{i_3}\psi_{H}
      \\
      &
     + Z ^i((\slashed{\nabla},\nabla_3,\nabla_4)\slashed{\mathrm{Ric}}+\Gamma\slashed{\mathrm{Ric}}).
     \end{split}
     \label{Eq3.30}
   \end{align}
   Also, we need the following equations for the $\mathrm{div}$-$\mathrm{curl}$ systems of $\eta$ and $\omega$:
   \begin{align}
     \begin{split}
        \mathrm{curl} Z ^i\underline{\eta}
        =_s& Z ^i\check{\sigma}+\sum_{ i_1+ i_2\le  i} Z ^{i_1}\psi Z ^{i_2}\psi,
     \end{split}
     \label{Eq3.31}
     \\
     \begin{split}
        \mathrm{curl} Z ^i\eta
        =_s& Z ^i\check{\sigma}+\sum_{ i_1+ i_2\le  i} Z ^{i_1}\psi Z ^{i_2}\psi,
     \end{split}
     \label{Eq3.31*}
     \\
     \begin{split}
        \mathrm{div} Z ^i\nabla\omega &=_s\mathrm{div}\nabla_4 Z ^i\eta
        +\sum_{ i_1+ i_2\le  i} Z ^{i_1}\psi Z ^{i_2}\psi.
     \end{split}
     \label{Eq3.32}
   \end{align}

   \subsection{Modified Estimates with Ricci terms}\label{SubsectionModifiedEstimates}
   In the following, we follow the proof of \cite{luk2017weak} while replacing the transporting vector fields by $\nabla_3,\nabla_4$ and commutators $\nabla$ by $ Z $ as indicated in the norms. We will prove \eqref{Eq3.1} in this subsection.

   \begin{prop}\label{Prop3.7}
     Under the bootstrap assumptions, there exists $\epsilon_0=\epsilon_0(\Delta_1)$ such that for $\epsilon\le \epsilon_0$,
     \begin{align}
       &\frac{1}{2}<\Omega<2,\\
       &0<c\le\det\gamma\le C,\\
       &|\gamma|,|\gamma^{-1}|\le C,\\
       &|b^A|\le C\Delta_1\epsilon,
     \end{align}
     where the constants $c,C$ depend only on $d,D$. To be explicit,
     \begin{align}
       \sum_{ i\le3}\| Z ^i\psi_0\|_{L^\infty_uL_{\underline{u}}^\infty L^2(S)} \lesssim& C(\mathcal{O}_{ini})+\epsilon^{\frac{1}{2}} \mathcal{O}_{3,2},
       \label{Eq4.40.1}
       \\
       \sum_{ i\le 4}\epsilon^{-\frac{1}{2}}\|f(u) Z ^i\psi_0\|_{L^\infty_uL_{\underline{u}}^2 L^2(S)} \lesssim& C(\mathcal{O}_{ini})+\epsilon^{\frac{1}{2}} \mathcal{\tilde{\mathcal{O}}}_{4,2},
       \label{Eq4.40.2}
       \\
       \sum_{ i\le 4}\epsilon^{-\frac{1}{2}}\| Z ^i\psi_0\|_{L^\infty_{\underline{u}}L_{u}^2 L^2(S)} \lesssim& C(\mathcal{O}_{ini})+\epsilon^{\frac{1}{2}} \mathcal{\tilde{\mathcal{O}}}_{4,2}.
       \label{Eq4.40.3}
     \end{align}
   \end{prop}
   \begin{proof}
     Apply Proposition \ref{PropTransport} and the transport equations \eqref{Eq3.3} - \eqref{Eq3.6}, recall that $\psi_0\in\{b,\gamma,\gamma^{-1},\log\Omega\},\ \psi\in\{
     \eta,\underline{\eta}\}$, we have:
   \begin{align*}
   \begin{split}
     \sum_{ i\le3}\| Z ^i\psi_0\|_{L^\infty_uL_{\underline{u}}^\infty L^2(S)}
     \lesssim&\ C(\mathcal{O}_{ini})
     +\epsilon\sum_{ i\le3}\| Z ^i\psi\|_{L^\infty_uL_{\underline{u}}^\infty L^2(S)}
     \\
     &+\epsilon^{\frac{1}{2}}\sum_{ i\le3}\| Z ^i\psi_{H}\|_{L_{\underline{u}}^2 L^\infty_u L^2(S)}
     \\
     \lesssim&\ C(\mathcal{O}_{ini})+\epsilon^{\frac{1}{2}} C(\Delta_1,\mathcal{E}_{2,\infty})\mathcal{O}_{3,2},
   \end{split}
   \\
   \begin{split}
     \sum_{ i\le4}\epsilon^{-\frac{1}{2}}\|f(u) Z ^i\psi_0\|_{L^\infty_uL_{\underline{u}}^2 L^2(S)}
     \lesssim&\ C(\mathcal{O}_{ini})
     +\epsilon^{-\frac{1}{2}}\Big(
     \epsilon\sum_{ i\le4}\|f(u) Z ^i\psi\|_{L^\infty_uL_{\underline{u}}^2 L^2(S)}
     \\
     &+\epsilon\sum_{ i\le4}\| Z ^i\psi_{H}\|_{L_{u}^\infty L^2_{\underline{u}} L^2(S)}
     \Big)
     \\
     \lesssim&\ C(\mathcal{O}_{ini})+\epsilon^{\frac{1}{2}} C(\Delta_1,\mathcal{E}_{2,\infty})\mathcal{\tilde{\mathcal{O}}}_{4,2},
   \end{split}
   \\
   \begin{split}
     \sum_{ i\le4}\epsilon^{-\frac{1}{2}}\| Z ^i\psi_0\|_{L^\infty_{\underline{u}}L_{u}^2 L^2(S)}
     \lesssim&\ C(\mathcal{O}_{ini})
     +\epsilon^{-\frac{1}{2}}\Big(\epsilon\sum_{ i\le4}\| Z ^i\psi\|_{L^\infty_{\underline{u}}L_{u}^2 L^2(S)}
     \\
     &+\epsilon\sum_{ i\le4}\| Z ^i\psi_{H}\|_{L^\infty_uL_{\underline{u}}^2 L^2(S)}
     \Big)
     \\
     \lesssim&\ C(\mathcal{O}_{ini})+\epsilon^{\frac{1}{2}} C(\Delta_1,\mathcal{E}_{2,\infty})\mathcal{\tilde{\mathcal{O}}}_{4,2}.
   \end{split}
   \end{align*}
   The first part of the proposition is implied by \eqref{Eq4.40.1} - \eqref{Eq4.40.3}.
   \end{proof}

   \begin{prop}\label{Prop4.6.1}
     Under the bootstrap assumptions, there exists $\epsilon_0=\epsilon_0(\Delta_1)$ such that for $\epsilon\le \epsilon_0$, Ricci curvature satisfy the estimates
     \begin{align}
   \begin{split}
     &\sum_{ i\le4}\|f(u) Z ^i\mathrm{Ric}\|_{L_u^\infty L^2_{\underline{u}}L^2(S)}
     \le C(\Delta_1,\mathcal{E}_{2,\infty})\epsilon^{\frac{1}{2}}(\tilde{\mathcal{O}}_{4,2}+\tilde{\mathcal{E}}_{4,2}),
     \end{split}
     \label{Eq3.33}
     \\
   \begin{split}
     &\sum_{ i\le4}\| Z ^i\mathrm{Ric}\|_{L^\infty_{\underline{u}} L^2_{u}L^2(S)}
     \le C(\Delta_1,\mathcal{E}_{2,\infty})\epsilon^{\frac{1}{2}}(\tilde{\mathcal{O}}_{4,2}+\tilde{\mathcal{E}}_{4,2}).
\end{split}
\label{Eq3.34}
   \end{align}
   \end{prop}
   \begin{proof}
   Following schematic equations \eqref{Eq3.3.1}, fixing $u, \underline{u}$, respectively
     \begin{align*}
   \begin{split}
     \sum_{ i\le4}\|f(u) Z ^i\mathrm{Ric}\|_{L_u^\infty L^2_{\underline{u}}L^2(S)}
     \lesssim&
     \sum_{ j\le  4}
     \Big(
     \| Z ^{j} r\|_{L^\infty_{u} L^2_{\underline{u}} L^2(S)}
     +\|f(u) Z ^{j} \psi_0 \|_{L^\infty_{u} L^2_{\underline{u}} L^2(S)}
     \Big)
     \\
     \le&\ C(\Delta_1,\mathcal{E}_{2,\infty})\epsilon^{\frac{1}{2}}(\tilde{\mathcal{O}}_{4,2}+\tilde{\mathcal{E}}_{4,2}),
     \end{split}
     \\
   \begin{split}
     \sum_{ i\le4}\| Z ^i\mathrm{Ric}\|_{L^\infty_{\underline{u}} L^2_{u}L^2(S)}
     \lesssim&\sum_{ j\le 4}
     \Big(
     \| Z ^{j} r\|_{L^\infty_{\underline{u}} L^2_{u}L^2(S)}
     +\| Z ^{j} \psi_0 \|_{L^\infty_{\underline{u}} L^2_{u}L^2(S)}
     \Big)
     \\
     \le&\ C(\Delta_1,\mathcal{E}_{2,\infty})\epsilon^{\frac{1}{2}}(\tilde{\mathcal{O}}_{4,2}+\tilde{\mathcal{E}}_{4,2}).
\end{split}
   \end{align*}
   These estimates conclude the proof with the help of \eqref{Eq4.40.2} - \eqref{Eq4.40.3}.
   \end{proof}

   \begin{prop}\label{Prop4.6}
     Under the bootstrap assumptions, there exists $\epsilon_0=\epsilon_0(\Delta_1)$ such that for $\epsilon\le \epsilon_0$,
     \begin{align*}
       &\mathcal{O}_{3,2}\le C(\mathcal{O}_{ini})+\epsilon^{\frac{1}{2}} C(\Delta_1,\mathcal{E}_{2,\infty})(1+\mathcal{R}_3+\tilde{\mathcal{O}}_{4,2}+\tilde{\mathcal{E}}_{4,2}).
     \end{align*}
   \end{prop}
   \begin{proof}
   For $\psi_{\underline{H}}$, first for fixed $u$ respectively, we apply Proposition \ref{PropTransport} to get
   \begin{align}
   \begin{split}
     \|f(u) Z ^3\psi_{\underline{H}}\|_{L^2_uL_{\underline{u}}^\infty L^2(S)}
     \le&\Big\|f(u)\| Z ^3\psi_{\underline{H}}\|_{L_{\underline{u}}^\infty L^2(S)}\Big\|_{L^2_u}\
     \\
     \le&\ C(\mathcal{O}_{ini})+
     \Big\|f(u)\|\nabla_4 Z ^3\psi_{\underline{H}}\|_{L_{\underline{u}}^1 L^2(S)}\Big\|_{L^2_u}.
     \label{Eq4.38}
     \end{split}
   \end{align}
   Next, we substitute $ Z ^3$-commuted transport equations \eqref{Eq3.9} - \eqref{Eq3.10}. The right-hand side of equations \eqref{Eq3.9} - \eqref{Eq3.10} contains Weyl curvature terms, quadratic terms in Ricci coefficients, and also terms containing Ricci curvature. The Weyl curvature and Ricci coefficients terms can be handled exactly as in \cite{luk2017weak} and can be bounded by $\epsilon^{\frac{1}{2}} C(\Delta_1)(1+\mathcal{R}+\tilde{\mathcal{O}}_{4,2})$.
   
   In the presence of the fluid, however, the Ricci curvature terms are coupled with the fluid and they are handled by \eqref{Eq3.33} in Proposition \ref{Prop4.6.1} together with H\"older's inequality.

   Similarly for $\psi_H$ with \eqref{Eq3.34} in Proposition \ref{Prop4.6.1} and
   \begin{align}
   \begin{split}
     \| Z ^3\psi_{H}\|_{L^2_{\underline{u}}L_{u}^\infty L^2(S)}
     \le &\Big|\Big\|| Z ^3\psi_{H}\|_{L_{u}^\infty L^2(S)}\Big\|_{L^2_{\underline{u}}}
     \\
     \le &\ C(\mathcal{O}_{ini})+
     \Big|\Big\||\nabla_3 Z ^3\psi_{H}\|_{L_{u}^1 L^2(S)}\Big\|_{L^2_{\underline{u}}}.
   \end{split}
   \label{Eq4.39}
   \end{align}
   For the extra commutator term in the $\nabla_3 Z ^3\psi_{H}$ equation, we can bound it by
   \begin{align*}
     &\Big\|\sum_{i_1+i_2+i_3+i_4\le i}  Z ^{i_1}\psi_H^{i_2} Z ^{i_3}\psi_{\underline{H}} Z ^{i_4}\psi_H\Big\|_{L_{\ub}^2L_u^1L^2(S)}
     \\
     \le &\sum_{i_1+i_2+i_3+i_4\le i,\ i_3,i_4\le1} \|Z ^{i_1}\psi_H^{i_2}\|_{L_{u}^\infty L_{\ub}^2 L^2(S)}
     \|f(u)^{-1}\|_{L_u^2}
     \|f(u)Z ^{i_3}\psi_{\Hb}\|_{L_{u}^2 L_{\ub}^\infty  L^\infty(S)}
     \|Z ^{i_4}\psi_H\|_{ L_{u}^\infty L_{\ub}^\infty  L^\infty(S)}
     \\
     &+\sum_{i_1+i_2+i_3+i_4\le i,\ i_1,i_4\le1}
     \|Z ^{i_1}\psi_H^{i_2}\|_{ L_{\ub}^2 L_{u}^\infty L^\infty(S)}
     \|f(u)^{-1}\|_{L_u^2}
     \|f(u)Z ^{i_3}\psi_{\Hb}\|_{L_{\ub}^\infty L_{u}^2 L^2(S)}
     \|Z ^{i_4}\psi_H\|_{ L_{\ub}^\infty L_{u}^\infty L^\infty(S)}
     \\
     \le&\ \epsilon^{\f{1}{2}}  C(\Delta_1)(1+\mathcal{R}+\tilde{\mathcal{O}}_{4,2}).
   \end{align*}
   As for $\psi\in\{\eta, \underline{\eta}\}$, we obtain
   \begin{align}
     &\sum_{ i\le3}\| Z ^i\eta\|_{L^\infty_uL_{\underline{u}}^\infty L^2(S)}
     \le C(\mathcal{O}_{ini})+
     \sum_{ i\le3}\|\nabla_4 Z ^i\eta\|_{L^\infty_u L^1_{\underline{u}} L^2(S)},
     \label{Eq4.40}
     \\
     &\sum_{ i\le3}\| Z ^i\underline{\eta}\|_{L^\infty_u L_{\underline{u}}^\infty L^2(S)}
     \le C(\mathcal{O}_{ini})+
     \sum_{ i\le3}\|\nabla_3 Z ^i\underline{\eta}\|_{L_{\underline{u}}^\infty L^1_u L^2(S)}.
     \label{Eq4.41}
   \end{align}
   Next, we substitute $ Z ^3$-commuted transport equations \eqref{Eq3.7} - \eqref{Eq3.8}. The Weyl curvature and Ricci coefficients terms can be handled exactly as in \cite{luk2017weak} with the same bound. The Ricci curvature terms can be controlled by \eqref{Eq3.33} - \eqref{Eq3.34}.
   \end{proof}
   \begin{prop}
     \label{Prop4.8.1}
     Under the bootstrap assumptions, there exists $\epsilon_0=\epsilon_0(\Delta_1)$ such that for $\epsilon\le \epsilon_0$, Ricci curvature satisfy the estimates
     \begin{align}
     \sum_{ i\le3}\|f(u) Z ^i((\slashed{\nabla},\nabla_3,\nabla_4)\mathrm{Ric}+\Gamma\mathrm{Ric})\|_{L^\infty_{\underline{u}}L^2_u L^2(S)}
     &\le \epsilon^{\frac{1}{2}}C(\Delta_1,\mathcal{E}_{2,\infty})
     (\tilde{\mathcal{O}}_{4,2}+\tilde{\mathcal{E}}_{4,2}+\tilde{\mathcal{F}}_{4,2}).
     \label{Eq3.42}
   \end{align}
   \end{prop}
   \begin{proof}
   Following schematic equations \eqref{Eq3.3.1} - \eqref{Eq3.3.2},
     \begin{align*}
     \sum_{ i\le3}\|f(u) Z ^i((\slashed{\nabla},\nabla_3,\nabla_4)\mathrm{Ric}+\Gamma\mathrm{Ric})\|_{L^\infty_{\underline{u}}L^2_u L^2(S)}
     &\lesssim\sum_{ i\le3}
     \|f(u) Z ^{i}(\nabla_3(\psi_0r)+\Gamma r)\|_{L^\infty_{\underline{u}} L^2_{u}L^2(S)}
     \\
     &\lesssim\sum_{ i\le3}
     \|f(u)(| Z ^{i}\Gamma\| r|+
     |\psi_0\| Z ^{i}\nabla_3r|)\|_{L^\infty_{\underline{u}} L^2_{u}L^2(S)}
     \\
     &\le \epsilon^{\frac{1}{2}} C(\Delta_1,\mathcal{E}_{2,\infty})(\tilde{\mathcal{O}}_{4,2}+\tilde{\mathcal{E}}_{4,2}+\tilde{\mathcal{F}}_{4,2}).
   \end{align*}
   We apply transport equation \eqref{Eq3.8} of $\nabla_3\underline{\eta}$ to get
    \begin{align*}
       Z ^3\nabla_3\nabla_4 b
      =&_s Z ^3\nabla_3(\eta-\underline{\eta})
   =_s Z ^3\nabla_3\underline{\eta}-\frac{1}{2} Z ^i\slashed{\nabla}\underline{\omega}
   \\
   =_s&\sum_{ i_1+ i_2+ i_3+ i_4\le  3}  Z ^{i_1}\psi^{i_2} Z ^{i_3}\underline{\beta}
     +
      Z ^{i_1}\psi^{i_2} Z ^{i_3}\psi  Z ^{i_4}\psi_{\underline{H}}
     + Z ^3\slashed{\mathrm{Ric}}.
    \end{align*}
    Combining with $\partial_{u}\gamma=_s\hat{\chi},\ \nabla_3\log\Omega=_s\underline{\omega}$, we bound $\nabla_3\psi_0$ in $L_{\underline{u}}^\infty L_u^2 L^2(S)$ norm with $\mathcal{\tilde{\mathcal{O}}}_{4,2}$ norms. We also notice that with $\Gamma$ denoting all Christoffel symbols perpendicular to the tori, it follows that
   $ Z ^i\Gamma =_s\sum_{ j\le i} Z ^j\psi+ Z ^i\psi_H+ Z ^j\psi_{\underline{H}}$.
   \end{proof}
   \begin{prop}
     Under the bootstrap assumptions, there exists $\epsilon_0=\epsilon_0(\Delta_1)$ such that for $\epsilon\le \epsilon_0$, the following estimate holds:
     \begin{align*}
     \tilde{\mathcal{O}}_{4,2}
     \le C(\mathcal{O}_{ini})(1+\mathcal{R}_3
     + C(\Delta_1,\mathcal{E}_{2,\infty})
     (\epsilon^{\frac{1}{2}}(1+\mathcal{R}_3+\tilde{\mathcal{O}}_{4,2}^2)
     +\tilde{\mathcal{E}}_{4,2}+\tilde{\mathcal{F}}_{4,2})).
     \label{Prop4.8}
     \end{align*}
   \end{prop}
   \begin{proof}
   For $(K,\check{\sigma})$, we apply Proposition \ref{PropTransport} to get
   \begin{align}
     \| Z ^2(K,\check{\sigma})\|_{L^\infty_uL_{\underline{u}}^\infty L^2(S)}
     &\le C(\mathcal{O}_{ini})+\|\nabla_3 Z ^2(K,\check{\sigma})\|_{L^\infty_{\underline{u}}L^1_u L^2(S)}.
   \end{align}
   Next, we substitute $ Z ^2$-commuted transport equations \eqref{Eq3.23}, \eqref{Eq3.28}. The right hand side of equations \eqref{Eq3.23}, \eqref{Eq3.28} contains Weyl curvature terms, Ricci coefficients multiplied by $(K,\check{\sigma})$, and quadratic terms in Ricci coefficients, extra commutator term as described in Proposition \ref{Prop2.2}, and also terms containing Ricci curvature. The Weyl curvature and Ricci coefficients terms can be handled exactly as in \cite{luk2017weak} with the same bound.
   For the extra commutator term in the $\nabla_3 Z ^i(K,\check{\sigma})$ equation, we can bound it by
   \begin{align*}
     &\Big\|\sum_{i_1+i_2+i_3+i_4\le i}  Z ^{i_1}\psi_H^{i_2} Z ^{i_3}\psi_{\underline{H}} Z ^{i_4}(K,\check{\sigma})\Big\|_{L_{\ub}^2L_u^1L^2(S)}
     \\
     \le &\sum_{i_1+i_2+i_3+i_4\le i} \|Z ^{i_1}\psi_H^{i_2}\|_{L_{u}^\infty L_{\ub}^2 L^\infty(S)}
     \|f(u)^{-1}\|_{L_u^2}
     \|f(u)Z ^{i_3}\psi_{\Hb}\|_{L_{u}^2 L_{\ub}^\infty  L^\infty(S)}
     \|Z ^{i_4}(K,\sigma)\|_{ L_{u}^\infty L_{\ub}^\infty  L^2(S)}
     \\
     \le\ &C(\Delta_1,\mathcal{E}_{2,\infty})(1+\epsilon\mathcal{R}_3).
   \end{align*}
   The Ricci curvature terms can be controlled by \eqref{Eq3.42}. Therefore, we arrive at
   \begin{align*}
     &\| Z ^2(K,\check{\sigma})\|_{L^\infty_uL_{\underline{u}}^\infty L^2(S)}
     \le C(\Delta_1,\mathcal{E}_{2,\infty})(1+\epsilon\mathcal{R}_3+\epsilon^{\f{1}{2}} \tilde{\mathcal{O}}^2_{4,2}).
   \end{align*}
   Next for $\mathrm{tr}\chi,\mathrm{tr}\underline{\chi}$, we apply Proposition \ref{PropTransport} to get
   \begin{align}
    \begin{split}
     \| Z ^4\mathrm{tr}\chi\|_{L^2_{\underline{u}}L_{u}^\infty L^2(S)}
     \le&\ \Big\| Z ^4\mathrm{tr}\chi\|_{L^2(S)}\Big\|_{L^2_{\underline{u}}L_{u}^\infty}
     \\
     \le&\ C(\mathcal{O}_{ini})+
     \Big\|
     \nabla_4 Z ^4\mathrm{tr}\chi\|_{L_{\underline{u}'}^1 L^2(S)}
     \Big\|_{L^2_{\underline{u}}L_{u}^\infty },
     \end{split}
     \\
     \begin{split}
     \|f(u) Z ^4\mathrm{tr}\underline{\chi}\|_{L^2_{u} L^\infty_{\underline{u}} L^2}
     \le&\ \Big\|f(u)\| Z ^4\mathrm{tr}\underline{\chi}\|_{L^2(S)}
     \Big\|_{L^2_{u} L^\infty_{\underline{u}}(S)}
     \\
     \le&\ C(\mathcal{O}_{ini})+
     \Big\|f(u)\|\nabla_3 Z ^4\mathrm{tr}\underline{\chi}\|_{L_{u'}^1 L^2(S)}\Big\|_{L^2_{u} L^\infty_{\underline{u}}}.
     \end{split}
   \end{align}
    Next, we substitute $ Z ^4$-commuted transport equations \eqref{Eq3.11}, \eqref{Eq3.12}. The right hand side of equations \eqref{Eq3.11}, \eqref{Eq3.12} contains quadratic terms in Ricci coefficients, the extra commutator terms, and also terms containing Ricci curvature. The Ricci coefficients terms can be handled exactly as in \cite{luk2017weak} with the same bound. For the extra commutator term in the $\nabla_3 Z ^4\mathrm{tr}\chib$ equation, we can bound it by
   \begin{align*}
     &\Big\|\sum_{i_1+i_2+i_3+i_4\le i}  Z ^{i_1}\psi_H^{i_2} Z ^{i_3}\psi_{\underline{H}} Z ^{i_4}\psi_{\Hb}\Big\|_{L_u^1L^2(S)}
     \\
     \le &\sum_{i_1+i_2+i_3+i_4\le i,\ i_1= i-1}
     f(u)^{-2}
     \|Z ^{i_1}\psi_H^{i_2}\|_{L_{u}^2 L^2(S)}
     \|f(u)Z ^{i_3}\psi_{\Hb}\|_{L_{u}^\infty L^\infty(S)}
     \|f(u)Z ^{i_4}\psi_{\Hb}\|_{ L_{u}^2 L^\infty(S)}
     \\
     &+\sum_{i_1+i_2+i_3+i_4\le i,\ i_1, i_3\le i-2}f(u)^{-2}
     \|Z ^{i_1}\psi_H^{i_2}\|_{ L_{u}^\infty L^\infty(S)}
     \|f(u)Z ^{i_3}\psi_{\Hb}\|_{L_{u}^2 L^\infty(S)}
     \|f(u)Z ^{i_4}\psi_{\Hb}\|_{ L_{u}^2 L^2(S)}
     \\
     \le&\ f(u)^{-2}  C(\Delta_1)(1+\mathcal{R}+\tilde{\mathcal{O}}_{4,2}).
   \end{align*}
   The Ricci curvature terms can be controlled by \eqref{Eq3.33} - \eqref{Eq3.34}. Also, we remark that assuming the bootstrap assumptions, it follows as in \cite{luk2017weak} that eventually we have
   \begin{align*}
     &\| Z ^4\mathrm{tr}\chi\|_{L^\infty_{\underline{u}}L_{u}^\infty L^2(S)}
     +\| f(u)^2Z ^4\mathrm{tr}\chib\|_{L^\infty_{u}L_{\ub}^\infty L^2(S)}
     \le C(\Delta_1,\mathcal{E}_{2,\infty},\tilde{\mathcal{O}}_{4,2}).
   \end{align*}


   For $\hat{\chi},\hat{\underline{\chi}}$, we apply Proposition \ref{PropElliptic} for elliptic systems to get
   \begin{align}
   \begin{split}
     &\sum_{ i\le3}\| Z ^i\slashed{\nabla}\hat{\chi}\|_{L_u^\infty L_{\underline{u}}^2 L^2(S)}
     \\
     \lesssim&\sum_{ i\le3}\|\mathrm{div} Z ^i\slashed{\nabla}\hat{\chi}\|_{L_u^\infty L_{\underline{u}}^2 L^2(S)}
     +\sum_{ i\le3}\| Z ^i\hat{\chi}\|_{L_u^\infty L_{\underline{u}}^2 L^2(S)},
   \end{split}
   \\
   \begin{split}
     &\sum_{ i\le3}\|f(u) Z ^i\slashed{\nabla}\hat{\underline{\chi}}\|_{L_{\underline{u}}^\infty L^2_u L^2(S)}
     \\
     \lesssim&\sum_{ i\le3}\|f(u)\mathrm{div} Z ^i\slashed{\nabla}\hat{\underline{\chi}}\|_{L_{\underline{u}}^\infty L^2_u L^2(S)}
     +\sum_{ i\le3}\|f(u) Z ^i\hat{\underline{\chi}}\|_{L_{\underline{u}}^\infty L^2_u L^2(S)}.
     \end{split}
   \end{align}
   Next, we substitute $ Z ^3$-commuted transport equations \eqref{Eq3.15}, \eqref{Eq3.16}. The right hand side of equations \eqref{Eq3.15}, \eqref{Eq3.16} contains $\mathrm{tr}\chi,\mathrm{tr}\underline{\chi}$, Weyl curvature terms, quadratic terms in Ricci coefficients, the extra commutator terms and also terms containing Ricci curvature. The Weyl curvature and Ricci coefficients terms can be handled exactly as in \cite{luk2017weak} with the same bound.
   For the extra commutator terms in the $\mathrm{div} Z ^3\hat{\chi},\ \mathrm{div} Z ^3\hat{\chib}$ equations,
   \begin{align}
   \begin{split}
     &\sum_{i_1+i_2+i_3+i_4\le i} 
     \Big\| Z ^{i_1}\psi^{i_2} Z ^{i_3}\psi_{H} Z ^{i_4}\psi_{H}\Big\|_{L_u^\infty L_{\ub}^2 L^2(S)}
     +\Big\|f(u) Z ^{i_1}\psi^{i_2} Z ^{i_3}\psi_{H} Z ^{i_4}\psi_{\Hb}\Big\|_{L_{\ub}^\infty L_{u}^2 L^2(S)}
     \\
     \le &\sum_{i_1+i_2+i_3+i_4\le i,\ i_3,i_4\le 1}
     \|Z ^{i_1}\psi^{i_2}\|_{L_u^\infty L_{\ub}^2 L^2(S)}
     \|Z ^{i_3}\psi_{H}\|_{L_u^\infty L_{\ub}^\infty L^\infty(S)}
     \|Z ^{i_4}\psi_H\|_{L_u^\infty L_{\ub}^\infty L^\infty(S)}
     \\
     &+\sum_{i_1+i_2+i_3+i_4\le i,\ i_1,i_4\le 1}
     \|Z ^{i_1}\psi^{i_2}\|_{L_u^\infty L_{\ub}^\infty L^\infty(S)}
     \|Z ^{i_3}\psi_{H}\|_{L_u^\infty L_{\ub}^2 L^2(S)}
     \|Z ^{i_4}\psi_H\|_{L_u^\infty L_{\ub}^\infty L^\infty(S)}
     \\
     &+\sum_{i_1+i_2+i_3+i_4\le i,\ i_3,i_4\le 1}
     \|Z ^{i_1}\psi^{i_2}\|_{L_{\ub}^\infty L_{u}^2 L^2(S)}
     \|Z ^{i_3}\psi_{H}\|_{L_{\ub}^\infty L_{u}^\infty L^\infty(S)}
     \|f(u)Z ^{i_4}\psi_{\Hb}\|_{L_{\ub}^\infty L_{u}^\infty L^\infty(S)}
     \\
     &+\sum_{i_1+i_2+i_3+i_4\le i,\ i_1,i_4\le 1}
     \|Z ^{i_1}\psi^{i_2}\|_{L_{\ub}^\infty L_{u}^\infty L^\infty(S)}
     \|Z ^{i_3}\psi_{H}\|_{L_{\ub}^\infty L_{u}^2 L^2(S)}
     \|f(u)Z ^{i_4}\psi_{\Hb}\|_{L_{\ub}^\infty L_{u}^\infty L^\infty(S)}
     \\
     &+\sum_{i_1+i_2+i_3+i_4\le i,\ i_1,i_3\le 1}
     \|Z ^{i_1}\psi^{i_2}\|_{L_{\ub}^\infty L_{u}^\infty L^\infty(S)}
     \|Z ^{i_3}\psi_{H}\|_{L_{\ub}^\infty L_{u}^2 L^\infty(S)}
     \|f(u)Z ^{i_4}\psi_{\Hb}\|_{L_{\ub}^\infty L_{u}^\infty L^2(S)}
     \\
     \le&\ C(\Delta_1).
     \end{split}
   \label{Eq4.54.1}
   \end{align}
   The Ricci curvature terms can be controlled by \eqref{Eq3.33} - \eqref{Eq3.34}.

   For the remaining terms $\| Z ^3(\epsilon\nabla_4)\hat{\chi}\|_{L_u^\infty L_{\underline{u}}^2 L^2(S)}$,\
   $\|f(u) Z ^3(\epsilon\nabla_4)\hat{\underline{\chi}}\|_{L_{\underline{u}}^\infty L^2_u L^2(S)}$, which was not treated in \cite{luk2017weak}, where the different choice of order of the commutators applied can be related by commutator estimates we have done, by Prop \ref{PropTransport}, for $ i\le 3$,
   \begin{align}
   \| Z ^{i}(\psi_H,\psi)\|_{L_{\underline{u}}^\infty L^\infty_u L^p(S)}
   \le C(\mathcal{O}_{ini})+\| Z ^{i}\nabla_4(\psi_H,\psi)\|_{ L^\infty_u L_{\underline{u}}^1 L^p(S)},
   \label{3.31}
   \end{align}
   we utilize equations \eqref{Eq3.13}, \eqref{Eq3.14} to get
   \begin{align*}
     &\sum_{ i\le3}
     \Big(
     \epsilon^{-\frac{1}{2}}\| Z ^i(\epsilon\nabla_4)\hat{\chi}\|_{L_u^\infty L_{\underline{u}}^2 L^2(S)}+
     \|f(u) Z ^i(\epsilon\nabla_4)\hat{\underline{\chi}}\|_{L_{\underline{u}}^\infty L^2_u L^2(S)}
     \Big)
     \\
     \lesssim&\ \epsilon^{\frac{1}{2}}\sum_{ i\le3}
     \Big(
     \| Z ^i(\psi_H\psi_H,\alpha)\|_{L_u^\infty L_{\underline{u}}^2 L^2(S)}+
     \|f(u) Z ^i(\psi_H\psi_{\underline{H}},\slashed{\psi},\psi\psi,\mathrm{Ric})\|_{L_{\underline{u}}^\infty L^2_u L^2(S)}
     \Big)
     \\
     \lesssim&\ \epsilon^{\frac{1}{2}}\Big( \sum_{ i\le3,\ i_1+i_2=i, i_2\le1}
     \| Z ^{i_1}\psi_H\|_{L_u^\infty L_{\underline{u}}^\infty L^2(S)}
     \| Z ^{i_2}\psi_H\|_{L_u^\infty L_{\underline{u}}^2 L^\infty(S)}
     +\| Z ^i\alpha\|_{L_u^\infty L_{\underline{u}}^2 L^2(S)}
     \\
     &+\sum_{ i\le3,\ i_1+i_2=i,\  i_1\le1}
     \| Z ^{i_1}(\psi_H,\psi)\|_{L_{\underline{u}}^\infty L^\infty_u L^\infty(S)}
     \|f(u) Z ^{i_2}(\psi_{\underline{H}},\psi)\|_{L_{\underline{u}}^\infty L^2_u L^2(S)}
     \\
     &+\sum_{ i\le3,\ i_1+i_2=i,\  i_2\le1}
     \| Z ^{i_1}(\psi_H,\psi)\|_{L_{\underline{u}}^\infty L^\infty_u L^2(S)}
     \|f(u) Z ^{i_2}(\psi_{\underline{H}},\psi)\|_{L_{\underline{u}}^\infty L^2_u L^\infty(S)}
     \\
     &+\|f(u) Z ^3 (\slashed{\nabla}\psi,\mathrm{Ric})\|_{L_{\underline{u}}^\infty L^2_u L^2(S)}\Big)
     \\
     \le& \ \epsilon^{\frac{1}{2}}  C(\Delta_1,\mathcal{E}_{2,\infty})(1+\mathcal{R}_3+\tilde{\mathcal{O}}_{4,2}^2
     +\tilde{\mathcal{E}}_{4,2}+\tilde{\mathcal{F}}_{4,2}).
   \end{align*}
   where the Ricci curvature term has been controlled by \eqref{Eq3.33} - \eqref{Eq3.34}.

   For $\eta,\underline{\eta}$, we apply Proposition \ref{PropElliptic} for elliptic systems to get
   \begin{align}
   \begin{split}
     &\sum_{ i\le3}\|f(u) Z ^i\slashed{\nabla}(\eta,\underline{\eta})\|_{L_{u}^\infty L^2_{\underline{u}} L^2(S)}
     \\
     \lesssim&\sum_{ i\le3}
     \big(
     \|f(u)(K- Z ^i\mathrm{div}(\eta,\underline{\eta}))
     -f(u)K\|_{L_{u}^\infty L^2_{\underline{u}} L^2(S)}
     \\
     &+\|f(u) Z ^i\mathrm{div}(\eta,\underline{\eta})\|_{L_{u}^\infty L^2_{\underline{u}} L^2(S)}
     +\|f(u) Z ^i(\eta,\underline{\eta})\|_{L_{u}^\infty L^2_{\underline{u}} L^2(S)}
     \big),
     \end{split}
     \\
     \begin{split}
     &\sum_{ i\le3}\| Z ^i\slashed{\nabla}(\eta,\underline{\eta})\|_{L_{\underline{u}}^\infty L^2_u L^2(S)}
     \\
     \lesssim&\sum_{ i\le3}
     \big(
     \|(K- Z ^i\mathrm{div}(\eta,\underline{\eta}))-K\|_{L_{\underline{u}}^\infty L^2_u L^2(S)}
     \\
     &+\| Z ^i\mathrm{div}(\eta,\underline{\eta})\|_{L_{\underline{u}}^\infty L^2_u L^2(S)}+\| Z ^i(\eta,\underline{\eta})\|_{L_{\underline{u}}^\infty L^2_u L^2(S)}
     \big).
     \end{split}
   \end{align}
   Next, we substitute $ Z ^3$-commuted transport equations \eqref{Eq3.24}, \eqref{Eq3.26}  and $\mathrm{curl}$ equations \eqref{Eq3.31}, \eqref{Eq3.31*}. The right hand side of equations \eqref{Eq3.24}, \eqref{Eq3.26}, \eqref{Eq3.31}, \eqref{Eq3.31*} contains Ricci coefficients multiplied by $(K,\check{\sigma})$, quadratic terms in Ricci coefficients, the extra terms from commutation, and also terms containing Ricci curvature. The Ricci coefficients terms can be handled exactly as in \cite{luk2017weak} with the same bound. For the extra commutator terms in the $\mathrm{div} Z ^3(\eta,\etab),\ \mathrm{curl}Z^3(\eta,\etab)$ equations,
   \begin{align*}
     &\sum_{i_1+i_2+i_3+i_4\le i} 
     \Big\| Z ^{i_1}\psi^{i_2} Z ^{i_3}\psi_{H} Z ^{i_4}\psi\Big\|_{L_u^\infty L_{\ub}^2 L^2(S)}
     +\Big\|f(u) Z ^{i_1}\psi^{i_2} Z ^{i_3}\psi_{H} Z ^{i_4}\psi\Big\|_{L_{\ub}^\infty L_{u}^2 L^2(S)}
     \\
     \le &\sum_{i_1+i_2+i_3+i_4\le i,\ i_3,i_4\le 1}
     \|Z ^{i_1}\psi^{i_2}\|_{L_u^\infty L_{\ub}^2 L^2(S)}
     \|Z ^{i_3}\psi_{H}\|_{L_u^\infty L_{\ub}^\infty L^\infty(S)}
     \|Z ^{i_4}\psi\|_{L_u^\infty L_{\ub}^\infty L^\infty(S)}
     \\
     &+\sum_{i_1+i_2+i_3+i_4\le i,\ i_1,i_4\le 1}
     \|Z ^{i_1}\psi^{i_2}\|_{L_u^\infty L_{\ub}^\infty L^\infty(S)}
     \|Z ^{i_3}\psi_{H}\|_{L_u^\infty L_{\ub}^2 L^2(S)}
     \|Z ^{i_4}\psi\|_{L_u^\infty L_{\ub}^\infty L^\infty(S)}
     \\
     &+\sum_{i_1+i_2+i_3+i_4\le i,\ i_3,i_4\le 1}
     \|Z ^{i_1}\psi^{i_2}\|_{L_{\ub}^\infty L_{u}^2 L^2(S)}
     \|f(u)Z ^{i_3}\psi_{H}\|_{L_{\ub}^\infty L_{u}^\infty L^\infty(S)}
     \|Z ^{i_4}\psi\|_{L_{\ub}^\infty L_{u}^\infty L^\infty(S)}
     \\
     &+\sum_{i_1+i_2+i_3+i_4\le i,\ i_1,i_4\le 1}
     \|Z ^{i_1}\psi^{i_2}\|_{L_{\ub}^\infty L_{u}^\infty L^\infty(S)}
     \|f(u)Z ^{i_3}\psi_{H}\|_{L_{\ub}^\infty L_{u}^2 L^2(S)}
     \|Z ^{i_4}\psi\|_{L_{\ub}^\infty L_{u}^\infty L^\infty(S)}
     \\
     \le&\ C(\Delta_1).
   \end{align*}
   The Ricci curvature terms can be controlled by \eqref{Eq3.33} - \eqref{Eq3.34}. We arrive at
   \begin{align*}
     &\epsilon^{-\frac{1}{2}}\Big(\sum_{ i\le3}\|f(u) Z ^i\slashed{\nabla}(\eta,\underline{\eta})\|_{L_{u}^\infty L^2_{\underline{u}} L^2(S)}
     +\sum_{ i\le3}\| Z ^i\slashed{\nabla}(\eta,\underline{\eta})\|_{L_{\underline{u}}^\infty L^2_u L^2(S)}\Big)
     \\
     \le&
     C(\mathcal{O}_{ini})(1+\mathcal{R}
     +C(\Delta_1,\mathcal{E}_{2,\infty})(\epsilon^{\frac{1}{2}}\tilde{\mathcal{O}}_{4,2}+\tilde{\mathcal{E}}_{4,2}
     +\tilde{\mathcal{F}}_{4,2}))
     \\
     &\quad+\epsilon^{\frac{1}{2}} C(\Delta_1,\mathcal{E}_{2,\infty})(1+\mathcal{R}_3+\tilde{\mathcal{O}}_{4,2}^2
     +\tilde{\mathcal{E}}_{4,2}+\tilde{\mathcal{F}}_{4,2}).
   \end{align*}

   For the remaining terms $\|f(u) Z ^3(\epsilon\nabla_4)(\eta,\underline{\eta})\|_{L_{u}^\infty L^2_{\underline{u}} L^2(S)}$,\
   $\| Z ^3(\epsilon\nabla_4)(\eta,\underline{\eta})\|_{L_{\underline{u}}^\infty L^2_u L^2(S)}$, which was not treated in \cite{luk2017weak}, we utilize Equation \eqref{Eq3.7}, and estimate \eqref{3.31}, we get
   \begin{align*}
     &\epsilon^{-\frac{1}{2}}\sum_{ i\le3}
     \Big(
     \|f(u) Z ^i(\epsilon\nabla_4)(\eta,\underline{\eta})\|_{L_{u}^\infty L^2_{\underline{u}} L^2(S)}+
     \| Z ^i(\epsilon\nabla_4)(\eta,\underline{\eta})\|_{L_{\underline{u}}^\infty L^2_u L^2(S)}
     \Big)
     \\
     \lesssim& \epsilon^{\frac{1}{2}}\sum_{ i\le3}
     \Big(
     \|f(u) Z ^i(\psi_H\psi,\beta,\mathrm{Ric})\|_{L_{u}^\infty L^2_{\underline{u}} L^2(S)}
     +\| Z ^i(\psi_H\psi,\beta,\mathrm{Ric})\|_{L_{\underline{u}}^\infty L^2_u L^2(S)}
     \Big)
     \\
     \lesssim&\epsilon^{\frac{1}{2}}\Big(\sum_{ i\le3,i_1+i_2, i_2\le1}
     \| Z ^{i_1}\psi_H\|_{L_{u}^\infty L^2_{\underline{u}} L^2(S)}
     \|f(u) Z ^{i_2}\psi\|_{L_{u}^\infty L^\infty_{\underline{u}} L^2(S)}
     \\
     &+\sum_{ i\le3,i_1+i_2, i_1\le1}
     \| Z ^{i_1}\psi_H\|_{L_{u}^\infty L^\infty_{\underline{u}} L^2(S)}
     \|f(u) Z ^{i_2}\psi\|_{L_{u}^\infty L^2_{\underline{u}} L^2(S)}
     +\sum_{ i\le3}\|f(u) Z ^i(\beta,\mathrm{Ric})\|_{L_{u}^\infty L^2_{\underline{u}} L^2(S)}
     \\
     &+\sum_{ i\le3,i_1+i_2, i_2\le1}
     \| Z ^i\psi_H\|_{L_{\underline{u}}^\infty L^2_u L^2(S)}
     \| Z ^i\psi\|_{L_{\underline{u}}^\infty L^\infty_u L^2(S)}
     \\
     &+\sum_{ i\le3,i_1+i_2, i_1\le1}
     \| Z ^i\psi_H\|_{L_{\underline{u}}^\infty L^\infty_u L^2(S)}
     \| Z ^i\psi\|_{L_{\underline{u}}^\infty L^2_u L^2(S)}
     +\sum_{ i\le3}\| Z ^i(\beta,\mathrm{Ric})\|_{L_{\underline{u}}^\infty L^2_u L^2(S)}\Big)
     \\
     \le&
     \ \epsilon^{\frac{1}{2}}  C(\Delta_1,\mathcal{E}_{2,\infty})(1+\mathcal{R}_3+\tilde{\mathcal{O}}_{4,2}^2
     +\tilde{\mathcal{E}}_{4,2}+\tilde{\mathcal{F}}_{4,2}).
   \end{align*}
   where the Ricci curvature term has been controlled by \eqref{Eq3.33} - \eqref{Eq3.34}.

   For $\underline{\omega}$, we apply Proposition \ref{PropElliptic} for elliptic systems to get
   \begin{align}
   \begin{split}
     &\sum_{ i\le2}\|f(u) Z ^i\slashed{\nabla}^2\underline{\omega}\|_{L_{\underline{u}}^\infty L_u^2L^2(S)}
     \\
     \lesssim&\sum_{ i\le2}
     \Big(
     \|f(u) Z ^i\mathrm{div}\slashed{\nabla}\underline{\omega}\|_{L_{\underline{u}}^\infty L_u^2L^2(S)}
     +\|f(u) Z ^i\slashed{\nabla}\underline{\omega}\|_{L_{\underline{u}}^\infty L_u^2L^2(S)}
     \Big)
     \\
     \lesssim&\sum_{ i\le2}
     \Big(
     \|f(u) Z ^i\mathrm{div}\slashed{\nabla}\nabla_4
     \underline{\omega}\|_{L_{\underline{u}}^1 L_u^2L^2(S)}
     +\|f(u) Z ^i\slashed{\nabla}\underline{\omega}\|_{L_{\underline{u}}^\infty L_u^2L^2(S)}
     \Big)
     \\
     \lesssim&\sum_{ i\le2}
     \Big(
     \|f(u) Z ^i\mathrm{div}
     (\Omega^2\slashed{\nabla}(\nabla_4
     \underline{\omega}-K)-\Omega^2(\nabla_4\underline{\beta}-\slashed{\nabla}K)
     +\nabla_4\underline{\beta})\|_{L_{\underline{u}}^1 L_u^2L^2(S)}
     \\
     &+\|f(u) Z ^i\slashed{\nabla}\underline{\omega}\|_{L_{\underline{u}}^\infty L_u^2L^2(S)}
     \Big)
     .
     \end{split}
   \end{align}
   Next, we substitute $ Z ^2\slashed{\nabla}^2$-commuted \eqref{Eq3.17} and $ Z ^2\slashed{\nabla}$-commuted \eqref{Eq3.20}. The right hand side of equations \eqref{Eq3.17}, \eqref{Eq3.20} contains Weyl curvature, Ricci coefficients multiplied by $(K,\check{\sigma})$, quadratic terms in Ricci coefficients, extra terms from commutation, and terms containing Ricci curvature. The commutator terms can be dealt with by \eqref{Eq4.54.1}. The Weyl curvature and Ricci coefficients terms can be handled exactly as in \cite{luk2017weak} with the same bound. The Ricci curvature terms can be controlled by \eqref{Eq3.42}. To conclude, we have controlled the case where all derivatives are angular.

   For the remaining terms $\|f(u) Z ^i(\epsilon\nabla_4)\underline{\omega}\|_{L_{\underline{u}}^\infty L_{u}^2L^2(S)}$, which was not treated in \cite{luk2017weak}, we utilize Equation \eqref{Eq3.17} to get
   \begin{align*}
     &\epsilon^{-\frac{1}{2}}\sum_{ i\le3}\|f(u) Z ^i(\epsilon\nabla_4)\underline{\omega}\|_{L_{\underline{u}}^\infty L_{u}^2L^2(S)}
     \\
     \lesssim&\ \epsilon^{\frac{1}{2}}\sum_{ i\le3}
     \|f(u) Z ^i(\psi \psi,K,\psi_H\psi_{\underline{H}},\mathrm{Ric})\|_{L_{\underline{u}}^\infty L_{u}^2L^2(S)}
     \\
     \lesssim&\ \epsilon^{\frac{1}{2}}\Big(\sum_{ i\le3,i=i_1+i_2, i_2\le1}
     \|f(u) Z ^{i_1}\psi\|_{L_{\underline{u}}^\infty L_{u}^2L^2(S)}
     \| Z ^{i_2}\psi\|_{L_{\underline{u}}^\infty L_{u}^\infty L^2(S)}
     \\
     &+\sum_{ i\le3,i=i_1+i_2, i_2\le1}
     \| Z ^{i_1}\psi_H\|_{L_{\underline{u}}^\infty L_{u}^2L^2(S)}
     \|f(u) Z ^{i_2}\psi_{\underline{H}}\|_{L_{\underline{u}}^\infty L_{u}^\infty L^2(S)}
     \\
     &+\sum_{ i\le3,i=i_1+i_2, i_1\le1}
     \| Z ^{i_1}\psi_H\|_{L_{\underline{u}}^\infty L_{u}^\infty L^2(S)}
     \|f(u) Z ^{i_2}\psi_{\underline{H}}\|_{L_{\underline{u}}^\infty L_{u}^2L^2(S)}
     \\
     &+\sum_{ i\le3}
     \| Z ^i (K,\mathrm{Ric})\|_{L_{\underline{u}}^\infty L_{u}^2L^2(S)}\Big)
     \\
     \le &\ \epsilon^{\frac{1}{2}}   C(\Delta_1,\mathcal{E}_{2,\infty})(1+\mathcal{R}_3+\tilde{\mathcal{O}}_{4,2}^2
     +\tilde{\mathcal{E}}_{4,2}+\tilde{\mathcal{F}}_{4,2}).
   \end{align*}
   where the Ricci curvature term has been controlled in \eqref{Eq3.42}.
   \end{proof}
   \begin{prop}\label{Prop3.12}
     Under the bootstrap assumptions, there exists $\epsilon_0=\epsilon_0(\Delta_1)$ such that for $\epsilon\le \epsilon_0$, Ricci curvature satisfy the estimates
     \begin{align}
     \begin{split}
       &\sum_{ i\le3}
       \Big(
       \| Z ^i( (\slashed{\nabla},\nabla_4) \slashed{\mathrm{Ric}}+(\psi+\psi_H)\slashed{\mathrm{Ric}})\|_{L^1_u L^2_{\underline{u}}L^2(S)}
     +\| Z ^i( (\slashed{\nabla},\nabla_4) \slashed{\mathrm{Ric}}+(\psi+\psi_H)\slashed{\mathrm{Ric}})\|_{L^1_{\underline{u}} L^2_{u}L^2(S)}
     \Big)
     \\
     \le\ &\epsilon C(\Delta_1,\mathcal{E}_{2,\infty})
     (\tilde{\mathcal{O}}_{4,2}+\tilde{\mathcal{E}}_{4,2}+\tilde{\mathcal{F}}_{4,2}),
     \label{Eq3.43}
     \end{split}
     \\
     \begin{split}
     &\sum_{ i\le3}
     \Big(
     \| Z ^i((\slashed{\nabla},\nabla_3,\nabla_4)\slashed{\mathrm{Ric}}+\Gamma\slashed{\mathrm{Ric}})\|_{ L^1_{u}L^2_{\underline{u}} L^2(S)}
     +\|f(u) Z ^i((\slashed{\nabla},\nabla_3,\nabla_4)\slashed{\mathrm{Ric}}+\Gamma\slashed{\mathrm{Ric}})\|_{ L^1_{\underline{u}}L^2_{u} L^2(S)}
     \Big)
     \\
     \le&\ \epsilon C(\Delta_1,\mathcal{E}_{2,\infty})
     (\tilde{\mathcal{O}}_{4,2}+\tilde{\mathcal{E}}_{4,2}+\tilde{\mathcal{F}}_{4,2}).
     \end{split}
     \label{Eq3.44}
     \end{align}
   \end{prop}
   \begin{proof}
     For \eqref{Eq3.43}, following schematic equations \eqref{Eq3.3.1} - \eqref{Eq3.3.2}, we have
     \begin{align*}
     \sum_{ i\le3}\| Z ^i( (\slashed{\nabla},\nabla_4) \slashed{\mathrm{Ric}}+(\psi+\psi_H)\slashed{\mathrm{Ric}})\||_{L^1_u L^2_{\underline{u}}L^2(S)}
     &\lesssim\epsilon\sum_{ i\le3}
     \|f(u) Z ^{i}(\partial_{\nu}(\psi_0r)+(\psi+\psi_H) r)\|_{L^2_u L^2_{\underline{u}}L^2(S)}
     \\
     &\le \epsilon C(\Delta_1,\mathcal{E}_{2,\infty})
     (\tilde{\mathcal{O}}_{4,2}+\tilde{\mathcal{E}}_{4,2}+\tilde{\mathcal{F}}_{4,2}).
   \end{align*}
     For \eqref{Eq3.44}, following schematic equations \eqref{Eq3.3.1} - \eqref{Eq3.3.2}, we have
     \begin{align*}
       \sum_{ i\le3}\| Z ^i((\slashed{\nabla},\nabla_3,\nabla_4)\slashed{\mathrm{Ric}}+\Gamma\slashed{\mathrm{Ric}})\|_{ L^1_{u}L^2_{\underline{u}} L^2(S)}
       &\lesssim\sum_{ i\le3}
     \| Z ^{i}(\nabla_3(\psi_0r)+\Gamma r)\|_{ L^1_{u}L^2_{\underline{u}}L^2(S)}
     \\
     &\lesssim\ \epsilon\sum_{ i\le3}
     \|f(u) Z ^{i}(\psi_0\nabla_3r+\Gamma r)\|_{ L^2_{u}L^2_{\underline{u}}L^2(S)}
     \\
     &\le\ \epsilon C(\Delta_1,\mathcal{E}_{2,\infty})(\tilde{\mathcal{O}}_{4,2}+\tilde{\mathcal{E}}_{4,2}+\tilde{\mathcal{F}}_{4,2}),
     \end{align*}
     and
     \begin{align*}
     \sum_{ i\le3}\|f(u) Z ^i((\slashed{\nabla},\nabla_3,\nabla_4)\slashed{\mathrm{Ric}}+\Gamma\slashed{\mathrm{Ric}})\|_{L^1_{\underline{u}}L^2_u L^2(S)}
     &\lesssim\sum_{ i_1\le i}
     \|f(u) Z ^{i_1}(\nabla_3(\psi_0r)+\Gamma r)\|_{L^1_{\underline{u}} L^2_{u}L^2(S)}
     \\
     &\lesssim\epsilon\sum_{ i_1\le i}
     \|f(u) Z ^{i_1}(\psi_0\nabla_3r+\Gamma r)\|_{L^\infty_{\underline{u}} L^2_{u}L^2(S)}
     \\
     &\le\epsilon C(\Delta_1,\mathcal{E}_{2,\infty})(\tilde{\mathcal{O}}_{4,2}+\tilde{\mathcal{E}}_{4,2}+\tilde{\mathcal{F}}_{4,2}),
   \end{align*}
   which conclude the proof.
   \end{proof}
   \begin{prop}\label{Prop3.13}
     Under the bootstrap assumptions, there exists $\epsilon_0=\epsilon_0(\Delta_1)$ such that for $\epsilon\le \epsilon_0$,
     \begin{align*}
     \mathcal{R}_3
     \le&\ C(\mathcal{O}_{ini},\mathcal{R}_{ini})
     (1+\epsilon^{\frac{1}{2}} C(\Delta_1,\mathcal{E}_{2,\infty})(1+\mathcal{R}_3+\tilde{\mathcal{O}}_{4,2}
     +\tilde{\mathcal{E}}_{4,2}+\tilde{\mathcal{F}}_{4,2})).
     \end{align*}
   \end{prop}
   \begin{proof}
   For $(K,\check{\sigma},\beta)$, we apply 
   $L^2$ energy estimates to get
   \begin{align*}
     &\epsilon^{-1}\sum_{ i\le3}
     \Big(
     \| Z ^i(K,\check{\sigma})\|_{L_{\underline{u}}^\infty L_u^2L^2(S)}^2
     +\| Z ^i\beta \|_{L_u^\infty L_{\underline{u}}^2L^2(S)}^2
     \Big)
     \\
     \le&\
     C(\mathcal{O}_{ini})
     +C(\Delta_1,\mathcal{E}_{2,\infty})\Big(
     \epsilon^{-1}\sum_{ i\le3}
     \big(
     \int_{D_{u,\underline{u}}}| Z ^i\check{\sigma}|
     |\nabla_4 Z ^i\check{\sigma}+\mathrm{div}* Z ^i\beta|
     \\
     &+\int_{D_{u,\underline{u}}} | Z ^i K|
     |\nabla_4 Z ^iK+\mathrm{div} Z ^i\beta|
     +\int_{D_{u,\underline{u}}}| Z ^i\beta\|\nabla_3Z^i\beta
     +\slashed{\nabla} Z ^iK-*\slashed{\nabla} Z ^i\check{\sigma}|
     \big)
     \\
     &+\sum_{i_1+i_2=i, i\le3}
     \int_{D_{u,\underline{u}}}
     \big(| Z ^i(K,\check{\sigma})|
     | Z ^{i_1}(\psi+\psi_H) Z ^{i_2}(K,\check{\sigma})|
      +
      | Z ^i\beta|
     | Z ^{i_1}(\psi+\psi_{\underline{H}}) Z ^{i_2}\beta|
     \big)\Big).
   \end{align*}
   Next, we substitute $ Z ^3$-commuted \eqref{Eq3.19}, \eqref{Eq3.23}, \eqref{Eq3.28}. The right hand side of equations \eqref{Eq3.19}, \eqref{Eq3.23}, \eqref{Eq3.28} contains Ricci coefficients multiplied by $(K,\check{\sigma})$, quadratic terms in Ricci coefficients, the extra commutator terms, and terms containing Ricci curvature. The Ricci coefficients terms can be handled exactly as in \cite{luk2017weak} with the same bound.
    For the extra commutator terms in the $\mathrm{div}Z^3(K,\check{\sigma},\b),\ \nabla_3Z^3\b$ equations,
   \begin{align*}
     &\sum_{i_1+i_2+i_3+i_4\le i}
     \Big\|Z ^{i_1}\psi^{i_2} Z ^{i_3}\psi_{H} Z ^{i_4}(K,\check{\sigma})\Big\|_{L_{u}^1 L_{{\ub}}^2 L^2(S)}
     \\
     &\qquad\qquad
     +\Big\|Z ^{i_1}\psi^{i_2} Z ^{i_3}\psi_{H} Z ^{i_4}\b\Big\|_{L_{\ub}^1 L_{u}^2 L^2(S)}
     +\Big\|Z ^{i_1}\psi_H^{i_2} Z ^{i_3}\psi_{\Hb} Z ^{i_4}\b\Big\|_{L_{u}^1 L_{{\ub}}^2 L^2(S)}
     \end{align*}
     \begin{align*}
     \le
     &\sum_{i_1+i_2+i_3+i_4\le i,\ i_3,i_4\le 1}
     \|Z ^{i_1}\psi^{i_2}\|_{L_{\ub}^\infty L_{u}^\infty L^2(S)}
     \|Z ^{i_3}\psi_{H}\|_{L_{\ub}^2 L_{u}^\infty L^\infty(S)}
     \|Z ^{i_4}(K,\check{\sigma})\|_{L_{\ub}^\infty L_{u}^1 L^\infty(S)}
     \\
     &+\sum_{i_1+i_2+i_3+i_4\le i,\ i_1,i_4\le 1}
     \|Z ^{i_1}\psi^{i_2}\|_{L_{\ub}^\infty L_{u}^\infty L^\infty(S)}
     \|Z ^{i_3}\psi_{H}\|_{L_{\ub}^2 L_{u}^\infty L^2(S)}
     \|Z ^{i_4}(K,\check{\sigma})\|_{L_{\ub}^\infty L_{u}^1 L^\infty(S)}
     \\
     &+\sum_{i_1+i_2+i_3+i_4\le i,\ i_1,i_3\le 1}
     \|Z ^{i_1}\psi^{i_2}\|_{L_{\ub}^\infty L_{u}^\infty L^\infty(S)}
     \|Z ^{i_3}\psi_{H}\|_{L_{\ub}^2 L_{u}^\infty L^\infty(S)}
     \|Z ^{i_4}(K,\check{\sigma})\|_{L_{\ub}^\infty L_{u}^1 L^2(S)}
     \\
     &\sum_{i_1+i_2+i_3+i_4\le i,\ i_3,i_4\le 1}
     \|Z ^{i_1}\psi^{i_2}\|_{L_{\ub}^\infty L_{u}^\infty L^2(S)}
     \|Z ^{i_3}\psi_{H}\|_{L_{\ub}^2 L_{u}^\infty L^\infty(S)}
     \|Z ^{i_4}\b\|_{L_{\ub}^2 L_{u}^2 L^\infty(S)}
     \\
     &+\sum_{i_1+i_2+i_3+i_4\le i,\ i_1,i_4\le 1}
     \|Z ^{i_1}\psi^{i_2}\|_{L_{\ub}^\infty L_{u}^\infty L^\infty(S)}
     \|Z ^{i_3}\psi_{H}\|_{L_{\ub}^2 L_{u}^\infty L^2(S)}
     \|Z ^{i_4}\b\|_{L_{\ub}^2 L_{u}^2 L^\infty(S)}
     \\
     &+\sum_{i_1+i_2+i_3+i_4\le i,\ i_1,i_3\le 1}
     \|Z ^{i_1}\psi^{i_2}\|_{L_{\ub}^\infty L_{u}^\infty L^\infty(S)}
     \|Z ^{i_3}\psi_{H}\|_{L_{\ub}^2 L_{u}^\infty L^\infty(S)}
     \|Z ^{i_4}\b\|_{L_{\ub}^2 L_{u}^2 L^2(S)}
     \\
     &\sum_{i_1+i_2+i_3+i_4\le i,\ i_3,i_4\le 1}
     \|Z ^{i_1}\psi_H^{i_2}\|_{L_{\ub}^\infty L_{u}^\infty L^2(S)}
     \|Z ^{i_3}\psi_{\Hb}\|_{L_{u}^2 L_{\ub}^\infty  L^\infty(S)}
     \|Z ^{i_4}\b\|_{L_{\ub}^2 L_{u}^2 L^\infty(S)}
     \\
     &+\sum_{i_1+i_2+i_3+i_4\le i,\ i_1,i_4\le 1}
     \|Z ^{i_1}\psi_H^{i_2}\|_{L_{\ub}^\infty L_{u}^\infty L^\infty(S)}
     \|Z ^{i_3}\psi_{\Hb}\|_{L_{u}^2 L_{\ub}^\infty  L^2(S)}
     \|Z ^{i_4}\b\|_{L_{\ub}^2 L_{u}^2 L^\infty(S)}
     \\
     &+\sum_{i_1+i_2+i_3+i_4\le i,\ i_1,i_3\le 1}
     \|Z ^{i_1}\psi_H^{i_2}\|_{L_{\ub}^\infty L_{u}^\infty L^\infty(S)}
     \|Z ^{i_3}\psi_{\Hb}\|_{L_{u}^2 L_{\ub}^\infty  L^\infty(S)}
     \|Z ^{i_4}\b\|_{L_{\ub}^2 L_{u}^2 L^2(S)}
     \\
     \le&\ C(\Delta_1)\ep^{\f{1}{2}}\mathcal{R}_3.
   \end{align*}
   The Ricci curvature terms can be controlled by \eqref{Eq3.43}.
   Therefore, H\"older's inequality can be applied to the additional Ricci curvature terms in the first integral, for example,
   \begin{align*}
     \epsilon^{-1}\int_{D_{u,\underline{u}}} | Z ^i\check{\sigma}|
     | Z ^i( (\slashed{\nabla},\nabla_4) \slashed{\mathrm{Ric}}+(\psi+\psi_H)\slashed{\mathrm{Ric}})|
     \le \mathcal{R}_3\cdot \epsilon^{\frac{1}{2}} C(\Delta_1,\mathcal{E}_{2,\infty})(\tilde{\mathcal{O}}_{4,2}+\tilde{\mathcal{E}}_{4,2}+\tilde{\mathcal{F}}_{4,2}).
   \end{align*}
   Similarly, for $(K,\check{\sigma},\underline{\beta})$, we apply 
   $L^2$ energy estimates to get
   \begin{align*}
     &\epsilon^{-1}\sum_{ i\le3}
     \Big(
     \|f(u) Z ^i(K,\check{\sigma})\|_{L_u^\infty   L_{\underline{u}}^2L^2(S)}^2
     +\|f(u) Z ^i\underline{\beta}\|_{L_{\underline{u}}^\infty L_u^2L^2(S)}^2
     \Big)
     \\
     \le&\
     C(\mathcal{O}_{ini})
     -\epsilon^{-1}\sum_{ i\le3}\int_{D_{u,\underline{u}}}
     f(u)\nabla_3f(u)| Z ^i(K,\check{\sigma})|^2
     \\
     &+C(\Delta_1,\mathcal{E}_{2,\infty})\Big(
     \epsilon^{-1}\sum_{ i\le3}
     \big(
     \int_{D_{u,\underline{u}}} f(u)^2| Z ^i\check{\sigma}|
     |\nabla_3 Z ^i\check{\sigma}+\mathrm{div}* Z ^i\underline{\beta}|
     \\
     &+\int_{D_{u,\underline{u}}} f(u)^2| Z ^i K|
     |\nabla_3 Z ^iK-\mathrm{div} Z ^i\underline{\beta}|
     \\
     &+\int_{D_{u,\underline{u}}} f(u)^2| Z ^i\underline{\beta}\|\nabla_4 Z ^i\underline{\beta}
     -\slashed{\nabla} Z ^iK-*\slashed{\nabla} Z ^i\check{\sigma}|
     \big)
     \\
     &+\epsilon^{-1}\sum_{i_1+i_2=i, i\le3}
     \int_{D_{u,\underline{u}}}f(u)^2
     \big(| Z ^i(K,\check{\sigma})|
     | Z ^{i_1}(\psi+\psi_{\underline{H}}) Z ^{i_2}(K,\check{\sigma})|
      +| Z ^i\beta|
     | Z ^{i_1}(\psi+\psi_{H}) Z ^{i_2}\beta|
     \big)\Big).
   \end{align*}
   For the second term, $\nabla_3f(u)>0$ has a favourable sign. We substitute $ Z ^3$-commuted \eqref{Eq3.20}, \eqref{Eq3.25}, \eqref{Eq3.27}. For the extra commutator terms in the $\mathrm{div}Z^3(K,\check{\sigma},\b),\ \nabla_3Z^3(K,\check{\sigma})$ equations,
   \begin{align*}
     &\sum_{i_1+i_2+i_3+i_4\le i}
     \Big\|f(u)Z ^{i_1}\psi^{i_2} Z ^{i_3}\psi_{H} Z ^{i_4}(K,\check{\sigma})\Big\|_{L_{\ub}^1 L_{{u}}^2 L^2(S)}
     \\
     &\qquad\qquad
     +\Big\|f(u)Z ^{i_1}\psi^{i_2} Z ^{i_3}\psi_{H} Z ^{i_4}\b\Big\|_{L_{u}^1 L_{\ub}^2 L^2(S)}
     +\Big\|f(u)Z ^{i_1}\psi_H^{i_2} Z ^{i_3}\psi_{\Hb} Z ^{i_4}(K,\check{\sigma})\Big\|_{L_{u}^1 L_{{\ub}}^2 L^2(S)}
     \\
     \le
     &\sum_{i_1+i_2+i_3+i_4\le i,\ i_3,i_4\le 1}
     \|Z ^{i_1}\psi^{i_2}\|_{L_{\ub}^\infty L_{u}^\infty L^2(S)}
     \|Z ^{i_3}\psi_{H}\|_{L_{\ub}^2 L_{u}^\infty L^\infty(S)}
     \|f(u)Z ^{i_4}(K,\check{\sigma})\|_{L_{\ub}^2 L_{u}^2 L^\infty(S)}
     \\
     &+\sum_{i_1+i_2+i_3+i_4\le i,\ i_1,i_4\le 1}
     \|Z ^{i_1}\psi^{i_2}\|_{L_{\ub}^\infty L_{u}^\infty L^\infty(S)}
     \|Z ^{i_3}\psi_{H}\|_{L_{\ub}^2 L_{u}^\infty L^2(S)}
     \|f(u)Z ^{i_4}(K,\check{\sigma})\|_{L_{\ub}^2 L_{u}^2 L^\infty(S)}
     \\
     &+\sum_{i_1+i_2+i_3+i_4\le i,\ i_1,i_3\le 1}
     \|Z ^{i_1}\psi^{i_2}\|_{L_{\ub}^\infty L_{u}^\infty L^\infty(S)}
     \|Z ^{i_3}\psi_{H}\|_{L_{\ub}^2 L_{u}^\infty L^\infty(S)}
     \|f(u)Z ^{i_4}(K,\check{\sigma})\|_{L_{\ub}^2 L_{u}^2 L^2(S)}
     \\
     &\sum_{i_1+i_2+i_3+i_4\le i,\ i_3,i_4\le 1}
     \|Z ^{i_1}\psi^{i_2}\|_{L_{\ub}^\infty L_{u}^\infty L^2(S)}
     \|Z ^{i_3}\psi_{H}\|_{L_{\ub}^2 L_{u}^\infty L^\infty(S)}
     \|f(u)Z ^{i_4}\b\|_{L_{u}^1 L_{\ub}^\infty L^\infty(S)}
     \\
     &+\sum_{i_1+i_2+i_3+i_4\le i,\ i_1,i_4\le 1}
     \|Z ^{i_1}\psi^{i_2}\|_{L_{\ub}^\infty L_{u}^\infty L^\infty(S)}
     \|Z ^{i_3}\psi_{H}\|_{L_{\ub}^2 L_{u}^\infty L^2(S)}
     \|f(u)Z ^{i_4}\b\|_{L_{u}^1 L_{\ub}^\infty L^\infty(S)}
     \\
     &+\sum_{i_1+i_2+i_3+i_4\le i,\ i_1,i_3\le 1}
     \|Z ^{i_1}\psi^{i_2}\|_{L_{\ub}^\infty L_{u}^\infty L^\infty(S)}
     \|Z ^{i_3}\psi_{H}\|_{L_{\ub}^2 L_{u}^\infty L^\infty(S)}
     \|f(u)Z ^{i_4}\b\|_{L_{u}^1 L_{\ub}^\infty  L^2(S)}
     \\
     &\sum_{i_1+i_2+i_3+i_4\le i,\ i_3,i_4\le 1}
     \|Z ^{i_1}\psi_H^{i_2}\|_{L_{\ub}^\infty L_{u}^\infty L^2(S)}
     \|Z ^{i_3}\psi_{\Hb}\|_{L_{u}^2 L_{\ub}^\infty  L^\infty(S)}
     \|f(u)Z ^{i_4}\b\|_{L_{\ub}^2 L_{u}^2 L^\infty(S)}
     \\
     &+\sum_{i_1+i_2+i_3+i_4\le i,\ i_1,i_4\le 1}
     \|Z ^{i_1}\psi_H^{i_2}\|_{L_{\ub}^\infty L_{u}^\infty L^\infty(S)}
     \|Z ^{i_3}\psi_{\Hb}\|_{L_{u}^2 L_{\ub}^\infty  L^2(S)}
     \|f(u)Z ^{i_4}\b\|_{L_{\ub}^2 L_{u}^2 L^\infty(S)}
     \\
     &+\sum_{i_1+i_2+i_3+i_4\le i,\ i_1,i_3\le 1}
     \|Z ^{i_1}\psi_H^{i_2}\|_{L_{\ub}^\infty L_{u}^\infty L^\infty(S)}
     \|Z ^{i_3}\psi_{\Hb}\|_{L_{u}^2 L_{\ub}^\infty  L^\infty(S)}
     \|f(u)Z ^{i_4}\b\|_{L_{\ub}^2 L_{u}^2 L^2(S)}
     \\
     \le&\ C(\Delta_1)\ep^{\f{1}{2}}\mathcal{R}_3.
   \end{align*}
   The Ricci curvature terms can be controlled by \eqref{Eq3.44}.

   For $(\beta,\alpha)$ in 
   Equations \eqref{Eq3.22}, \eqref{Eq3.29},
   denote the lower order terms by
   \begin{align*}
     &\nabla_4 Z ^i\beta-\slashed{\mathrm{div}} Z ^i \alpha=F_{i,1},
     \\
     &\nabla_3 Z ^i\alpha-\slashed{\nabla}\hat{\otimes} Z ^i\beta=F_{i,2}.
   \end{align*}
   Apply 
   $L^2$ energy estimates to get
   \begin{align*}
     \int_{D_{u,\underline{u}}} \nabla_4
     \left< Z ^i\beta, Z ^i\beta\right>_\gamma\mathrm{d}\sigma
     =&\int_{D_{u,\underline{u}}}
     \left<\slashed{\mathrm{div}} Z ^i \alpha, Z ^i\beta\right>_\gamma\mathrm{d}\sigma
     \\
     &+\int_{D_{u,\underline{u}}} 
     \left<F_{i,1},  Z ^i\beta\right>_\gamma
     \\
     &+\Omega^{-2}
     \left< Z ^i\beta, Z ^i\beta\right>_{\nabla_4\gamma}\mathrm{d}\sigma,
   \end{align*}
   \begin{align*}
     \int_{D_{u,\underline{u}}}\left<\nabla_3 Z ^i\alpha, Z ^i\alpha\right>_\gamma\mathrm{d}\sigma
     =&\int_{D_{u,\underline{u}}}\left<\slashed{\nabla}\hat{\otimes} Z ^i\beta, Z ^i\alpha\right>_\gamma\mathrm{d}\sigma
     \\
     &+\int_{D_{u,\underline{u}}}\left<F_{i,2}, Z ^i\alpha\right>_\gamma
     \\
     &+\left< Z ^i\alpha, Z ^i\alpha\right>_{\nabla_3\gamma}\mathrm{d}\sigma,
   \end{align*}
   where $\mathrm{tr}\alpha=0$ implies
   \begin{align*}
     \int_{D_{u,\underline{u}}}\left<\slashed{\mathrm{div}} Z ^i \alpha, Z ^i\beta\right>_\gamma\mathrm{d}\sigma
     =-\int_{D_{u,\underline{u}}}\left<\slashed{\nabla}\hat{\otimes} Z ^i\beta, Z ^i\alpha\right>_\gamma\mathrm{d}\sigma
     +\int_{D_{u,\underline{u}}}\frac{1}{2}\left< Z ^i\alpha\cdot \gamma^{-1}\slashed{\nabla}\gamma, Z ^i\beta\right>_\gamma\mathrm{d}\sigma.
   \end{align*}
   This yields
   \begin{align*}
     &\frac{1}{2}\Omega^{-2}\int_{\underline{H}_{\underline{u}}}| Z ^i\beta|_\gamma^2\mathrm{d}\sigma
     +\frac{1}{2}\int_{H_u}| Z ^i \alpha|_\gamma^2\mathrm{d}\sigma
     \\
     =&-\frac{1}{2}\int_{D_{u,\underline{u}}}| Z ^i\beta|_\gamma^2\mathrm{tr}\chi\mathrm{d}\sigma
     -\frac{1}{2}\int_{D_{u,\underline{u}}}| Z ^i\alpha|_\gamma^2\mathrm{tr}\underline{\chi}\mathrm{d}\sigma
     +\int_{D_{u,\underline{u}}}\frac{1}{2}\left< Z ^i\alpha\cdot \gamma^{-1}\slashed{\nabla}\gamma, Z ^i\beta\right>_\gamma\mathrm{d}\sigma
     \\
     &+\int_{D_{u,\underline{u}}}\slashed{\mathrm{div}}b| Z ^i\alpha|_\gamma^2\mathrm{d}\sigma
     +\int_{D_{u,\underline{u}}} \left<F_{i,1},  Z ^i\beta\right>_\gamma\mathrm{d}\sigma
     +\int_{D_{u,\underline{u}}} \left<F_{i,2},  Z ^i\alpha\right>_\gamma\mathrm{d}\sigma,
   \end{align*}
   where $F_{i,1},F_{i,2}$ coming from equations \eqref{Eq3.22}, \eqref{Eq3.29} contains Ricci coefficients multiplied by $(K,\check{\sigma}),\alpha,\beta$, cubic terms in Ricci coefficients, and terms containing Ricci curvature. The cubic term in Ricci coefficients terms satisfy:
   \begin{align*}
     &\sum_{ i_1+ i_2+ i_3\le i+1}\| Z ^{i_1}\psi_H Z ^{i_2}\psi_H Z ^{i_3}\psi_{\underline{H}}
      Z ^i\alpha\|_{L^1_uL^1_{\underline{u}}L^1(S)}
     \\
     \lesssim &\sum_{|j_1|\le i+1, |j_2|, |j_3|\le i}
     \| Z ^{j_1}\psi_H\|_{L_u^\infty L_{\underline{u}}^2 L^2(S)}
     \| Z ^{j_2}\psi_H\|_{L_u^\infty L_{\underline{u}}^\infty L^2(S)}
     \| Z ^{j_3}\psi_{\underline{H}}\|_{ L_u^1 L_{\underline{u}}^\infty L^2(S)}
     \| Z ^i\alpha\|_{L_{u}^\infty L_{\underline{u}}^2 L^2(S)}
     \\
     &+\sum_{|j_1|, |j_2|\le i, |j_3|\le i+1}
     \| Z ^{j_1}\psi_H\|_{L_u^\infty L_{\underline{u}}^\infty L^2(S)}
     \| Z ^{j_2}\psi_H\|_{L_u^\infty L_{\underline{u}}^\infty L^2(S)}
     \\
     &\qquad\qquad\qquad\qquad
     \|f(u)^{-1}\|_{L_u^2}
     \|f(u) Z ^{j_3}\psi_{\underline{H}}\|_{ L_{\underline{u}}^2 L_u^2 L^2(S)}
     \| Z ^i\alpha\|_{L_{u}^\infty L_{\underline{u}}^2 L^2(S)}.
   \end{align*}
   The terms inlvoving $(K,\check{\sigma}),\alpha,\beta$ satisfy:
   \begin{align*}
     &\sum_{ i_1+ i_2\le i}\| Z ^{i_1}\psi_H Z ^{i_2}(K,\sigma,\beta) Z ^i\beta\|_{L^1_uL^1_{\underline{u}}L^1(S)}
     \\
     \lesssim& \sum_{j_1, j_2\le i}\| Z ^{j_1}\psi_H\|_{L_{\underline{u}}^1 L_u^\infty L^2(S)}
     \| Z ^{j_2}(K,\sigma,\beta)\|_{L_{\underline{u}}^\infty L_u^2L^2(S)}
     \| Z ^i\beta\|_{L_{\underline{u}}^\infty L_u^2L^2(S)},
     \\
     &\sum_{ i_1+ i_2\le i}\| Z ^{i_1}\psi Z ^{i_2}\alpha Z ^i\beta\|_{L^1_uL^1_{\underline{u}}L^1(S)}
     \\
     \lesssim& \sum_{j_1, j_2\le i}\| Z ^{j_1}\psi\|_{ L_{\underline{u}}^\infty L_{u}^\infty L^2(S)}
     \| Z ^{j_2}\alpha\|_{L_{u}^2 L_{\underline{u}}^2L^2(S)}
     \| Z ^i\beta\|_{L_{\underline{u}}^2 L_u^2L^2(S)}.
   \end{align*}
   For the extra commutator terms in the $\mathrm{div}Z^3(\a,\b),\ \nabla_3Z^3\a$ equations,
   \begin{align*}
     &\sum_{i_1+i_2+i_3+i_4\le i}
     \Big\|Z ^{i_1}\psi^{i_2} Z ^{i_3}\psi_{H} Z ^{i_4}\a\Big\|_{L_{\ub}^1 L_{{u}}^2 L^2(S)}
     \\
     &\qquad\qquad
     +\Big\|Z ^{i_1}\psi^{i_2} Z ^{i_3}\psi_{H} Z ^{i_4}\b\Big\|_{L_{u}^1 L_{\ub}^2 L^2(S)}
     +\Big\|Z ^{i_1}\psi_H^{i_2} Z ^{i_3}\psi_{\Hb} Z ^{i_4}\a\Big\|_{L_{u}^1 L_{{\ub}}^2 L^2(S)}
     \\
     \le
     &\sum_{i_1+i_2+i_3+i_4\le i,\ i_3,i_4\le 1}
     \|Z ^{i_1}\psi^{i_2}\|_{L_{\ub}^\infty L_{u}^\infty L^2(S)}
     \|Z ^{i_3}\psi_{H}\|_{L_{\ub}^2 L_{u}^\infty L^\infty(S)}
     \|Z ^{i_4}\a\|_{L_{\ub}^2 L_{u}^2 L^\infty(S)}
     \\
     &+\sum_{i_1+i_2+i_3+i_4\le i,\ i_1,i_4\le 1}
     \|Z ^{i_1}\psi^{i_2}\|_{L_{\ub}^\infty L_{u}^\infty L^\infty(S)}
     \|Z ^{i_3}\psi_{H}\|_{L_{\ub}^2 L_{u}^\infty L^2(S)}
     \|Z ^{i_4}\a\|_{L_{\ub}^2 L_{u}^2 L^\infty(S)}
     \\
     &+\sum_{i_1+i_2+i_3+i_4\le i,\ i_1,i_3\le 1}
     \|Z ^{i_1}\psi^{i_2}\|_{L_{\ub}^\infty L_{u}^\infty L^\infty(S)}
     \|Z ^{i_3}\psi_{H}\|_{L_{\ub}^2 L_{u}^\infty L^\infty(S)}
     \|Z ^{i_4}\a\|_{L_{\ub}^2 L_{u}^2 L^2(S)}
     \\
     &\sum_{i_1+i_2+i_3+i_4\le i,\ i_3,i_4\le 1}
     \|Z ^{i_1}\psi^{i_2}\|_{L_{\ub}^\infty L_{u}^\infty L^2(S)}
     \|1\|_{L_u^2L_{\ub}^\infty L^\infty(S)}
     \|Z ^{i_3}\psi_{H}\|_{ L_{\ub}^2 L_{u}^\infty L^\infty(S)}
     \|Z ^{i_4}\b\|_{L_{\ub}^\infty L_{u}^2 L^\infty(S)}
     \\
     &+\sum_{i_1+i_2+i_3+i_4\le i,\ i_1,i_4\le 1}
     \|Z ^{i_1}\psi^{i_2}\|_{L_{\ub}^\infty L_{u}^\infty L^\infty(S)}
     \|1\|_{L_u^2L_{\ub}^\infty L^\infty(S)}
     \|Z ^{i_3}\psi_{H}\|_{L_{\ub}^2 L_{u}^\infty L^2(S)}
     \|Z ^{i_4}\b\|_{L_{\ub}^\infty L_{u}^2 L^\infty(S)}
     \end{align*}
     \begin{align*}
     &+\sum_{i_1+i_2+i_3+i_4\le i,\ i_1,i_3\le 1}
     \|Z ^{i_1}\psi^{i_2}\|_{L_{\ub}^\infty L_{u}^\infty L^\infty(S)}
     \|1\|_{L_u^2L_{\ub}^\infty L^\infty(S)}
     \|Z ^{i_3}\psi_{H}\|_{L_{\ub}^2 L_{u}^\infty L^\infty(S)}
     \|Z ^{i_4}\b\|_{L_{\ub}^\infty L_{u}^2 L^2(S)}
     \\
     &\sum_{i_1+i_2+i_3+i_4\le i,\ i_3,i_4\le 1}
     \|Z ^{i_1}\psi_H^{i_2}\|_{L_{\ub}^\infty L_{u}^\infty L^2(S)}
     \|Z ^{i_3}\psi_{\Hb}\|_{L_{u}^2 L_{\ub}^\infty  L^\infty(S)}
     \|Z ^{i_4}\a\|_{L_{\ub}^2 L_{u}^2 L^\infty(S)}
     \\
     &+\sum_{i_1+i_2+i_3+i_4\le i,\ i_1,i_4\le 1}
     \|Z ^{i_1}\psi_H^{i_2}\|_{L_{\ub}^\infty L_{u}^\infty L^\infty(S)}
     \|Z ^{i_3}\psi_{\Hb}\|_{L_{u}^2 L_{\ub}^\infty  L^2(S)}
     \|Z ^{i_4}\a\|_{L_{\ub}^2 L_{u}^2 L^\infty(S)}
     \\
     &+\sum_{i_1+i_2+i_3+i_4\le i,\ i_1,i_3\le 1}
     \|Z ^{i_1}\psi_H^{i_2}\|_{L_{\ub}^\infty L_{u}^\infty L^\infty(S)}
     \|Z ^{i_3}\psi_{\Hb}\|_{L_{u}^2 L_{\ub}^\infty  L^\infty(S)}
     \|Z ^{i_4}\a\|_{L_{\ub}^2 L_{u}^2 L^2(S)}
     \\
     \le&\ C(\Delta_1)\ep^{\f{1}{2}}\mathcal{R}_3,
   \end{align*}
   where for the second term, we first reduce to $L_u^2L_{\ub}^2L^2(S)=L_{\ub}^2L_u^2L^2(S)$ norm.
   The terms in The Ricci curvature terms can be controlled by \eqref{Eq3.43} together with H\"older's inequality, for example,
   \begin{align*}
     \int_{D_{u,\underline{u}}}| Z ^i\beta|
     | Z ^i( (\slashed{\nabla},\nabla_4) \slashed{\mathrm{Ric}}+(\psi+\psi_H)\slashed{\mathrm{Ric}})|
     \le \mathcal{R}_3\cdot\epsilon^{\frac{1}{2}} C(\Delta_1,\mathcal{E}_{2,\infty})(\tilde{\mathcal{O}}_{4,2}+\tilde{\mathcal{E}}_{4,2}+\tilde{\mathcal{F}}_{4,2}).
   \end{align*}
   Combining the estimates above, the proof is complete.
   \end{proof}

   \section{Estimates for the Fluid}\label{Section5}
      We now turn to prove the estimates for the fluid variables, for Theorem \ref{SmoothExist}, on every compact region $[0,u_*']\times[0,\underline{u}_*]\times\mathbb{T}^2$, within which a solution to the system exists, where $u_*'\le u_*,\ \underline{u}_*'\le \underline{u}_*$, stated as follows as in \eqref{Eq3.001}.
   \begin{align}
     \begin{split}
     &
     (1-(p')^{\frac{1}{3}})\tilde{\mathcal{E}}_{4,2}^2
     \le C(\mathcal{E}_{ini})(1
     +\epsilon C(\Delta_1,\mathcal{E}_{2,\infty})),
     \\
     &\tilde{\mathcal{F}}_{4,2} \le C(\Delta_1,\mathcal{E}_{2,\infty})(1+\tilde{\mathcal{E}}_{4,2}+\mathcal{O}_{3,2}),
     \\
     &\mathcal{E}_{2,\infty}\le C(\Delta_1)(1+\tilde{\mathcal{E}}_{4,2}^3),
     \end{split}
     \label{Eq4.1}
   \end{align}
   where the norms are defined in Section \ref{norms}.

   By choice of the initial data, $ Z ^i v$ are initially bounded for $ i\le2$. Due to the component coming from the singular Christoffel symbols of the spacetime metric, $\nabla_3v$ is also singular. In the following, we will prove energy estimates for $ Z ^iv$ and derive estimates for the other derivative $\nabla_3$ from the equations.

   We will decompose the fluid velocity $v$ with respect to the frame $(e_A,e_3,e_4)$ as
   \begin{align*}
      &v^\alpha e_\alpha
      =\slashed{v}^A e_A+v^3 e_3+v^4e_4
      =\slashed{v}^A e_A+w e_3+\underline{w} e_4.
   \end{align*}
   The Euler equation satisfied by $v^3$ restricts our choice to $e_3=\p_u+b^A\partial_A,e_4=\Om^{-2}\p_{\ub}$. To be explicit, the fluid components are expected to be continuous up to the singularity only under this choice of null coordinates, see Section \ref{SingularFluidVariables}. Recall that in the following, we denote by $\Gamma$ the Ricci coefficients perpendicular to the torus.

   \subsection{Estimates for $\nabla_3$ Derivative of Fluid Variables}\label{SingularFluidVariables}
   \begin{prop}\label{SingularEnergy}
     Under the bootstrap assumptions, there exists $\epsilon_0=\epsilon_0(\Delta_1)>0$ such that for $\epsilon\le \epsilon_0$,
     \begin{align*}
     &\tilde{\mathcal{F}}_{4,2} \le \ep^{\f{1}{2}}C(\Delta_1,\mathcal{E}_{2,\infty})(1+\tilde{\mathcal{E}}_{4,2}+\mathcal{O}_{3,2}).
     \end{align*}
   \end{prop}
   \begin{proof}
   The normalization condition
   \begin{align*}
     &-\Omega^2 w\underline{w}+\gamma_{AB}\slashed{v}^A\slashed{v}^B=-1,
   \end{align*}
   implies that
   \begin{align*}
     \Omega^2w\underline{w} &\ge1.
   \end{align*}
   By fundamental theorem of calculus, for $\epsilon$ chosen such that $\epsilon^{\frac{1}{2}}\tilde{\mathcal{E}}_{2,\infty}\le 1$,
   \begin{align*}
     &|v^\alpha|\le |v^\alpha|\Big{|}_{S_{0,0}}+2\epsilon^{\frac{1}{2}}\tilde{\mathcal{E}}_{2,\infty}\le C(D,d),
   \end{align*}
   we arrive at
   \begin{align*}
     & v^3,v^4\ge \Omega^{-2}C^{-1}.
   \end{align*}
   Utilizing the equations, denote by $v^\lambda D_\lambda$ for $\lambda\in\{A,B,4\}$,
   \begin{align*}
     &  v^3 D_3L+ D_3v^3=-v^\lambda D_\lambda L- D_\lambda v^\lambda,
     \\
     & v^3 D_3v^\alpha+v^\alpha p'v^3 D_3L-\frac{1}{2} p'\delta^{\alpha}_4 D_3 L
     =-v^\lambda D_\lambda v^\alpha-v^\alpha p'v^\lambda D_\lambda L-p' (1-\delta^{\alpha}_4)D^\alpha L.
   \end{align*}
   Rearrange the equations as
   \begin{align*}
     &  D_3v^3=(v^3(1-p'))^{-1}(-v^\lambda D_\lambda v^\alpha-v^\alpha p'v^\lambda D_\lambda L
     -\frac{1}{2}p'D_4L
     -p'v^3(-v^\lambda D_\lambda L- D_\lambda v^\lambda)),
     \\
     &  D_3L=(v^3)^{-1}(-v^\lambda D_\lambda L- D_\lambda v^\lambda- D_3v^3),
     \\
     &  D_3v^\alpha=(v^3)^{-1}(-v^\lambda D_\lambda v^\alpha-v^\alpha p'v^\lambda D_\lambda L- p' (1-\delta^{4}_{\alpha})D^\alpha L-v^\alpha p'v^3 D_3L+\frac{1}{2} p'\delta^{\alpha}_4 D_3 L).
   \end{align*}
   Expand covariant derivatives as
   \begin{align*}
     & D_\alpha v^\beta
     =\nabla_\alpha \slashed{v}^\beta -\left<D_\alpha e^\beta-\nabla_\alpha e^\beta,e_\gamma\right> v^\gamma.
   \end{align*}
   Recall the schematic notation for $r\in\{\slashed{v},L\}$. Here $\mathcal{E}_{2,\infty},\mathcal{O}_{1,\infty}$ are specifically used as constants. Then the commuted equations abstract to, for $i\le 3$, schematically
   \begin{align}
     &\nabla_3 Z ^i r
     =_s\epsilon^{-1}\sum_{ j\le i+1} Z ^{j} r
     +\sum_{ j\le i}Z^{j}\Gamma.
     \label{Eq5.1}
   \end{align}
   Consequently, there exists $C=C(\Delta_1,\mathcal{E}_{2,\infty})$ such that
   \begin{align*}
     &\sum_{ i\le3}\|f(u)\nabla_3 Z ^i r\|_{L_{\underline{u}}^2 L_{u}^2 L^2(S)}
     \lesssim
     \sum_{ i\le4}\|f(u)\epsilon^{-1} Z ^i r\|_{L_{\underline{u}}^2 L_{u}^2 L^2(S)}
     +\sum_{ i\le3}
     \|f(u) Z ^{i}\Gamma\|_{L_{\underline{u}}^2 L_{u}^2 L^2(S)}.
   \end{align*}
   Note that by smallness of $f(u)$, it follows that $\|f(u)\epsilon^{-1}\|_{L_u^2}=\epsilon$. We arrive at
   \begin{align*}
     \tilde{\mathcal{F}}_{4,2} &\le \ep^{\f{1}{2}}
     C(\Delta_1,\mathcal{E}_{2,\infty})(1+\tilde{\mathcal{E}}_{4,2}+\mathcal{O}_{3,2}),
   \end{align*}
   representing bounds for $\nabla_3v$.
   \end{proof}
   
   \subsection{Commuting the equations}\label{CommutingFluids}
   As in \cite{speck2012nonlinear}, define auxiliary variable $L$,
   \begin{align*}
     &L(\tau)=\int_{0}^{\tau}\frac{1}{p(\tilde{\tau})+\tilde{\tau}}\mathrm{d}\tilde{\tau}.
   \end{align*}
   Derivatives of $L$ recover derivatives of $\tau$, since
   \begin{align*}
     \nabla L &= \frac{\tau}{p(\tau)+\tau}\nabla\log\tau,
   \end{align*}
   and the former is a smooth function in $\log\tau$. Then the equations \eqref{Eq2.51}, \eqref{Eq2.52} for the fluid variables reduce to
   \begin{align}
     &v^\alpha D_\alpha L+\mathrm{div}v=0,
       \label{Eq4.3}
       \\
       &v^\alpha D_\alpha v^\beta
       +p'(\tau)(v^\alpha D_\alpha L v^\beta+D^\beta L)=0.
       \label{Eq4.4}
   \end{align}
   Recall that $\Gamma$ denotes a general Ricci coefficient perpendicular to the tori. Recall that we defined the schematic notation in Section \ref{Section4} omitting coefficients of size $C(\Delta_1,\mathcal{E}_{2,\infty})$. 
   For $i\le 4$, schematically
   \begin{align}
      Z ^i\Gamma &=_s\sum_{ j\le i} Z ^j\psi+ Z ^j\psi_H+ Z ^j\psi_{\underline{H}}.
     \label{Eq4.4.1}
   \end{align}
   The equations \eqref{Eq4.3} - \eqref{Eq4.4} for fluid variables schematically read
   \begin{align}
     &v^\alpha \nabla_\alpha L+ \nabla_\alpha v^\alpha
     =-v^\beta\Gamma_{\alpha\beta}^\alpha,
     \label{4.12}
       \\
       &v^\alpha \nabla_\alpha v^\beta
       +p'(\tau)(v^\alpha \nabla_\alpha L v^\beta+\partial^\beta L)
       =-v^\alpha v^\nu\Gamma_{\alpha\nu}^\beta.
       \label{4.13}
   \end{align}
   Applying \eqref{Eq5.1}, commutation of \eqref{4.12} - \eqref{4.13} with $Z^i$ implies the schematic commuted equations, for $i\le 4$,
   \begin{align}
     &v^\alpha \nabla_\alpha  Z ^i L+ \nabla_\alpha  Z ^i v^\alpha
     =_s
     \ep^{-1}\sum_{i_1+i_2\le i+1,\ i_1,i_2>0}Z^{i_1}r Z^{i_2}r
     +\sum_{i_1+i_2+i_3\le i} Z ^{i_1}\Gamma^{i_2}
      Z ^{i_3}r,
     \label{4.14}
       \\
       &v^\alpha \nabla_\alpha Z ^i v^\beta
       +p'(\tau)(v^\alpha \nabla_\alpha Z ^i L v^\beta+\nabla^\beta Z ^i L)
       =_s
       \ep^{-1}\sum_{i_1+i_2\le i+1,\ i_1,i_2>0}Z^{i_1}r Z^{i_2}r
       +\sum_{i_1+i_2+i_3\le i} Z ^{i_1}\Gamma^{i_2}
      Z ^{i_3}r.
       \label{4.15}
   \end{align}
   Here one might suspect that the quadratic term in $r$ produces a shock, especially in the case where $i=i_1=i_2=1$. However, we will see in the third line of right hand side in proof of Proposition \ref{Prop4.4} that the extra smallness of $u_*\le\ep^2$ prevents a shock from formation.
   \subsection{$L^\infty$ bounds of Fluid Variables}\label{Section5.3}
   \begin{prop}\label{SingularEnergy2}
     Under the bootstrap assumptions, there exists $\epsilon_0=\epsilon_0(\Delta_1)>0$ such that for $\epsilon\le \epsilon_0$,
     \begin{align*}
     &\mathcal{E}_{2,\infty}\le C(\Delta_1)(1+\tilde{\mathcal{E}}_{4,2}^3).
     \end{align*}
   \end{prop}
   \begin{proof}
   Sobolev embedding as in Proposition \ref{PropSobolev} implies
   \begin{align*}
     &\sum_{ i\le2}\| Z ^i r\|_{L_{\underline{u}}^\infty L_{u}^\infty L^\infty(S)}
     \le C(\Delta_1) \epsilon^{-\frac{1}{2}}
     \sum_{ i\le4}\| Z ^i r\|_{L_{\underline{u}}^\infty L_{u}^2 L^2(S)}.
   \end{align*}
   On the other hand, \eqref{Eq5.1} implies
   \begin{align*}
     \sum_{ i\le1}\|f(u)^2\nabla_3 Z ^i r\|_{L_{u}^\infty L_{\underline{u}}^\infty L^\infty(S)}
     \le&\ C(\Delta_1)
     \sum_{i_1+i_2+i_3+i_4\le 2}\Big(
     \| Z ^{i_1} r Z ^{i_2} r Z ^{i_3} r\|_{L_{u}^\infty L_{\underline{u}}^\infty L^\infty(S)}
     \\
     &+
     \|f(u)^2 Z ^{i_1}\Gamma Z ^{i_2} r Z ^{i_3} r Z ^{i_4} r\|_{L_{u}^\infty L_{\underline{u}}^\infty L^\infty(S)}
     \Big),
   \end{align*}
   where $\Gamma$ is controlled by $\mathcal{O}_{1,\infty}\le\Delta_1$. Thus,
   \begin{align*}
     \mathcal{E}_{2,\infty}\le C(\Delta_1) (1+\tilde{\mathcal{E}}_{4,2}^3),
   \end{align*}
   concludes the proof.
   \end{proof}
   The choice of the null coordinates can now be seen to be essential in the following sense. The transport equations \eqref{4.12} - \eqref{4.13} involves Christoffel symbols. Specifically, the coefficient $\nabla_4^i\Gamma_{\underline{u}\underline{u}}^{\underline{u}}=-2\nabla_4^{i+1}\log\Omega$
   arising from $\nabla_4v^{\underline{u}}$ term in $\nabla_4^i$-commuted transport equations is controlled by
   \begin{align*}
     |\nabla_4^{i+1}\log\Omega|
     \le& \int_S |\nabla_4^{i+1}\slashed{\nabla}\log\Omega|^2
     +|\nabla_4^{i+1}\log\Omega|^2 \mathrm{d}\sigma
     \\
     =& \int_S|\nabla_4^{i+2}(\eta+\underline{\eta})|^2
     +|\nabla_4^{i+1}\log\Omega|^2\mathrm{d}\sigma,
   \end{align*}
   following \cite{luk2017weak}. Under this approach, we need $\tilde{\mathcal{O}}_{5,2}$ in controlling $\tilde{\mathcal{F}}_{4,2}$. Thus the bootstrap arguments break down. The choice of $e_3,e_4$ ensures
   \begin{align*}
     &\nabla_4\underline{w}
     =e_4(\Omega^2 v^{\underline{u}})
     =2\nabla_4\log\Omega v^{\underline{u}}
     +e_4v^{\underline{u}}
     =D_{4}v^{\underline{u}},
     \\
     &\nabla_4w
     =e_4v^u
     =D_{4}v^u
     +\Omega^{-2}\nabla_A\log\Omega v^u,
   \end{align*}
   where the noncontrollable term $\nabla_4\log\Omega$ does not arise in isolation.
   \subsection{Energy Currents}
   For simplicity, denote the commuted fluid variables by
   \begin{align*}
     \dot{r}= Z ^jr.
   \end{align*}
   Following \cite{Disconzi_2019}, we define the current for the fluid part to be
   \begin{align}
     \dot{J}^\mu
     =p'v^\mu\dot{L}^2
     +2p'\dot{u}^\mu\dot{L}
     +v^\mu\dot{u}^\alpha\dot{u}_\alpha.
     \label{Eq4.5}
   \end{align}
   \begin{prop}\label{Prop4.3}
     Under bootstrap assumptions as stated in Theorem \ref{ContinuousExist}, there exists $\epsilon_0=\epsilon_0(\Delta_1)>0$ such that, for $i\le4$, schematically
     \begin{align*}
       D_\mu\dot{J}^\mu &=_s Z ^{i}r (Z ^{i}r
       +\ep^{-1}\sum_{i_1+i_2\le i+1,\ i_1,i_2>0}Z^{i_1}rZ^{i_2}r
     +\sum_{ i_1+ i_2+ i_3\le i+1} Z ^{i_1}\Gamma^{i_2}
      Z ^{i_3}r),
     \end{align*}
     and
     \begin{align*}
       \dot{J}^{4} &=_s (1-(p')^{\frac{1}{3}}) Z ^{i}r Z ^{i}r
       +O(Z^{i}r\sum_{ i_1+ i_2\le i} Z ^{i_1} r
      Z ^{i_2}r),\\
     \dot{J}^{3} &=_s (1-(p')^{\frac{1}{3}}) Z ^{i}r Z ^{i}r
       +O(Z^{i}r\sum_{ i_1+ i_2\le i} Z ^{i_1} r
      Z ^{i_2}r).
     \end{align*}
   \end{prop}
   \begin{proof}
   Denote the lower order terms in Equations \eqref{4.14} - \eqref{4.15} by
   \begin{align}
     &v^\alpha D_\alpha\dot{L}+ D_\alpha\dot{v}^\alpha
     =A,
     \\
     &v^\alpha D_\alpha\dot{v}^\beta
     +p'(v^\alpha D_\alpha\dot{L}v^\beta+D^\beta \dot{L})
     =A^\beta,
   \end{align}
   where schematically, eliminating $\nabla_3r$ by Proposition \ref{SingularEnergy} - \ref{SingularEnergy2},
   \begin{align*}
     A,A^\beta &=_s
     \sum_{ i_1+ i_2+ i_3\le i+1} Z ^{i_1}\Gamma^{i_2}
      Z ^{i_3}r.
   \end{align*}
   Notice moreover that the normalization condition $v^\alpha v_\alpha=-1$ implies, schematically
   \begin{align*}
     v^\alpha\dot{v}_\alpha=B=_s
     \sum_{i_1+i_2\le i} Z ^{i_1} r
      Z ^{i_2}r.
   \end{align*}
   Then
   \begin{align}
   \begin{split}
     \dot{v}^\alpha\dot{v}_\alpha
     =&\ |\dot{\slashed{v}}|_\gamma^2
     -2\frac{2v_A\dot{v}^A-2\underline{w}\dot{w}-B}{2w}\dot{w}
     =|\dot{\slashed{v}}|_\gamma^2
     -2\frac{2v_A\dot{v}^A-2w\dot{\underline{w}}-B}{2\underline{w}}\dot{\underline{w}}
     \\
     \ge&\max\{\frac{1}{2w\underline{w}}|\dot{\slashed{v}}|_\gamma^2,
     \max\{
     \frac{1}{w^2}\dot{w}^2,\
     \frac{1}{\underline{w}^2}\dot{\underline{w}}^2
     \}
     \}+B\frac{\dot{w}}{w}+B\frac{\dot{\underline{w}}}{\underline{w}}.
     \end{split}
     \label{Eq4.8}
   \end{align}
   We compute the divergence of $\dot{J}$ to obtain
   \begin{align}
   \begin{split}
      D_\mu\dot{J}^\mu
     =&\ p' D_\mu v^\mu\dot{L}^2+2 D_\mu v^\mu\dot{v}^\alpha\dot{v}_\alpha
     +p'' \frac{D_\mu L}{P+\tau}(v^\mu\dot{L}^2+2\dot{v}^\mu\dot{L})
     \\
     &+2p'\dot{L}(v^\mu D_\mu\dot{L}+ D_\mu\dot{v}^\mu)
     +2(p'\dot{v}^\mu D_\mu\dot{L}+v^\mu D_\mu\dot{v}^\alpha\dot{v}_\alpha)
     \\
     =&\ p' D_\mu v^\mu\dot{L}^2+2 D_\mu v^\mu\dot{v}^\alpha\dot{v}_\alpha
     +p'' \frac{D_\mu L}{P+\tau}(v^\mu\dot{L}^2+2\dot{u}^\mu\dot{L})
     +2p'\dot{L} A
     +2\dot{v}^\beta A_\beta
     \ .
     \end{split}
     \label{Eq4.11}
   \end{align}
   Also, according to Equation \eqref{Eq4.8},
   \begin{align}
   \begin{split}
     \dot{J}^4
     =&p'w\dot{L}^2+2p'\dot{w}\dot{L}
     +w\dot{u}^\alpha\dot{u}_\alpha\\
     \ge&\ (p'-(p')^{\frac{4}{3}})w\dot{L}^2
     +((p')^{\frac{1}{3}}-(p')^{\frac{2}{3}})\frac{\dot{w}^2}{w}\\
     &
     +\frac{1-(p')^{\frac{1}{3}}}{2}(\frac{\dot{\underline{w}}^2}{\underline{w}}
     +\frac{1}{2w\underline{w}}|\dot{\slashed{v}}|^2)
     +B\frac{\dot{w}}{w}+B\frac{\dot{\underline{w}}}{\underline{w}}\ .
     \end{split}
     \label{Eq4.9}
   \end{align}
   Similarly proof applies for $\dot{J}^3$.
   \end{proof}
   \begin{prop}\label{Prop4.4}
      Under the bootstrap assumptions, there exists $\epsilon_0=\epsilon_0(\Delta_1)$ such that for $\epsilon\le \epsilon_0$, the following energy estimates holds:
     \begin{align*}
     (1-(p')^{\frac{1}{3}})\tilde{\mathcal{E}}_{4,2}^2
     \le C(\mathcal{E}_{ini})(1
     + \epsilon C(\Delta_1,\mathcal{E}_{2,\infty})).
     \end{align*}
   \end{prop}
   \begin{proof}
   Applying Inequalities \eqref{Eq4.11}, \eqref{Eq4.9} to the identity
   \begin{align*}
     \iint_{D_{u,\underline{u}}} D_\mu \dot{J}^\mu
     =\iint_{D_{u,\underline{u}}} \Gamma_\mu\dot{J}^\mu
     +\int_{H_{u}}\dot{J}^3+\int_{\underline{H}_{\underline{u}}}\dot{J}^4
     -\int_{H_0}\dot{J}^3-\int_{\underline{H}_0}\dot{J}^4,
   \end{align*}
   for all $(u,\underline{u})\in[0,u_*)\times[0,\underline{u}_*)$ with the $ Z ^i$-commuted current, $ i\le 4$, implies
   \begin{align*}
     &(1-(p')^{\frac{1}{3}})
     (\| Z ^i r\|_{L_{\underline{u}}^\infty L_{u}^2L^2(S)}^2
     +\| Z ^i r\|_{L_{u}^\infty L_{\underline{u}}^2L^2(S)}^2)
     \\
     &-(1-(p')^{\frac{1}{3}})(\| Z ^i r\|_{L_{\underline{u}}^\infty L_{u}^2L^2(S_{u,0})}^2
     +\| Z ^i r\|_{L_{u}^\infty L_{\underline{u}}^2L^2(S_{0,\underline{u}})}^2)
     \\
     \lesssim&\iint_{D_{u,\underline{u}}} (|D_\mu\dot{J}^\mu|+|\Gamma\|\dot{J}|)
     +(\int_{\underline{H}_{\underline{u}}}+\int_{H_u})(|\sum_{ i_1+ i_2\le i} Z ^{i_1} r
      Z ^{i_2}r|(|\frac{\dot{w}}{w}|+|\frac{\dot{\underline{w}}}{\underline{w}}|))
     \\
     \lesssim&\epsilon^{\frac{1}{2}}
     \sum_{ i_1+ i_2\le i, i_2\le 2}
     \|f(u)^{-1} Z ^ir\|_{L_{\underline{u}}^2 L_{u}^2 L^2(S)}
     \|f(u) Z ^{i_1}(\psi,\psi_{\underline{H}})\|_{L_{\underline{u}}^\infty L_{u}^2 L^2(S)}
     \| Z ^{i_2}r\|_{L_{\underline{u}}^\infty L_{u}^\infty L^\infty(S)}
     \\
     &+
     \sum_{ i_1+ i_2\le i, i_1\le 2}
     \|f(u)^{-1} Z ^ir\|_{ L_{\underline{u}}^2 L_{u}^2 L^2(S)}
     \|f(u) Z ^{i_1}(\psi,\psi_{\underline{H}})\|_{L_{u}^2 L_{\underline{u}}^\infty L^\infty(S)}
     \| Z ^{i_2}r\|_{L_{u}^\infty L_{\underline{u}}^2 L^2(S)}
     \\
     &+\ep^{-1}\epsilon^2\sum_{i_1+i_2\le i+1,\ i_1\ge i_2>0}
     \|Z^ir\|_{L_u^\infty L_{\ub}^2 L^2(S)}
     \|Z^{i_1}r\|_{L_u^\infty L_{\ub}^2 L^2(S)}
     \|Z^{i_2}r\|_{L_u^\infty L_{\ub}^\infty L^\infty(S)}
     \\
     &+\epsilon^{\frac{1}{2}}
     \sum_{ i_1+ i_2\le i, i_2\le 2}
     \| Z ^ir\|_{L_{u}^2 L_{\underline{u}}^2 L^2(S)}
     \| Z ^{i_1}\psi_{H}\|_{L_{u}^\infty L_{\underline{u}}^2 L^2(S)}
     \| Z ^{i_2}r\|_{L_{u}^\infty L_{\underline{u}}^\infty L^\infty(S)}
     \\
     &+
     \sum_{ i_1+ i_2\le i, i_1\le 2}
     \| Z ^ir\|_{ L_{u}^2 L_{\underline{u}}^2 L^2(S)}
     \| Z ^{i_1}\psi_{H}\|_{L_{\underline{u}}^2 L_{u}^\infty L^\infty(S)}
     \| Z ^{i_2}r\|_{L_{\underline{u}}^\infty L_{u}^2 L^2(S)}
     \\
     &+\sum_{ i_1+ i_2\le i
     }
     \| Z ^{i_1} r\|_{L_{\underline{u}}^\infty L_{u}^2 L^2(S)}
     \| Z ^{i_2} r\|_{L_{\underline{u}}^\infty L_{u}^2 L^2(S)}
     \\
     &+\sum_{ i_1+ i_2\le i
     }
     \| Z ^{i_1} r\|_{L_{u}^\infty L_{\underline{u}}^2 L^2(S)}
     \| Z ^{i_2} r\|_{ L_{u}^\infty L_{\underline{u}}^2 L^2(S)}.
   \end{align*}
   Multiplying both sides by $\epsilon^{-1}$, we conclude that
   \begin{align}
     (1-(p')^{\frac{1}{3}})\tilde{\mathcal{E}}_{i,2}^2
     \le C(\mathcal{E}_{ini})(1+C(\Delta_1)\tilde{\mathcal{E}}_{i-1,2}^2
     + \epsilon C(\Delta_1,\mathcal{E}_{2,\infty})
     \tilde{\mathcal{E}}_{i,2}(1+\tilde{\mathcal{E}}_{i,2})).
     \label{Eq4.10}
   \end{align}
   Summing up the Inequalities \eqref{Eq4.10} with $i=3,4$, and noticing that
   \begin{align*}
     \tilde{\mathcal{E}}_{2,2}^2 &\le \epsilon\mathcal{E}_{2,\infty}^2,
   \end{align*}
   we arrive at the proposition.
   \end{proof}

   \section{Conclusion of proof of the a priori estimates\\ and the proof of continuous extendibility}\label{Section6}
   We first prove Theorem \ref{SmoothExist}.
   As in Section \ref{Section3.1}, Theorem \ref{SmoothExist} follows from Theorem \ref{AprioriEstimates}. In Section \ref{Section5} and Section \ref{Section6}, we have proved the estimates \eqref{Eq3.1} and \eqref{Eq3.001} for the case of $\delta=u_*+\underline{u}_*$. It remains to prove \eqref{Eq3.1} and \eqref{Eq3.001} for general $\delta$. The estimates in Proposition \ref{Prop3.7} - Proposition \ref{Prop3.12} and Proposition \ref{SingularEnergy} - Proposition \ref{Prop4.3} are derived from transport equations. Therefore, the proof applies without modification. For the energy estimates in Proposition \ref{Prop3.13} and Proposition \ref{Prop4.4}, the extra boundary terms arising from the spacelike boundary $\Sigma_\delta$ have the same sign as the main terms and thus do not affect the estimates.
   
   We next prove the Theorem \ref{ContinuousExist} following \cite{luk2017weak}. For continuity of the metric on the singular boundary, it suffices to notice that the terms in Equations \eqref{Eq3.3} - \eqref{Eq3.6} concerning Ricci curvature are all bounded by $\mathcal{E}_{2,\infty}$. For the regularity statement of Ricci coefficients, it suffices to notice that $\mathrm{Ric}$ arising in transport equations of $\nabla_3(\psi,\psi_{\underline{H}}),\ \nabla_4(\psi,\psi_H)$ lie in $L^1_uL^\infty$. Continuity of the fluid variables up to the boundary, follows from Proposition \ref{PropTransport} applied to Equations \eqref{Eq2.51}, \eqref{Eq2.52}, since the lower order terms are all bounded by $\mathcal{E}_{2,\infty}$.
   
   Finally, as promised in Corollary \ref{Rem1.1}, we prove that the family of metrics satisfying \eqref{CurvatureBlowup} does not admit a $C_{loc}^{0,1}$-extension beyond the singular boundary $\{u=u_*\}$. The existence of the family of initial data follows by imposing \eqref{Eq2.7001} as in the proof of Proposition \ref{Prop3.1}.
   \begin{proof}[Proof of Corollary \ref{Rem1.1}]
   By Theorem \ref{AprioriEstimates}, $\psi_{\underline{H}},f(u)\psi_{\underline{H}},f(u)\underline{\beta},K,\check{\sigma},f(u)\partial\mathrm{Ric}$
   are in $L^\infty_{\underline{u}}L^2_uL^\infty(S)$ as $u\rightarrow u_*$ and $\psi$ remains bounded. The transport equation \eqref{Eq3.30} of $|\underline{\alpha}|$ schematically says
   \begin{align*}
     \nabla_4\underline{\alpha}
     =_s\slashed{\nabla}\underline{\beta}
     +\mathrm{tr}\chi \underline{\alpha}
     +\omega\underline{\alpha}
     +\psi_{\underline{H}}\cdot(K,\sigma)
     +\psi_{\underline{H}}^2\psi_H
     +\nabla\slashed{\mathrm{Ric}}
     +\Gamma\slashed{\mathrm{Ric}},
   \end{align*}
   where $\omega=0$. Apply Gr\"onwall's inequality, for any $u'<u_*,\ \underline{u}'<\underline{u}_*$,
   \begin{align*}
     &\int_0^{u'}|\underline{\alpha}(u,\underline{u}')-\underline{\alpha}(u,0)|du
     \\
     \le& C(\Delta_1,\mathcal{E}_{2,\infty}) \int\Big(\exp(\int\mathrm{tr}\chi d\underline{u})\int (\slashed{\nabla}\underline{\beta}
     +\psi_{\underline{H}}\cdot(K,\sigma)
     +\psi_{\underline{H}}^2\psi_H
     +\partial\slashed{\mathrm{Ric}})
     d\underline{u}\Big)du
     \\
     \le& C(\Delta_1,\mathcal{E}_{2,\infty})
     \exp(\tilde{\mathcal{O}}_{4,2})
     \epsilon^{\frac{1}{2}}(\mathcal{O}_{1,\infty}+\mathcal{R}_3+\mathcal{E}_{2,\infty}).
   \end{align*}
   Therefore, the bounds of $|\underline{\alpha}|$ in \eqref{CurvatureBlowup} is propagated up to the singularity as:
   \begin{align}\label{Eq6.2}
       \int_0^{u'}|\underline{\alpha}_{11}|(u,\underline{u}),\
       \int_0^{u'}|\underline{\alpha}_{22}|(u,\underline{u}) & \gtrsim \int_0^{u'}\frac{1}{(u_*-u)^2\log^p(\frac{1}{u_*-u})}du,
   \end{align}
   while
   \begin{align}\label{Eq6.3}
       \int_0^{u'}|\underline{\alpha}_{12}|(u,\underline{u}) \lesssim \int_0^{u'}\frac{1}{(u_*-u)^2\log^{2p}(\frac{1}{u_*-u})}du.
     \end{align}
   According to \cite{Sbierski2024}, we have the following theorem:
    \begin{thm}\cite[Theorem 3.13]{Sbierski2024}
     Let $(M,g)$ be a Lorentzian manifold in a double null gauge. Assume that the metric $g$ extends in $(u,\underline{u},\theta^A)$-coordinates extends continuously as a Lorentzian metric to the manifold with boundary $\overline{M}=(-1,1)\times(-1,0]\times\mathbb{T}^2$. Moreover, assume that for some $0<C<\infty$,
     \begin{align*}
       \sup_M\{|\omega|+|\eta|+|\underline{\eta}|+|\chi|\}\le C,\
       \int_{-1}^{0}\sup_{(\underline{u},\theta)\in(-1,1)\times\mathbb{T}^2}(|\chi|+|\slashed{\nabla}b|)du\le C.
     \end{align*}
     Assume in addition that for every $\bar{p}\in \partial\overline{M}$ and any neighborhood $\overline{W}\subset\overline{M}$ of $\bar{p}$, there exists a $\bar{q}\in \partial\overline{M}\cap\overline{W}$ and a compact neighborhood $\overline{V}\subset\overline{W}$ of $\bar{q}$ such that, for all $u'\le 0$ close enough to $0$, the set $\overline{V}\cap\{u\le u'\}$ is also a compact smooth manifold with corners, and there exist
       continuous vector fields $\overline{X_i}$ on $\overline{V}$, $i=1,2,3,4$, and an $\epsilon>0$ such that
   \begin{align*}
     &|\int_{\overline{D_{u_k,\underline{u}_*}}} R(\hat{X}_1,\hat{X}_2,\hat{X}_3,\hat{X}_4)|\rightarrow\infty,
   \end{align*}
   for any $|\hat{X}_i-\overline{X_i}|<\epsilon$ and $u_k\rightarrow u_*$. Then there exists no $C_{loc}^{0,1}$ extension of $D_{u_*,\underline{u}_*}$ across the singular boundary $\{u=u_*\}$.
    \end{thm}
     By $\mathrm{tr}_\gamma\underline{\alpha}=0$ as a Weyl tensor, we see
     \begin{align*}
       &|\underline{\alpha}_{22}|
       \le C(\Delta_1)(|\underline{\alpha}_{11}|+|\underline{\alpha}_{12}|).
     \end{align*}
     Apply \eqref{Eq6.2} - \eqref{Eq6.3}, we have
     \begin{align*}
       &\int^{u_k}_0\iint|\underline{\alpha}| d\theta d\underline{u} du
       \le C(\Delta_1)
       \int^{u_k}_0\iint\Big{(}|\underline{\alpha}_{11}|
       +\frac{1}{(u_*-u)^2\log^p(\frac{1}{u_*-u})}\Big{)}d\theta d\underline{u} du.
     \end{align*}
     On the other hand, by Cauchy--Schwarz inequality, the boundedness of $\mathcal{R}_3$ imply that
     \begin{align*}
       &\int^{u_k}_0\iint|(K,\check{\sigma},\alpha,\beta,\underline{\beta})|d\theta d\underline{u} du
       \\
       \le&
     \epsilon^{-\frac{1}{2}}\| \beta\|_{L_{u}^\infty L_{\underline{u}}^2 L^2(S)}
     +\epsilon^{-\frac{1}{2}}\| (K,\check{\sigma})\|_{L_{\underline{u}}^\infty L_{u}^2L^2(S)}
     +\epsilon^{-\frac{1}{2}}\|f(u) \underline{\beta}\|_{L_{\underline{u}}^\infty L_{u}^2L^2(S)}
     +\|\alpha\|_{L_u^\infty L_{\underline{u}}^2 L^2(S)}
     \\
       \le& \mathcal{R}_3
       \le C(\Delta_1)
       \le C(\Delta_1)\int^{u_k}_0\iint \frac{1}{(u_*-u)^2\log^p(\frac{1}{u_*-u})}d\theta d\underline{u} du.
     \end{align*}
     Therefore, the condition is satisfied with $e_1,e_3,e_1,e_3$, since $\underline{\alpha}_{11}=W_{1313}$ implies that
   \begin{align*}
     &\Big{|}\int^{u_k}_0\iint R(\hat{e}_1,\hat{e}_3,\hat{e}_1,\hat{e}_3)d\theta d\underline{u} du\Big{|}
     \\
     \ge& \int^{u_k}_0\iint|\underline{\alpha}_{11}|-\epsilon |R|d\theta d\underline{u} du
     \\
     \ge& \int^{u_k}_0\iint(C_1-\epsilon C_2(\Delta_1))\frac{1}{(u_*-u)^2\log^p(\frac{1}{u_*-u})}d\theta d\underline{u} du
     \rightarrow\infty.
   \end{align*}
   Thus, this implies the $C_{loc}^{0,1}$-inextensibility of the weak null singularity.
   \end{proof}
   
   \subsection*{Acknowledgements}
   The author thanks Jonathan Luk for suggesting the problem and sharing many insights from the works \cite{luk2017weak}, as well as offering valuable comments on an earlier version of the manuscript. The author also thanks Maxime Van de Moortel for his generous help. The author also thanks Dawei Shen and Jingbo Wan for suggestions that improved the proof. This work is supported by the NSF DMS-2304445.
   %

   \appendix
   \section{Equations}\label{Equations}
   Following \cite{Christodoulou1989-1990}, we derive the equations for the Ricci coefficients and curvature components in the modified double null coordinates as in Theorem \ref{ContinuousExist}. In Section \ref{SchematicEquations}, we derive the commuted equations from the equations below. 
   To start, we record the Christoffel symbols
   \begin{align}
    &D_Ae_4=-\zeta_Ae_4+\chi_{AB}e_B,\ &&
    D_Ae_3=\zeta_Ae_3+\chib_{AB}e_B,
    \label{EqA.01}\\
    &D_3e_4=2\omb e_4+2\eta_A e_A,\ &&
    D_4e_3=2\om e_3+2\etab_Ae_A,
    \label{EqA.02}\\
    &D_4e_A=\etab e_4+\chi_{AB}e_B,\ &&
    D_3e_A=\eta_Ae_3+\chib_{AB}e_B-D_Ab^B\Om^{-1}e_B,
    \label{EqA.03}\\
    &D_4e_4=-2\om e_4,\ &&
    D_3e_3=-2\omb e_3,
    \label{EqA.04}\\
    &D_Ae_B=\nabla_Ae_B+\frac{1}{2}\chib_{AB}e_4
    +\frac{1}{2}\chi_{AB}e_3.&&
    \label{EqA.05}
\end{align}
And the commutation formulae
\begin{align}
  &[e_A,e_4]=-\nabla_A\log\Om e_4,\ \quad\quad
  [e_A,e_3]=-\nabla_A\log\Om e_3+D_Ab^B\Om^{-1}e_B,
  \label{EqA.06}\\
  &[e_3,e_4]=\nabla_4\log\Om e_3-\nabla_3\log\Om e_4-D_4b^B\Om^{-1}e_B.
  \label{EqA.07}
\end{align}
   Firstly, $\chi$ and $\underline{\chi}$ obey the following null structure equations:
   \begin{align}
     &\nabla_4\mathrm{tr}\chi
     +\frac{1}{2}(\mathrm{tr}\chi)^2
     +2\underline{\omega}\slashed{\mathrm{tr}}\chi
     =-|\hat{\chi}|^2
     -\mathrm{Ric}_{44},
     \label{Eq2.1}
     \\
     &\nabla_4\hat{\chi}+\mathrm{tr}\chi\hat{\chi}
     =-2\omega\hat{\chi}-\alpha,
     \label{Eq2.2}
     \\
     &\nabla_3\mathrm{tr}\underline{\chi}
     +\frac{1}{2}(\mathrm{tr}\underline{\chi})^2
     +2\underline{\omega}\slashed{\mathrm{tr}}\underline{\chi}
     =-|\hat{\underline{\chi}}|^2
     -\mathrm{Ric}_{33},
     \label{Eq2.3}
     \\
     &\nabla_3\hat{\underline{\chi}}+\mathrm{tr}\underline{\chi}\hat{\underline{\chi}}
     =-2\underline{\omega}\hat{\underline{\chi}}-\underline{\alpha},
     \label{Eq2.4}
     \\
     &\nabla_4\mathrm{tr}\underline{\chi}+\frac{1}{2}\mathrm{tr}\chi\mathrm{tr}\underline{\chi}
     =2\omega\mathrm{tr}\underline{\chi}+2\rho+\frac{1}{3}R-\hat{\chi}\cdot\underline{\hat{\chi}}
     +2\mathrm{div}\underline{\eta}+2|\underline{\eta}|^2,
     \label{Eq2.5}
     \\
     &\nabla_4\hat{\underline{\chi}}+\frac{1}{2}\mathrm{tr}\chi\hat{\underline{\chi}}
     =\nabla\hat{\otimes}\underline{\eta}+2\omega\hat{\underline{\chi}}
     -\frac{1}{2}\mathrm{tr}\underline{\chi}\hat{\chi}+\underline{\eta}\hat{\otimes}\underline{\eta}
     +\widehat{\mathrm{Ric}},
     \label{Eq2.6}
     \\
     &\nabla_3\mathrm{tr}\chi+\frac{1}{2}\mathrm{tr}\underline{\chi}\mathrm{tr}\chi
     =2\underline{\omega}\mathrm{tr}\chi+2\rho+\frac{1}{3}R-\hat{\chi}\cdot\underline{\hat{\chi}}
     +2\mathrm{div}\eta+2|\eta|^2,
     \label{Eq2.7}
     \\
     &\nabla_3\hat{\chi}+\frac{1}{2}\mathrm{tr}\underline{\chi}\hat{\chi}
     =\nabla\hat{\otimes}\eta+2\underline{\omega}\hat{\chi}
     -\frac{1}{2}\mathrm{tr}\chi\hat{\underline{\chi}}+\eta\hat{\otimes}\eta
     +\widehat{\mathrm{Ric}},
     \label{Eq2.8}
     \\
     &\slashed{\mathrm{div}}\hat{\chi}
     +\zeta\hat{\chi}
     =\frac{1}{2}\nabla\slashed{\mathrm{tr}}\chi
     +\frac{1}{2}\slashed{\mathrm{tr}}\chi \zeta
     -\beta,
     \label{Eq2.9}
     \\
     &\slashed{\mathrm{div}}\hat{\underline{\chi}}
     -\zeta\hat{\chi}
     =\frac{1}{2}\nabla\slashed{\mathrm{tr}}\underline{\chi}
     +\frac{1}{2}\slashed{\mathrm{tr}}\underline{\chi} \zeta
     +\underline{\beta}.
     \label{Eq2.10}
   \end{align}
   Here $\widehat{\slashed{\mathrm{Ric}}}=\slashed{\mathrm{Ric}}-(\frac{1}{2}R+\mathrm{Ric}_{34})\gamma$ denotes the traceless part of the Ricci curvature restricted to the tangent directions on the tori $S_{u,\underline{u}}$.
   The other Ricci coefficients satisfy the following null structure equations:
   \begin{align}
     &\nabla_4\eta=-\chi\cdot(\eta-\underline{\eta})-\beta
     -\frac{1}{2}\mathrm{Ric}_{4\cdot},
     \label{Eq2.11}
     \\
     &\nabla_3\underline{\eta}=-\underline{\chi}\cdot(\underline{\eta}-\eta)+\underline{\beta}
     -\frac{1}{2}\mathrm{Ric}_{3\cdot},
     \label{Eq2.12}
     \\
     &\nabla_4\underline{\omega}=2\omega\underline{\omega}-\eta\cdot\underline{\eta}
     +\frac{1}{2}|\eta|^2+\frac{1}{2}\rho+\frac{1}{4}\mathrm{Ric}_{34}+\frac{1}{12}R,
     \label{Eq2.13}
     \\
     &\nabla_3\omega=2\omega\underline{\omega}-\eta\cdot\underline{\eta}+\frac{1}{2}|\underline{\eta}|^2
     +\frac{1}{2}\rho+\frac{1}{4}\mathrm{Ric}_{34}+\frac{1}{12}R.
     \label{Eq2.14}
   \end{align}
   Also
   \begin{align}
        &\nabla_3\eta+\nabla_3\underline{\eta} =\slashed{\nabla}\underline{\omega},
        \label{Eq2.15}
        \\
     &\nabla_4\eta+\nabla_4\underline{\eta} =\slashed{\nabla}\omega.
     \label{Eq2.16}
   \end{align}
   The metric components $b,\gamma,\slashed{\nabla}\log\Omega$ satisfy equations:
   \begin{align}
     &\Om^{-2}\p_{\ub}b=2\eta-2\underline{\eta},
     \label{Eq2.17}
     \\
     &\Om^{-2}\p_{\ub}\gamma=2\chi,
     \label{Eq2.18}
     \\
     &\slashed{\nabla}\log\Omega
     =\frac{1}{2}(\eta+\underline{\eta}),
     \label{Eq2.21}
     \\
     &\nabla_4\log\Omega
     =-2\omega,
     \label{Eq2.22.1}
     \\
     &\nabla_3\log\Omega
     =-2\underline{\omega}.
     \label{Eq2.22}
   \end{align}
   We define Riemann curvature tensors by
   \begin{align*}
     & \bar{\alpha}_{AB}=R(e_A,e_4,e_B,e_4),\
     \underline{\bar{\alpha}}_{AB}=R(e_A,e_3,e_B,e_3),
     \\
     & \bar{\beta}_A=\frac{1}{2}R(e_A,e_4,e_3,e_4),\
     \underline{\bar{\beta}}_A=\frac{1}{2}R(e_A,e_3,e_3,e_4),
     \\
     & \bar{\rho}=\frac{1}{4}R(e_4,e_3,e_4,e_3),\
     \bar{\sigma}=\frac{1}{4}*R(e_4,e_3,e_4,e_3),
   \end{align*}
   where, as compared with Weyl tensors,
   \begin{align*}
     &\alpha=\bar{\alpha}-\frac{1}{2}\slashed{\gamma}\mathrm{Ric}_{44},\
     \underline{\alpha}=\underline{\bar{\alpha}}-\frac{1}{2}\slashed{\gamma}\mathrm{Ric}_{33},
     \\
     &\beta=\bar{\beta}-\frac{1}{2}\mathrm{Ric}_{4\cdot},\
     \underline{\beta}=\underline{\bar{\beta}}+\frac{1}{2}\mathrm{Ric}_{3\cdot},
     \\
     &\rho=\bar{\rho}-\frac{1}{2}\mathrm{Ric}_{34}-\frac{1}{6}R,\
     \sigma=\bar{\sigma}.
   \end{align*}
   Equations in terms of curvature can be deduced from:
   \begin{align*}
     &R_{A3B3}=\underline{\bar{\alpha}}_{AB},\
     R_{A4B4}=\bar{\alpha}_{AB},
     \\
     &R_{A334}=2\underline{\bar{\beta}}_A,\
     R_{A434}=2\bar{\beta}_A,
     \\
     &R_{3434}=4\bar{\rho},\
     R_{AB34}=2\sigma\slashed{\epsilon}_{AB},
     \\
     &R_{A3BC}=*(\underline{\bar{\beta}}+\mathrm{Ric}_{3\cdot})_A\slashed{\epsilon}_{BC}
     =\slashed{\gamma}_{AB}(\underline{\bar{\beta}}+\mathrm{Ric}_{3\cdot})_C
     -\slashed{\gamma}_{AC}(\underline{\bar{\beta}}+\mathrm{Ric}_{3\cdot})_B,
     \\
     &R_{A4BC}=-*(\bar{\beta}+\mathrm{Ric}_{4\cdot})_A\slashed{\epsilon}_{BC}
     =-\slashed{\gamma}_{AB}(\bar{\beta}-\mathrm{Ric}_{4\cdot})_C
     +\slashed{\gamma}_{AC}(\bar{\beta}-\mathrm{Ric}_{4\cdot})_B,
     \\
     &R_{A4B3}=(-\bar{\rho}+\frac{1}{2}R+\mathrm{Ric}_{34})\slashed{\gamma}_{AB}
     +\bar{\sigma}\slashed{\epsilon}_{AB}-\slashed{\mathrm{Ric}}_{AB},
     \\
     &R_{ABCD}=(-\bar{\rho}+\frac{1}{2}R+\mathrm{Ric}_{34})\slashed{\epsilon}_{AB}\slashed{\epsilon}_{CD}.
   \end{align*}
   Following \cite{Christodoulou1989-1990}, we derive the following equations:
   \begin{align*}
     &D_4R_{A4B4}=\nabla_4\bar{\alpha}_{AB}-4\omega\bar{\alpha}_{AB},
     \\
     &D_3R_{A3B3}=\nabla_3\underline{\bar{\alpha}}_{AB}-4\underline{\omega}\underline{\bar{\alpha}}_{AB},
     \\
     &D_3R_{A4B4}=\nabla_3\bar{\alpha}_{AB}+4\underline{\omega}\bar{\alpha}_{AB}
     -8(\eta\hat{\otimes}(\bar{\beta}+\frac{1}{2}\mathrm{Ric}_{4\cdot}))_{AB},
     \\
     &D_4R_{A3B3}=\nabla_4\underline{\bar{\alpha}}_{AB}+4\omega\underline{\bar{\alpha}}_{AB}
     +8(\underline{\eta}\hat{\otimes}(\underline{\bar{\beta}}-\frac{1}{2}\mathrm{Ric}_{3\cdot}))_{AB},
     \\
     &D_AR_{B4C4}=\nabla_A\bar{\alpha}_{BC}
     -(\chi\hat{\otimes}(\bar{\beta}+\mathrm{Ric}_{4\cdot}))_{ABC}+2\zeta_A\bar{\alpha}_{BC},
     \\
     &D_AR_{B3C3}=\nabla_A\underline{\bar{\alpha}}_{BC}
     -(\underline{\chi}\hat{\otimes}(\underline{\bar{\beta}}-\mathrm{Ric}_{3\cdot}))_{ABC}
     -2\zeta_A\underline{\bar{\alpha}}_{BC},
     \\
     &D_4R_{A434}=2\nabla_4\bar{\beta}_A+4\omega\bar{\beta}_A
     -2\underline{\eta}^B\bar{\alpha}_{AB},
   \\
     &D_3R_{A334}=2\nabla_3\underline{\bar{\beta}}_A
     +4\underline{\omega}\underline{\bar{\beta}}_A
     +2\eta^B\underline{\bar{\alpha}}_{AB},
     \\
     &D_3R_{A434}=2\nabla_3\bar{\beta}_A-4\underline{\omega}\bar{\beta}_A
     +\eta_A(-6\bar{\rho}+R+2\mathrm{Ric}_{34})-2(\slashed{\mathrm{Ric}}\cdot\eta)_A
     -6*\eta_A\bar{\sigma},
     \\
     &D_4R_{A334}=2\nabla_4\underline{\bar{\beta}}_A-4\omega\underline{\bar{\beta}}_A
     -\underline{\eta}_A(-6\bar{\rho}+R+2\mathrm{Ric}_{34})-2(\slashed{\mathrm{Ric}}\cdot\underline{\eta})_A
     -6*\underline{\eta}_A\bar{\sigma},
     \\
     &D_AR_{B434}=2\nabla_A\bar{\beta}_B
     +2\zeta_A\bar{\beta}_B
     -\underline{\chi}_{A}{}^{C}\bar{\alpha}_{BC}
     +\chi_{AB}(-3\bar{\rho}+\frac{1}{2} R+\mathrm{Ric}_{34})
     -3\chi_{CB}\slashed{\epsilon}_A{}^C\bar{\sigma}-\chi\cdot*\slashed{\mathrm{Ric}},
     \\
     &D_AR_{B334}=2\nabla_A\underline{\bar{\beta}}_B
     -2\zeta_A\underline{\bar{\beta}}_B
     +\underline{\chi}_{A}{}^{C}\underline{\bar{\alpha}}_{BC}
     -\underline{\chi}_{AB}(-3\bar{\rho}+\frac{1}{2} R+\mathrm{Ric}_{34})
     +3\underline{\chi}_{CB}\slashed{\epsilon}_A{}^C\bar{\sigma}+\underline{\chi}\cdot*\slashed{\mathrm{Ric}},
     \\
     &D_4R_{3434}=4\nabla_4\bar{\rho}-8\underline{\eta}^A\bar{\beta}_A,
     \\
     &D_3R_{3434}=4\nabla_3\bar{\rho}+8\eta^A\underline{\bar{\beta}}_A,
     \\
     &D_AR_{3434}=4\nabla_A\bar{\rho}+4(\chi_A{}^B\underline{\bar{\beta}}_B-\underline{\chi}_A{}^B\bar{\beta}_B),
     \\
     &D_4R_{AB34}=2\slashed{\epsilon}_{AB}\nabla_4\bar{\sigma}
     +4(\underline{\eta}_A(\bar{\beta}_B-\frac{1}{2}\mathrm{Ric}_{4B})
     -\underline{\eta}_B(\bar{\beta}_A-\frac{1}{2}\mathrm{Ric}_{4A})),
     \\
     &D_3R_{AB34}=2\slashed{\epsilon}_{AB}\nabla_3\bar{\sigma}
     +4(\eta_A(\underline{\bar{\beta}}_B+\frac{1}{2}\mathrm{Ric}_{3B})
     -\eta_B(\underline{\bar{\beta}}_A+\frac{1}{2}\mathrm{Ric}_{3A})),
     \\
     &D_AR_{BC34}=2\slashed{\epsilon}_{BC}\nabla_A\bar{\sigma}
     +2(\chi_{AB}(\bar{\underline{\beta}}_C+\frac{1}{2}\mathrm{Ric}_{3C})
     -\chi_{AC}(\bar{\underline{\beta}}_B+\frac{1}{2}\mathrm{Ric}_{3B}))
     \\
     &\qquad\qquad\qquad+2(\underline{\chi}_{AB}(\bar{\beta}_C-\frac{1}{2}\mathrm{Ric}_{4C})
     -\underline{\chi}_{AC}(\bar{\beta}_B-\frac{1}{2}\mathrm{Ric}_{4B})).
   \end{align*}
   For curvature components, we have the following equations:
   \begin{align}
     \begin{split}
     &\nabla_3\sigma+\frac{3}{2}\mathrm{tr}\underline{\chi}\sigma+\slashed{\mathrm{div}}*\underline{\beta}
     +(2\eta-\zeta)\wedge\underline{\beta}+\frac{1}{2}\hat{\chi}\wedge\underline{\alpha}
     \\
     &=\frac{1}{2}(\slashed{\epsilon}\cdot \nabla\slashed{\mathrm{Ric}}_{3\cdot}
     -\slashed{\epsilon}\cdot(\underline{\chi}\times\slashed{\mathrm{Ric}}
     +2\eta\cdot\slashed{\mathrm{Ric}}_{3\cdot}),
     \end{split}
     \label{Eq2.23}
     \\
     \begin{split}
     &\nabla_4\sigma+\frac{3}{2}\mathrm{tr}\chi\sigma
     +\mathrm{div}*\beta
     +(2\underline{\eta}+\zeta)\wedge\beta
     -\frac{1}{2}\hat{\underline{\chi}}\wedge\alpha
     \\
     &=\frac{1}{2}(\slashed{\epsilon}\cdot \nabla\slashed{\mathrm{Ric}}_{4\cdot}
     -\slashed{\epsilon}\cdot(\chi\times\slashed{\mathrm{Ric}}
     +2\underline{\eta}\cdot\slashed{\mathrm{Ric}}_{3\cdot}),
     \end{split}
     \label{Eq2.24}
     \\
     \begin{split}
     &\nabla_3\rho+\frac{1}{2}\mathrm{tr}\underline{\chi}(3\rho+\frac{1}{2}\slashed{\mathrm{Ric}}_{34})+\slashed{\mathrm{div}}\underline{\beta}
     +(2\eta-\zeta)\cdot\underline{\beta}+\frac{1}{2}\hat{\chi}\cdot\underline{\alpha}
     \\
     &=\frac{1}{2}(
     -\nabla_4\slashed{\mathrm{Ric}}_{33}
     -\slashed{\mathrm{div}}\mathrm{Ric}_{3\cdot}
     +\underline{\omega}\slashed{\mathrm{Ric}}_{43}-2\omega\slashed{\mathrm{Ric}}_{33}))
     -\frac{1}{4}\mathrm{tr}\chi\slashed{\mathrm{Ric}}_{33}
     +\frac{1}{2}\underline{\chi}\cdot*\slashed{\mathrm{Ric}}
     -\frac{1}{6}\nabla_3R,
     \end{split}
     \label{Eq2.25}
     \end{align}
     \begin{align}
     \begin{split}
     &\nabla_4\rho+\frac{1}{2}\mathrm{tr}\chi(3\rho+\frac{1}{2}\slashed{\mathrm{Ric}}_{34})
     -\slashed{\mathrm{div}}\beta
     -(2\underline{\eta}+\zeta)\cdot\beta+\frac{1}{2}\hat{\underline{\chi}}\cdot\alpha
     \\
     &=
     \frac{1}{2}(
     -\nabla_3\slashed{\mathrm{Ric}}_{44}
     +\slashed{\mathrm{div}}\mathrm{Ric}_{4\cdot}
     +\omega\slashed{\mathrm{Ric}}_{34}-2\underline{\omega}\slashed{\mathrm{Ric}}_{44})
     -\frac{1}{4}\mathrm{tr}\underline{\chi}\slashed{\mathrm{Ric}}_{44}
     +\frac{1}{2}\chi\cdot*\slashed{\mathrm{Ric}}
     -\frac{1}{6}\nabla_3R,
     \end{split}
     \label{Eq2.26}
     \\
     \begin{split}
     &\nabla_4\underline{\beta}+\mathrm{tr}\chi\underline{\beta}
     +2\omega\underline{\beta}+\nabla\rho+\frac{1}{6}\nabla R
     -*\nabla\sigma+3\underline{\eta}(\rho+\frac{1}{6}\mathrm{Ric}_{34})
     -3*\underline{\eta}\sigma
     -2\hat{\underline{\chi}}\cdot\beta
     \\
     &=\frac{1}{2}
     (-\nabla_\cdot\slashed{\mathrm{Ric}}_{34}+\nabla_3\slashed{\mathrm{Ric}}_{\cdot4}
     +\underline{\chi}\cdot\slashed{\mathrm{Ric}}_{4\cdot}
     +\chi\cdot\slashed{\mathrm{Ric}}_{3\cdot}
     -\eta\slashed{\mathrm{Ric}}_{34}
     -\underline{\omega}\slashed{\mathrm{Ric}}_{4\cdot}-2\eta\cdot\slashed{\mathrm{Ric}}),
     \end{split}
     \label{Eq2.27}
     \\
     \begin{split}
     &\nabla_3\underline{\beta}
     +2\mathrm{tr}\underline{\chi}\underline{\beta}
     +\mathrm{div}\underline{\alpha}
     +2\underline{\omega}\underline{\beta}
     +(\eta-2\zeta)\cdot\underline{\alpha}
     \\
     &=\frac{1}{2}
     (-\nabla_3\slashed{\mathrm{Ric}}_{\cdot3}
     -2\mathrm{tr}\underline{\chi}\slashed{\mathrm{Ric}}_{\cdot3}
     +\nabla_\cdot\slashed{\mathrm{Ric}}_{33}
     -2\underline{\chi}\cdot\slashed{\mathrm{Ric}}_{3\cdot}
     +\underline{\omega}\slashed{\mathrm{Ric}}_{3\cdot}
     +(\eta-2\zeta)\slashed{\mathrm{Ric}}_{33}),
     \end{split}
     \label{Eq2.28}
     \\
     \begin{split}
     &\nabla_3\beta
     +\mathrm{tr}\underline{\chi}\beta
     -2\underline{\omega}\beta
     -\nabla\rho-\frac{1}{6}\nabla R-*\nabla\sigma-2\hat{\chi}\cdot\underline{\beta}
     -3(\eta(\rho+\frac{1}{6}\mathrm{Ric}_{34})+*\eta\sigma)
     \\
     &=\frac{1}{2}
     (\nabla_\cdot\slashed{\mathrm{Ric}}_{34}
     -\nabla_4\slashed{\mathrm{Ric}}_{\cdot3}
     -\underline{\chi}\cdot\slashed{\mathrm{Ric}}_{4\cdot}
     -\chi\cdot\slashed{\mathrm{Ric}}_{3\cdot}
     -\underline{\eta}\slashed{\mathrm{Ric}}_{34}
     -\omega\slashed{\mathrm{Ric}}_{3\cdot}+2\underline{\eta}\slashed{\mathrm{Ric}}),
     \end{split}
     \label{Eq2.29}
     \\
     \begin{split}
     &\nabla_4\beta
     +2\mathrm{tr}\chi\beta
     -\mathrm{div}\alpha
     +2\omega\beta
     -(\underline{\eta}+2\zeta)\cdot\alpha
     \\
     &=-\frac{1}{2}
     (-\nabla_4\slashed{\mathrm{Ric}}_{\cdot4}
     -2\mathrm{tr}\chi \slashed{\mathrm{Ric}}_{\cdot4}
     +\nabla\slashed{\mathrm{Ric}}_{44}
     -2\chi\cdot\slashed{\mathrm{Ric}}_{4\cdot}
     +\omega\slashed{\mathrm{Ric}}_{4\cdot}
     -(\underline{\eta}+2\zeta)\slashed{\mathrm{Ric}}_{44}),
     \end{split}
     \label{Eq2.30}
   \\
     \begin{split}
     &\nabla_3\alpha+\frac{1}{2}\mathrm{tr}\underline{\chi}\alpha
     =\nabla\hat{\otimes}\beta+4\underline{\omega}\alpha
     -3(\hat{\chi}(\rho+\frac{1}{6}\mathrm{Ric}_{34})+*\hat{\chi}\sigma)+(\zeta+4\eta)\hat{\otimes}\beta
     \\
     &-\nabla_4\slashed{\mathrm{Ric}}
     -\frac{1}{2}\chi\slashed{\mathrm{Ric}}_{34}-\frac{1}{2}\underline{\chi}\slashed{\mathrm{Ric}}_{44}
     -\slashed{\mathrm{Ric}}\times\chi
     +\slashed{\mathrm{Ric}}_{4\cdot}\otimes\underline{\eta}
     \\
     &-\frac{1}{4}
     (\nabla_3\slashed{\mathrm{Ric}}_{44}-\nabla_4\slashed{\mathrm{Ric}}_{34}
     -(4\eta-\zeta)\cdot\slashed{\mathrm{Ric}}_{4\cdot}
     +2\underline{\eta}\slashed{\mathrm{Ric}}_{4\cdot}
     +\omega\slashed{\mathrm{Ric}}_{34}
     +2\underline{\omega}\slashed{\mathrm{Ric}}_{44})\slashed{\gamma}
     \\
     &+\frac{1}{2}(\slashed{\epsilon}\cdot\nabla\slashed{\mathrm{Ric}}_{4\cdot}
     -\slashed{\epsilon}\cdot(\chi\times\slashed{\mathrm{Ric}})
     -\slashed{\epsilon}\cdot(\underline{\eta}\otimes\slashed{\mathrm{Ric}}_{4\cdot})
     )\slashed{\epsilon},
     \end{split}
     \label{Eq2.31}
     \\
     \begin{split}
     &\nabla_4\underline{\alpha}+\frac{1}{2}\mathrm{tr}\chi\underline{\alpha}
     =-\nabla\hat{\otimes}\underline{\beta}+4\omega\underline{\alpha}
     -3(\hat{\underline{\chi}}(\rho+\frac{1}{6}\mathrm{Ric}_{34})-*\hat{\underline{\chi}}\sigma)
     +(\zeta-4\underline{\eta})\hat{\otimes}\underline{\beta}
     \\
     &-\nabla_3\slashed{\mathrm{Ric}}
     -\frac{1}{2}\chi\slashed{\mathrm{Ric}}_{33}-\frac{1}{2}\underline{\chi}\slashed{\mathrm{Ric}}_{34}
     -\slashed{\mathrm{Ric}}\times\underline{\chi}
     +\slashed{\mathrm{Ric}}_{3\cdot}\otimes\eta
     \\
     &-\frac{1}{4}
     (\nabla_4\slashed{\mathrm{Ric}}_{33}-\nabla_3\slashed{\mathrm{Ric}}_{34}
     -(4\underline{\eta}-\zeta)\cdot\slashed{\mathrm{Ric}}_{3\cdot}
     +2\eta\slashed{\mathrm{Ric}}_{3\cdot}
     +\underline{\omega}\slashed{\mathrm{Ric}}_{34}
     +2\omega\slashed{\mathrm{Ric}}_{44})\slashed{\gamma}
     \\
     &+\frac{1}{2}(\slashed{\epsilon}\cdot\nabla\slashed{\mathrm{Ric}}_{3\cdot}
     -\slashed{\epsilon}\cdot(\underline{\chi}\times\slashed{\mathrm{Ric}})
     -\slashed{\epsilon}\cdot(\eta\otimes\slashed{\mathrm{Ric}}_{3\cdot})
     )\slashed{\epsilon}.
     \end{split}
     \label{Eq2.32}
   \end{align}
   As in \cite{luk2017weak}, modify the equations by considering the system in
   \begin{align}
     &K=-\rho+\frac{1}{2}\hat{\chi}\cdot\hat{\underline{\chi}}-\frac{1}{4}\mathrm{tr}\chi\mathrm{tr}\underline{\chi}
     +\frac{1}{3}R+\frac{1}{2}\slashed{\mathrm{Ric}}_{34},
     \label{Eq2.33}
     \\
     &\check{\sigma}=\sigma+\frac{1}{2}\hat{\underline{\chi}}\wedge\hat{\chi}
     \label{Eq2.34}.
   \end{align}
   Then the Bianchi equations are free of $\alpha,\underline{\alpha}$, as in the following equations:
   \begin{align}
   \begin{split}
     \nabla_3\beta+\mathrm{tr}\underline{\chi}\beta
     &=-\nabla K+*\nabla\check{\sigma}+2\underline{\omega}\beta+2\hat{\chi}\cdot\underline{\beta}
     -3(\eta K-*\eta\check{\sigma})
     +\frac{1}{2}(\nabla(\hat{\chi}\cdot\underline{\hat{\chi}})+*\nabla(\hat{\chi}\wedge\underline{\hat{\chi}}))
     \\
     &+\frac{3}{2}(\eta\hat{\chi}\cdot\underline{\hat{\chi}}+*\eta\hat{\chi}\wedge\hat{\underline{\chi}})
     -\frac{1}{4}(\nabla\mathrm{tr}\chi\mathrm{tr}\underline{\chi}
     +\mathrm{tr}\chi\nabla\mathrm{tr}\underline{\chi})
     -\frac{3}{4}\eta\mathrm{tr}\chi\mathrm{tr}\underline{\chi}
     \\
     &+\frac{1}{2}
     (\nabla\slashed{\mathrm{Ric}}_{34}-\nabla_4\slashed{\mathrm{Ric}}_{\cdot3}
     -\underline{\chi}\cdot\slashed{\mathrm{Ric}}_{4\cdot}
     -\chi\cdot\slashed{\mathrm{Ric}}_{3\cdot}
     -\underline{\eta}\slashed{\mathrm{Ric}}_{34}
     -\omega\slashed{\mathrm{Ric}}_{3\cdot}+2\underline{\eta}\slashed{\mathrm{Ric}})
     \\
     &+\frac{1}{2}\nabla\slashed{\mathrm{Ric}}_{34}
     +\frac{3}{2}\eta\slashed{\mathrm{Ric}}_{34},
     \end{split}
     \label{Eq2.35}
     \\
     \begin{split}
     \nabla_4\check{\sigma}+\frac{3}{2}\mathrm{tr}\chi\check{\sigma}
     &=-\mathrm{div}*\beta-\zeta\wedge\beta-2\underline{\eta}\wedge\beta
     -\frac{1}{2}\hat{\chi}\wedge(\nabla\hat{\otimes}\underline{\eta})
     -\frac{1}{2}\hat{\chi}\wedge(\underline{\eta}\hat{\otimes}\underline{\eta})
     \\
     &+\frac{1}{2}(\slashed{\epsilon}\cdot \nabla\slashed{\mathrm{Ric}}_{4\cdot}
     -\slashed{\epsilon}\cdot(\chi\times\slashed{\mathrm{Ric}}
     +2\underline{\eta}\cdot\slashed{\mathrm{Ric}}_{3\cdot}+\hat{\chi}\wedge\widehat{\slashed{\mathrm{Ric}}}),
     \end{split}
     \label{Eq2.36}
     \\
     \begin{split}
     \nabla_4K+\mathrm{tr}\chi K
     =&-\mathrm{div}\beta-\zeta\cdot\beta-2\underline{\eta}\cdot\beta
     +\frac{1}{2}\hat{\chi}\cdot\nabla\hat{\otimes}\underline{\eta}
     +\frac{1}{2}\hat{\chi}\cdot(\underline{\eta}\hat{\otimes}\underline{\eta})
     -\frac{1}{2}\mathrm{tr}\chi\mathrm{div}\underline{\eta}-\frac{1}{2}\mathrm{tr}\chi|\underline{\eta}|^2
     \\
     &
     -\frac{1}{2}(-\nabla_4\slashed{\mathrm{Ric}}_{34}+D_4\mathrm{Ric}_{34}-D_3\mathrm{Ric}_{44}
     +\slashed{\mathrm{div}}\slashed{\mathrm{Ric}}_{4\cdot})
     -(2\underline{\eta}+\zeta)\frac{1}{2}\slashed{\mathrm{Ric}}_{4\cdot}
     \\
     &
     +D_4(\frac{1}{3}R+\frac{1}{2}\slashed{\mathrm{Ric}}_{34})
     +\mathrm{tr}\chi(\frac{1}{3}R+\frac{1}{2}\slashed{\mathrm{Ric}}_{34})
     +\frac{1}{2}\hat{\chi}\cdot\widehat{\mathrm{Ric}}
     +\frac{1}{4}\mathrm{tr}\underline{\chi}\mathrm{Ric}_{44}
     \\
     &+\frac{1}{2}\mathrm{tr}\chi(\frac{1}{2}R+\frac{1}{2}\slashed{\mathrm{Ric}}_{34})
     -\frac{1}{2}\chi\cdot*\slashed{\mathrm{Ric}}+\frac{1}{6}\nabla_4R,
     \end{split}
     \label{Eq2.37}
     \\
     \begin{split}
     \nabla_4\underline{\beta}+\mathrm{tr}\chi\underline{\beta}
     &=\nabla K+*\nabla\check{\sigma}+2\omega\underline{\beta}+2\hat{\underline{\chi}}\cdot\beta
     +3(\underline{\eta} K+*\underline{\eta}\check{\sigma})
     -\frac{1}{2}(\nabla(\hat{\chi}\cdot\underline{\hat{\chi}})-*\nabla(\hat{\chi}\wedge\underline{\hat{\chi}}))
     \\
     &-\frac{3}{2}(\eta\hat{\chi}\cdot\underline{\hat{\chi}}-*\eta\hat{\chi}\wedge\hat{\underline{\chi}})
     +\frac{1}{4}(\nabla\mathrm{tr}\chi\mathrm{tr}\underline{\chi}
     +\mathrm{tr}\chi\nabla\mathrm{tr}\underline{\chi})
     +\frac{3}{4}\underline{\eta}\mathrm{tr}\chi\mathrm{tr}\underline{\chi}
     \\
     &
     +\frac{1}{2}
     (\nabla\slashed{\mathrm{Ric}}_{43}-\nabla_3\slashed{\mathrm{Ric}}_{\cdot4}
     -\underline{\chi}\cdot\slashed{\mathrm{Ric}}_{4\cdot}
     -\chi\cdot\slashed{\mathrm{Ric}}_{3\cdot}
     +\eta\slashed{\mathrm{Ric}}_{34}
     -\underline{\omega}\slashed{\mathrm{Ric}}_{4\cdot}+2\eta\slashed{\mathrm{Ric}}),
     \end{split}
     \label{Eq2.38}
     \\
     \begin{split}
     \nabla_3\check{\sigma}+\frac{3}{2}\mathrm{tr}\underline{\chi}\check{\sigma}
     =&-\mathrm{div}*\underline{\beta}+\zeta\wedge\underline{\beta}-2\eta\wedge\underline{\beta}
     +\frac{1}{2}\hat{\underline{\chi}}\wedge(\nabla\hat{\otimes}\eta)
     +\frac{1}{2}\hat{\chi}\wedge(\underline{\eta}\hat{\otimes}\underline{\eta})
     \\
     &+\frac{1}{2}(\slashed{\epsilon}\cdot \nabla\slashed{\mathrm{Ric}}_{3\cdot}
     -\slashed{\epsilon}\cdot(\underline{\chi}\times\slashed{\mathrm{Ric}}
     +2\eta\cdot\slashed{\mathrm{Ric}}_{3\cdot}+\hat{\chi}\wedge\widehat{\slashed{\mathrm{Ric}}}),
     \end{split}
     \label{Eq2.39}
     \\
     \begin{split}
     \nabla_3K+\mathrm{tr}\underline{\chi}
     K
     =&\mathrm{div}\underline{\beta}-\zeta\cdot\underline{\beta}+2\eta\cdot\underline{\beta}
     +\frac{1}{2}\hat{\underline{\chi}}\cdot\nabla\hat{\otimes}\underline{\eta}
     +\frac{1}{2}\hat{\underline{\chi}}\cdot(\eta\hat{\otimes}\eta)
     -\frac{1}{2}\mathrm{tr}\underline{\chi}\mathrm{div}\eta-\frac{1}{2}\mathrm{tr}\underline{\chi}|\eta|^2
     \\
     &
     -\frac{1}{2}(-\nabla_3\slashed{\mathrm{Ric}}_{34}+D_3\slashed{\mathrm{Ric}}_{34}-D_4\slashed{\mathrm{Ric}}_{33})
     +\slashed{\mathrm{div}}\slashed{\mathrm{Ric}}_{3\cdot}
     +(2\eta-\zeta)\frac{1}{2}\slashed{\mathrm{Ric}}_{3\cdot}
     \\
     &+D_3
     (\frac{1}{3}R+\frac{1}{2}\slashed{\mathrm{Ric}}_{34})
     +\mathrm{tr}\underline{\chi}(\frac{1}{3}R+\frac{1}{2}\slashed{\mathrm{Ric}}_{34})
     +\frac{1}{2}\hat{\underline{\chi}}\cdot\widehat{\mathrm{Ric}}
     +\frac{1}{4}
     \mathrm{tr}\chi\mathrm{Ric}_{33}
     \\
     &+\frac{1}{2}\mathrm{tr}\underline{\chi}(\frac{1}{2}R+\frac{1}{2}\slashed{\mathrm{Ric}}_{34})
     -\frac{1}{2}\underline{\chi}\cdot*\slashed{\mathrm{Ric}}+\frac{1}{6}\nabla_3R.
     \end{split}
     \label{Eq2.40}
   \end{align}
   We remark that the Ricci and scalar curvature terms in \eqref{Eq2.37} can be reduced by contracted second Bianchi identity,
   \begin{align*}
     -\frac{1}{2}D_4\mathrm{Ric}_{34}-\frac{1}{2}D_3\mathrm{Ric}_{44}+g^{AB}D_A\mathrm{Ric}_{B4}
     =&\frac{1}{2} D_4R,
   \end{align*}
   to the schematic form indicated in \eqref{Eq3.23} as
   \begin{align*}
     &-\frac{1}{2}(-\nabla_4\slashed{\mathrm{Ric}}_{34}+D_4\mathrm{Ric}_{34}-D_3\mathrm{Ric}_{44}
     +  \slashed{\mathrm{div}}\slashed{\mathrm{Ric}}_{4\cdot})-\frac{1}{2}\chi\cdot*\slashed{\mathrm{Ric}}
     \\
     &
     -\frac{3}{4}\mathrm{tr}\chi\slashed{\mathrm{Ric}}_{34}
     +\frac{1}{2}\mathrm{tr}\chi(R+\frac{1}{2}\slashed{\mathrm{Ric}}_{34})
     \\
     &
     +D_4(\frac{1}{2}R+\frac{1}{2}\slashed{\mathrm{Ric}}_{34})
     +\mathrm{tr}\chi(\frac{1}{2}R+\frac{1}{2}\slashed{\mathrm{Ric}}_{34})
     +\frac{1}{2}\hat{\chi}\cdot\widehat{\mathrm{Ric}}
     +\frac{1}{2}\mathrm{tr}\underline{\chi}\mathrm{Ric}_{44}
     -(2\underline{\eta}+\zeta)\frac{1}{2}\slashed{\mathrm{Ric}}_{4\cdot}
     \\
     =&
     2\underline{\eta}^A\mathrm{Ric}_{4A}
     +\frac{1}{2}\slashed{\mathrm{div}}\slashed{\mathrm{Ric}}_{4\cdot}
     +\zeta\cdot\slashed{\mathrm{Ric}}_{4\cdot}
     -g^{AB}\chi_{A}{}^C\slashed{\mathrm{Ric}}_{CB}
     -\frac{1}{2}\mathrm{tr}\underline{\chi}\mathrm{Ric}_{44}
     -\frac{1}{2}\mathrm{tr}\chi\mathrm{Ric}_{34}
     \\
     &
     -\frac{1}{2}\chi\cdot*\slashed{\mathrm{Ric}}
     -\frac{3}{4}\mathrm{tr}\chi\slashed{\mathrm{Ric}}_{34}
     +\frac{1}{2}\mathrm{tr}\chi(\frac{1}{2}R+\frac{1}{2}\slashed{\mathrm{Ric}}_{34})
     \\
     &
     +\mathrm{tr}\chi(\frac{1}{2}R+\frac{1}{2}\slashed{\mathrm{Ric}}_{34})
     +\frac{1}{2}\hat{\chi}\cdot\widehat{\mathrm{Ric}}
     +\frac{1}{2}\mathrm{tr}\underline{\chi}\mathrm{Ric}_{44}
     -(2\underline{\eta}+\zeta)\frac{1}{2}\slashed{\mathrm{Ric}}_{4A}
     \\
     =&
     (\eta^A+\underline{\eta}^A)\mathrm{Ric}_{4A}
     +\frac{1}{2}\slashed{\mathrm{div}}\slashed{\mathrm{Ric}}_{4\cdot}\ .
   \end{align*}
   Following \cite[(5.28)]{Christodoulou:2008nj}, we can also derive the above result from Codazzi equations as
   \begin{align}
     \begin{split}
     \mathcal{L}_4K+\mathrm{tr}\chi K
     =&\Omega^{-2}(\slashed{\mathrm{div}}\slashed{\mathrm{div}}(\Omega^2\hat{\chi})
     -\frac{1}{2}\slashed{\Delta}(\Omega^2\mathrm{tr}\chi))
     \\
     =&\Omega^{-2}
     (\slashed{\mathrm{div}}((\slashed{\nabla}\Omega^2-\zeta\Omega^2)\cdot\hat{\chi}
     -\frac{1}{2}\slashed{\nabla}\Omega^2 \mathrm{tr}\chi
     +\Omega^2(\frac{1}{2}\mathrm{tr}\chi\zeta-\beta+\frac{1}{2}\slashed{\mathrm{Ric}}_{4\cdot}))).
     \end{split}
     \label{Eq2.42}
   \end{align}
   Concluding the arguments above, we arrive at
   \begin{prop}\label{PropA.1}
     The Ricci curvature terms in \eqref{Eq2.36} for $e_4\check{\sigma}$, \eqref{Eq2.30} for $\nabla_4\beta$, and $\eqref{Eq2.37}$ for $e_4K$ reads, schematically, as
   \begin{align*}
     & (\slashed{\nabla},\nabla_4) \slashed{\mathrm{Ric}}+(\psi+\psi_H)\slashed{\mathrm{Ric}}.
   \end{align*}
   \end{prop}
   \begin{proof}
     We note that for the projected tensors $\slashed{\mathrm{Ric}},\slashed{\mathrm{Ric}}_{\alpha\cdot},\slashed{\mathrm{Ric}}_{\alpha\beta}$ for $\alpha,\beta\in\{3,4\}$,
     \begin{align*}
       &(D_A,D_4)\mathrm{Ric}
       =_s(\slashed{\nabla},\nabla_4) \slashed{\mathrm{Ric}}+(\psi+\psi_{H})\slashed{\mathrm{Ric}}.
     \end{align*}
     Thus the proof is complete by noting that the singular terms $\underline{\omega},\underline{\chi}$ are not present.
   \end{proof}
   Combining the equations above,
   \begin{align}
   \begin{split}
     \nabla_4(\slashed{\mathrm{div}}&\eta
     -K)
     =\slashed{\mathrm{div}}(
     -\frac{1}{2}\slashed{\mathrm{tr}}\chi\underline{\eta}
     -\hat{\chi}(\eta-\underline{\eta})+\frac{1}{2}\slashed{\mathrm{tr}}\chi\underline{\eta}-\slashed{\mathrm{Ric}}_{4\cdot} )
     -\mathrm{tr}\chi K
     \\
     &
     +\Omega^{-1}
     (\slashed{\mathrm{div}}((\slashed{\nabla}\Omega-\zeta\Omega)\cdot\hat{\chi}
     -\frac{1}{2}\slashed{\nabla}\Omega^2\mathrm{tr}\chi
     +\Omega(\frac{1}{2}\mathrm{tr}\chi\zeta-\beta+\frac{1}{2}\slashed{\mathrm{Ric}}_{4\cdot})))
     +[\nabla_4,\mathrm{div}]\eta.
     \end{split}
     \label{Eq2.43}
   \end{align}
   For the elliptic system in $\eta$, we also have
   \begin{align}
     \mathrm{curl}\eta=&\frac{1}{2}\chi\wedge\underline{\chi}
     +\sigma=\check{\sigma}.
     \label{Eq2.41}
   \end{align}
   And similarly,
   \begin{align}
     \begin{split}
     &-\mathrm{curl}\underline{\eta}=\frac{1}{2}\chi\wedge\underline{\chi}
     +\sigma=\check{\sigma},
     \end{split}
     \label{Eq2.44}
     \\
     \begin{split}
     \mathcal{L}_3K+\mathrm{tr}\underline{\chi} K
     =&\Omega^{-2}(\slashed{\mathrm{div}}\slashed{\mathrm{div}}(\Omega^2\hat{\underline{\chi}})
     -\frac{1}{2}\slashed{\Delta}(\Omega^2\mathrm{tr}\underline{\chi}))
     \\
     =&\Omega^{-2}
     (\slashed{\mathrm{div}}((\slashed{\nabla}\Omega^2-\zeta\Omega^2)\cdot\hat{\underline{\chi}}
     -\frac{1}{2}\slashed{\nabla}\Omega^2\mathrm{tr}\underline{\chi}
     +\Omega^2(\frac{1}{2}\mathrm{tr}\underline{\chi}\underline{\zeta}-\underline{\beta}+\frac{1}{2}\slashed{\mathrm{Ric}}_{3\cdot}))).
     \end{split}
     \label{Eq2.45}
   \end{align}
   Thus,
   \begin{align}
   \begin{split}
     \nabla_3(\slashed{\mathrm{div}}&\underline{\eta}
     -K)
     =\slashed{\mathrm{div}}(
     -\frac{1}{2}\slashed{\mathrm{tr}}\underline{\chi}\eta
     +\hat{\underline{\chi}}(\eta-\underline{\eta})+\frac{1}{2}\slashed{\mathrm{tr}}\underline{\chi}\eta-\slashed{\mathrm{Ric}}_{3\cdot} )
     -\mathrm{tr}\underline{\chi} K
     \\
     &
     +\Omega^{-1}
     (\slashed{\mathrm{div}}((\slashed{\nabla}\Omega-\zeta\Omega)\cdot\hat{\underline{\chi}}
     -\frac{1}{2}\slashed{\nabla}\Omega\mathrm{tr}\underline{\chi}
     +\Omega(\frac{1}{2}\mathrm{tr}\underline{\chi}\underline{\zeta}-\underline{\beta}-\frac{1}{2}\slashed{\mathrm{Ric}}_{3\cdot})))
     +[\nabla_4,\mathrm{div}]\eta.
     \end{split}
     \label{Eq2.46}
   \end{align}
   Also, $\omega,\underline{\omega}$ satisfy a elliptic system:
   \begin{align}
     &\mathrm{div}\nabla\omega
     =\mathrm{div}\nabla_4(\eta+\underline{\eta}),
     \label{Eq2.47}
     \\
     &\mathrm{curl}\nabla\omega=0,
     \label{Eq2.48}
     \\
     &\mathrm{div}\nabla\underline{\omega}
     =\mathrm{div}\nabla_3(\eta+\underline{\eta}),
     \label{Eq2.49}
     \\
     &\mathrm{curl}\nabla\underline{\omega}=0.
     \label{Eq2.50}
   \end{align}
   For the fluid part, projecting along flows by $v$ implies
   \begin{align}
       &v_\alpha D^\alpha\tau+(p+\tau)\mathrm{div}v=0,
       \label{Eq2.51}
       \\
       &(p+\tau)v_\alpha D^\alpha v_\beta-(p+\rho)v_\alpha \Gamma^{\alpha}_{\beta\gamma}v^\gamma
       +p'(\tau)v_\alpha D^\alpha\tau v_\beta+D_\beta \tau)=0.
       \label{Eq2.52}
   \end{align}

   \bibliography{./Einstein+Fluid_version_2}

@article{speck2012nonlinear,
  title = {The nonlinear future stability of the {FLRW} family of solutions to the {E}uler–{E}instein system with a positive cosmological constant},
  author = {Speck, Jared},
  journal = {Selecta Mathematica},
  volume = {18},
  issue = {3},
  pages = {633--715},
  year = {2012},
  month = {09},
}

@article{luk2017weak,
  title = {Weak null singularities in general relativity},
  author = {Jonathan Luk},
  journal = {Journal of the American Mathematical Society},
  volume = {31},
  issue = {1},
  year = {2017},
  month = {09},
}

@article{Disconzi_2019,
   title={The Relativistic {E}uler Equations: Remarkable Null Structures and Regularity Properties},
   volume={20},
   ISSN={1424-0661},
   url={http://dx.doi.org/10.1007/s00023-019-00801-7},
   DOI={10.1007/s00023-019-00801-7},
   number={7},
   journal={Annales Henri Poincaré},
   publisher={Springer Science and Business Media LLC},
   author={Disconzi, Marcelo M. and Speck, Jared},
   year={2019},
   month=may, pages={2173–2270} }

@book{Christodoulou1989-1990,
author = {Christodoulou, D. and Klainerman, S.},
series = {Princeton Mathematical Series},
publisher = {Princeton University Press},
volume = {41},
title = {The global nonlinear stability of the {M}inkowski space},
year = {1993},
}

@article{luk2019strong,
      title={Strong cosmic censorship in spherical symmetry for two-ended asymptotically flat initial data {I}. The interior of the black hole region},
      author={Jonathan Luk and Sung-Jin Oh},
      journal = {Annals of Mathematics},
  volume = {190 },
  issue = {1},
  pages = {1--111},
  year = {2019},
}

@article{Luk2019,
      title={Strong Cosmic Censorship in Spherical Symmetry for Two-Ended Asymptotically Flat Initial Data {II}: The Exterior of the Black Hole Region},
      author={Jonathan Luk and Sung-Jin Oh},
      journal = {Annals of PDE},
  volume = {5},
  issue = {1},
  year = {2019},
}

@misc{dafermos2017interiordynamicalvacuumblack,
      title={The interior of dynamical vacuum black holes {I}: The ${C}^0$-stability of the {K}err {C}auchy horizon},
      author={Mihalis Dafermos and Jonathan Luk},
      year={2017},
      note={arXiv:gr-qc/1710.01722},
      url={https://arxiv.org/abs/1710.01722},
}

@misc{disconzi2023recent,
      title={Recent developments in mathematical aspects of relativistic fluids},
      author={Marcelo M. Disconzi},
      year={2023},
      note={arXiv:math.AP/2308.09844},
}

@article{Taylor2017,
  title = {The Global Nonlinear Stability of {M}inkowski Space for the Massless {E}instein–{V}lasov System},
  author = {Taylor, Martin},
  journal = {Annals of PDE},
  volume = {3},
  issue = {1},
  year = {2017},
  month = {03},
}

@article{HISCOCK1981110,
title = {Evolution of the interior of a charged black hole},
journal = {Physics Letters A},
volume = {83},
number = {3},
pages = {110-112},
year = {1981},
issn = {0375-9601},
doi = {https://doi.org/10.1016/0375-9601(81)90508-9},
url = {https://www.sciencedirect.com/science/article/pii/0375960181905089},
author = {William A. Hiscock}
}

@article{PhysRevLett.63.1663,
  title = {Inner-horizon instability and mass inflation in black holes},
  author = {Poisson, E. and Israel, W.},
  journal = {Phys. Rev. Lett.},
  volume = {63},
  issue = {16},
  pages = {1663--1666},
  numpages = {0},
  year = {1989},
  month = {Oct},
  publisher = {American Physical Society},
  doi = {10.1103/PhysRevLett.63.1663},
  url = {https://link.aps.org/doi/10.1103/PhysRevLett.63.1663}
}

@article{PhysRevD.41.1796,
  title = {Internal structure of black holes},
  author = {Poisson, Eric and Israel, Werner},
  journal = {Phys. Rev. D},
  volume = {41},
  issue = {6},
  pages = {1796--1809},
  numpages = {0},
  year = {1990},
  month = {Mar},
  publisher = {American Physical Society},
  doi = {10.1103/PhysRevD.41.1796},
  url = {https://link.aps.org/doi/10.1103/PhysRevD.41.1796}
}

@article{Ori1995HowGA,
  title={How generic are null spacetime singularities?},
  author={Omos Ori and Eanna Flanagan},
  journal={Physical review. D, Particles and fields},
  year={1995},
  volume={53 4},
  pages={
          R1754-R1758
        },
  url={https://api.semanticscholar.org/CorpusID:21534485}
}

@article{Speck2017,
  title={Shock Formation for 2{D} Quasilinear Wave Systems Featuring Multiple Speeds: Blowup for the Fastest Wave, with Non-trivial Interactions up to the Singularity},
  author={Speck, Jared},
  journal={Annals of PDE},
  year={2017},
  volume={4},
  issue = {1},
  url={https://doi.org/10.1007/s40818-017-0042-8}
}

@article{Luk2013NonlinearIO,
  title={Nonlinear interaction of impulsive gravitational waves for the vacuum {E}instein equations},
  author={Jonathan Luk and Igor Rodnianski},
  journal={},
  year={2017},
      note={arXiv:gr-qc/1301.1072},
      url={https://arxiv.org/abs/1301.1072},
}

@book{Christodoulou:2008nj,
    author = "Christodoulou, Demetrios",
    title = "{The Formation of Black Holes in General Relativity}",
    eprint = "0805.3880",
    archivePrefix = "arXiv",
    primaryClass = "gr-qc",
    year = "2008",
    note={arXiv:gr-qc/0805.3880},
}

@article{Yu2011DynamicalFO,
  title={Dynamical formation of black holes due to the condensation of matter field},
  author={Pin Yu},
  journal={},
  year={2011},
  note={arXiv:math.AP/1105.5898},
      url={https://arxiv.org/abs/1105.5898},
}

@article{Luk2018,
author = {Luk, Jonathan and Speck, Jared},
year = {2010},
month = {10},
pages = {1-169},
title = {Shock formation in solutions to the 2{D} compressible {E}uler equations in the presence of non-zero vorticity},
volume = {214},
journal = {Inventiones mathematicae},
doi = {10.1007/s00222-018-0799-8}
}

@article{FritzJohnFormationofSingularities,
author = {John, Fritz},
title = {Formation of singularities in one-dimensional nonlinear wave propagation},
journal = {Communications on Pure and Applied Mathematics},
volume = {27},
number = {3},
pages = {377-405},
doi = {https://doi.org/10.1002/cpa.3160270307},
url = {https://onlinelibrary.wiley.com/doi/abs/10.1002/cpa.3160270307},
eprint = {https://onlinelibrary.wiley.com/doi/pdf/10.1002/cpa.3160270307},
year = {1974}
}

@article{DafermosBlackHole,
author = {Dafermos, Mihalis and Rodnianski, Igor},
year = {2013},
title = {Lectures on black holes and linear waves},
volume = {17},
pages = {97-205},
journal = {Evolution equations, Clay Mathematics Proceedings},
}

@article{417492a798bd4c2abf778c47bb3969dc,
title = "On the formation of trapped surfaces",
author = "Sergiu Klainerman and Igor Rodnianski",
year = "2012",
month = jun,
doi = "10.1007/s11511-012-0077-3",
language = "English (US)",
volume = "208",
pages = "211--333",
journal = "Acta Mathematica",
issn = "0001-5962",
publisher = "Springer Netherlands",
number = "2",
}

@article{https://doi.org/10.1002/cpa.21531,
author = {Luk, Jonathan and Rodnianski, Igor},
title = {Local Propagation of Impulsive GravitationalWaves},
journal = {Communications on Pure and Applied Mathematics},
volume = {68},
number = {4},
pages = {511-624},
doi = {https://doi.org/10.1002/cpa.21531},
url = {https://onlinelibrary.wiley.com/doi/abs/10.1002/cpa.21531},
eprint = {https://onlinelibrary.wiley.com/doi/pdf/10.1002/cpa.21531},
year = {2015}
}

@article{https://doi.org/10.1002/cpa.20071,
author = {Dafermos, Mihalis},
title = {The interior of charged black holes and the problem of uniqueness in general relativity},
journal = {Communications on Pure and Applied Mathematics},
volume = {58},
number = {4},
pages = {445-504},
doi = {https://doi.org/10.1002/cpa.20071},
url = {https://onlinelibrary.wiley.com/doi/abs/10.1002/cpa.20071},
eprint = {https://onlinelibrary.wiley.com/doi/pdf/10.1002/cpa.20071},
year = {2005}
}

@article{a965bffe-de90-3a7a-a652-4c0d2147a13b,
 ISSN = {0003486X},
 URL = {http://www.jstor.org/stable/3597235},
 author = {Mihalis Dafermos},
 journal = {Annals of Mathematics},
 number = {3},
 pages = {875--928},
 publisher = {Annals of Mathematics},
 title = {Stability and Instability of the {C}auchy Horizon for the Spherically Symmetric {E}instein-{M}axwell-Scalar Field Equations},
 urldate = {2024-06-22},
 volume = {158},
 year = {2003}
}

@article{Moortel2017StabilityAI,
  title={Stability and Instability of the Sub-extremal {R}eissner–{N}ordstr{\"o}m Black Hole Interior for the {E}instein–{M}axwell–{K}lein–{G}ordon Equations in Spherical Symmetry},
  author={{Van de Moortel}, Maxime},
  journal={Communications in Mathematical Physics},
  year={2017},
  volume={360},
  pages={103-168},
  url={https://api.semanticscholar.org/CorpusID:119148710}
}

@article{93760b6492084bb69a1843fdee21adfc,
title = "Mass Inflation and the {${C}^2$} -inextendibility of Spherically Symmetric Charged Scalar Field Dynamical Black Holes",
author = "{Van de Moortel}, Maxime",
year = "2021",
month = mar,
doi = "10.1007/s00220-020-03923-w",
language = "English (US)",
volume = "382",
pages = "1263--1341",
journal = "Communications In Mathematical Physics",
issn = "0010-3616",
publisher = "Springer New York",
number = "2",
}

@article{VandeMoortel:2019ike,
    author = "Van de Moortel, Maxime",
    title = "{The breakdown of weak null singularities inside black holes}",
    eprint = "1912.10890",
    archivePrefix = "arXiv",
    primaryClass = "gr-qc",
    doi = "10.1215/00127094-2022-0096",
    journal = "Duke Math. J.",
    volume = "172",
    number = "15",
    pages = "2957--3012",
    year = "2023"
}

@article{doi:10.1098/rspa.1990.0009,
author = {Rendall, A. D. },
title = {Reduction of the characteristic initial value problem to the {C}auchy problem and its applications to the {E}instein equations},
journal = {Proceedings of the Royal Society of London. A. Mathematical and Physical Sciences},
volume = {427},
number = {1872},
pages = {221-239},
year = {1990},
doi = {10.1098/rspa.1990.0009},
}

@article{doi:10.1142/S0217732314502058,
author = {Disconzi, Marcelo M. and Pingali, Vamsi P.},
title = {On the {C}hoquet-{B}ruhat–{Y}ork–{F}riedrich formulation of the {E}instein–{E}uler equations},
journal = {Modern Physics Letters A},
volume = {29},
number = {39},
pages = {145-205},
year = {2014},
doi = {10.1142/S0217732314502058},
}

@article{Sbierski2024,
    author = {Sbierski, Jan},
    title = {Lipschitz inextendibility of weak null singularities from curvature blow-up},
    month = "9",
    year = "2024",
    journal={},
    note={arXiv:gr-qc/2409.18838},
      url={https://arxiv.org/abs/2409.18838},
}

@article{gautam2024latetimetailsmassinflation,
      title={Late-time tails and mass inflation for the spherically symmetric {E}instein-{M}axwell-scalar field system}, 
      author={Onyx Gautam},
      year={2024},
      journal={},
      eprint={2412.17927},
      archivePrefix={arXiv},
      primaryClass={gr-qc},
      note={arXiv:gr-qc/2412.17927},
      url={https://arxiv.org/abs/2412.17927}, 
}

@article{LOSQ,
author = {Luk, Jonathan and Oh, Sung-Jin and Shlapentokh-Rothman, Yakov},
title = {A Scattering Theory Approach to {C}auchy Horizon Instability and Applications to Mass Inflation},
journal = {Annales Henri Poincaré},
volume = {24},
pages = {363-411},
year = {2023},
doi = {10.1007/s00023-022-01216-7},
}

@article{Luk:2015qja,
    author = "Luk, Jonathan and Oh, Sung-Jin",
    title = {{Proof of linear instability of the {R}eissner\textendash{}{N}ordstr\"om {C}auchy horizon under scalar perturbations}},
    eprint = "1501.04598",
    archivePrefix = "arXiv",
    primaryClass = "gr-qc",
    doi = "10.1215/00127094-3715189",
    journal = "Duke Math. J.",
    volume = "166",
    number = "3",
    pages = "437--493",
    year = "2017"
}

@article{SBierSKIJANPROOF,
    author = "Sbierski, Jan",
    title = {Instability of the {K}err {C}auchy Horizon Under Linearised Gravitational Perturbations},
    doi = "10.1007/s40818-023-00146-9",
    journal = "Annals of PDE",
    volume = "9",
    number = "1",
    year = "2023",
}

@article{LUK20161948,
title = {Instability results for the wave equation in the interior of {K}err black holes},
journal = {Journal of Functional Analysis},
volume = {271},
number = {7},
pages = {1948-1995},
year = {2016},
issn = {0022-1236},
doi = {https://doi.org/10.1016/j.jfa.2016.06.013},
url = {https://www.sciencedirect.com/science/article/pii/S0022123616301501},
author = {Jonathan Luk and Jan Sbierski}
}

@article{Dafermos_2016,
   title={Time-Translation Invariance of Scattering Maps and Blue-Shift Instabilities on {K}err Black Hole Spacetimes},
   volume={350},
   ISSN={1432-0916},
   url={http://dx.doi.org/10.1007/s00220-016-2771-z},
   DOI={10.1007/s00220-016-2771-z},
   number={3},
   journal={Communications in Mathematical Physics},
   publisher={Springer Science and Business Media LLC},
   author={Dafermos, Mihalis and Shlapentokh-Rothman, Yakov},
   year={2016},
   month=oct, pages={985–1016} 
   }

@article {MR4657222,
    AUTHOR = {Ma, Siyuan and Zhang, Lin},
     TITLE = {Precise late-time asymptotics of scalar field in the interior
              of a subextreme {K}err black hole and its application in
              strong cosmic censorship conjecture},
   JOURNAL = {Trans. Amer. Math. Soc.},
  FJOURNAL = {Transactions of the American Mathematical Society},
    VOLUME = {376},
      YEAR = {2023},
    NUMBER = {11},
     PAGES = {7815--7856},
      ISSN = {0002-9947,1088-6850},
   MRCLASS = {83C75 (35Q75 58J45 83C57)},
  MRNUMBER = {4657222},
MRREVIEWER = {Wolfgang\ Hasse},
       DOI = {10.1090/tran/8957},
       URL = {https://doi.org/10.1090/tran/8957},
}

@article{Gurriaran_2025,
   title={Precise Asymptotics of the Spin $+2$ {Teukolsky} Field in the {Kerr} Black Hole Interior},
   volume={406},
   ISSN={1432-0916},
   url={http://dx.doi.org/10.1007/s00220-025-05332-3},
   DOI={10.1007/s00220-025-05332-3},
   number={7},
   journal={Communications in Mathematical Physics},
   publisher={Springer Science and Business Media LLC},
   author={Gurriaran, Sebastian},
   year={2025},
   month=jun, }

@article{vandemoortel2025strongcosmiccensorshipconjecture,
      title={The Strong Cosmic Censorship Conjecture}, 
      author={{Van de Moortel},Maxime},
      year={2025},
      eprint={2501.13180},
      archivePrefix={arXiv},
      primaryClass={gr-qc},
      note={arXiv:gr-qc/2501.13180},
      url={https://arxiv.org/abs/2501.13180}, 
}

@article{LocalizedBigBang,
      title={Localized Big Bang Stability for the Einstein-Scalar Field Equations}, 
   volume={248},
   ISSN={1432-0673},
   url={https://doi.org/10.1007/s00205-023-01939-9},
   DOI={10.1007/s00205-023-01939-9},
   issue={1},
   journal={Archive for Rational Mechanics and Analysis},
      author={Beyer, Florian and Oliynyk, Todd A.},
      year={2023}, 
}

@misc{hao2024wellposednessfreeboundarybarotropic,
      title={Well-posedness for the free boundary barotropic fluid model in general relativity}, 
      author={Zeming Hao and Wei Huo and Shuang Miao},
      year={2024},
      eprint={2410.01616},
      archivePrefix={arXiv},
      primaryClass={math.AP},
      url={https://arxiv.org/abs/2410.01616},
      note={arXiv:math.AP/2410.01616} 
}

@article{10.1093/imrn/rnr201,
    author = {Luk, Jonathan},
    title = {On the Local Existence for the Characteristic Initial Value Problem in General Relativity},
    journal = {International Mathematics Research Notices},
    volume = {2012},
    number = {20},
    pages = {4625-4678},
    year = {2011},
    month = {10},
    issn = {1073-7928},
    doi = {10.1093/imrn/rnr201},
    url = {https://doi.org/10.1093/imrn/rnr201},
    eprint = {https://academic.oup.com/imrn/article-pdf/2012/20/4625/19003965/rnr201.pdf},
}

@misc{luk2026formationweaknullsingularity,
      title={The formation of a weak null singularity in the interior of generic rotating black holes}, 
      author={Jonathan Luk and Jan Sbierski},
      year={2026},
      eprint={2604.04877},
      archivePrefix={arXiv},
      primaryClass={gr-qc},
      url={https://arxiv.org/abs/2604.04877},
      note={arXiv:gr-qc/2604.04877},
}

@misc{gurriaran2026nonlinearinstabilitykerrcauchy,
      title={Non-linear instability of the Kerr Cauchy horizon near $i_+$}, 
      author={Sebastian Gurriaran},
      year={2026},
      eprint={2603.17911},
      archivePrefix={arXiv},
      primaryClass={gr-qc},
      url={https://arxiv.org/abs/2603.17911},
      note={arXiv:gr-qc/2603.17911},
}
   \bibliographystyle{plain}
\end{document}